\documentclass[11pt]{imsart}

\usepackage[margin=1in]{geometry}

\RequirePackage{amsthm,amsmath,amsfonts,amssymb}
\RequirePackage[authoryear]{natbib}

\usepackage[utf8]{inputenc} 
\usepackage[T1]{fontenc}    
\usepackage[ruled,vlined]{algorithm2e}
\usepackage{array}
\usepackage{caption}
\usepackage{subcaption}
\usepackage{hyperref}       
\usepackage{url}            
\usepackage{booktabs}       
\usepackage{amsfonts}       
\usepackage{nicefrac}       
\usepackage{microtype}      
\usepackage{xcolor}         
\usepackage{centernot}      
\usepackage{booktabs}
\usepackage{lipsum}
\usepackage{graphicx}
\usepackage{amsmath,amsfonts,amssymb,amsthm}
\usepackage{bm}             
\usepackage{accents}        
\usepackage{cleveref}       
\usepackage{longtable}      
\usepackage{enumitem}
\usepackage{cuted}
\usepackage{afterpage}

\startlocaldefs
\graphicspath{ {./images/} }

\newcommand{\independent}{\perp\mkern-9.5mu\perp}
\newcommand{\notindependent}{\centernot{\independent}}

\usepackage{soul}

\newtheorem{theorem}{Theorem}
\newtheorem{proposition}{Proposition}
\newtheorem{corollary}{Corollary}

\theoremstyle{definition}
\newtheorem{assumption}{Assumption}
\newtheorem{definition}{Definition}
\newtheorem{remark}{Remark}

\allowdisplaybreaks

\endlocaldefs

\begin{document}

\begin{frontmatter}
\title{Almost exact Mendelian randomization}
\runtitle{Almost exact MR}

\begin{aug}
\author[A]{\fnms{Matthew J}~\snm{Tudball}\ead[label=e1]{matt.tudball@bristol.ac.uk}},
\author[A]{\fnms{George}~\snm{Davey Smith}\ead[label=e2]{KZ.Davey-Smith@bristol.ac.uk}}
\and
\author[B]{\fnms{Qingyuan}~\snm{Zhao}\ead[label=e3]{qyzhao@statslab.cam.ac.uk}}
\address[A]{MRC Integrative Epidemiology Unit, University of Bristol\printead[presep={,\ }]{e1,e2}}

\address[B]{Statistical Laboratory, University of Cambridge\printead[presep={,\ }]{e3}}
\end{aug}

\begin{abstract}
  Mendelian randomization (MR) is a natural experimental design based on the
random transmission of genes from parents to
offspring. However, this inferential basis is typically only implicit
or used as an informal justification. As
parent-offspring data becomes more widely available, we advocate a different
approach to MR that is exactly based on this natural randomization,
thereby formalizing the analogy between MR and randomized
controlled trials. We begin by developing a causal graphical model
for MR which represents several biological processes
  and phenomena, including population structure, gamete
formation, fertilization, genetic linkage, and pleiotropy. This causal
graph is then used to detect biases in population-based MR studies and
identify sufficient
confounder adjustment sets to correct these biases. We then propose a
randomization test in the within-family MR design using
the exogenous randomness in meiosis and fertilization, which is extensively
studied in genetics.
Besides its transparency and conceptual appeals, our approach
also offers some practical advantages, including robustness to
misspecified phenotype models, robustness to weak instruments, and
elimination of bias arising from population structure,
assortative mating, dynastic effects, and horizontal pleiotropy. We
conclude with an analysis of a pair of negative and positive
controls in the Avon Longitudinal Study of Parents and
Children. The accompanying \texttt{R} package can be found at
\url{https://github.com/matt-tudball/almostexactmr}.


\end{abstract}



\end{frontmatter}

\section{Introduction} \label{sec:intro}

\subsection{A brief history of Mendelian
  randomization} \label{sec:intro:historyofmr}

Mendelian randomization (MR) is a natural experimental design that uses the
random allocation of genes from parents to offspring as a foundation
for causal inference \citep{Sanderson2022}. The ideas behind MR can be traced
back to the intertwined beginnings of modern statistics and genetics
about a century ago. In one of the earliest examples, \citet{Wright1920} used
selective inbreeding of guinea pigs to investigate the causes of
colour variation and, in particular, the relative contribution of
heredity and environment. In a later defence of this work,
\citet[p.\ 251]{Wright1923} argued that his analysis of path
coefficients, a precursor to modern causal graphical models, ``rests
on the validity of the premises, i.e., on the evidence for Mendelian
heridity'', and the ``universality'' of Mendelian laws justifies
ascribing a causal interpretation to his findings.

At around the same time, \citet{fisher26_arran_field_exper} was contemplating the randomization principle in experimental design and
used it to justify his analysis of variance (ANOVA) procedure, which was
partly motivated by genetic problems. In fact, the term ``variance'' first
appeared in Fisher's groundbreaking paper that bridged Darwin's
theory of evolution and Mendel's theory of genetic inheritance
\citep{fisher19_correlation}. \citet{Fisher1935} described
randomization as the ``reasoned basis'' (p.\ 12) for inference and
``the physical basis of the validity of the
test'' (p.\ 17). In a later Bateson Lecture to a genetics audience,
Fisher revealed that his
factorial design of experiments derives ``its structure and its
name from the simultaneous inheritance of Mendelian factors''
\citep[p. 330]{Fisher1951}. Indeed,
Fisher viewed randomness in meiosis as uniquely shielding geneticists from the
difficulties of establishing reliably controlled comparisons,
remarking that ``the different genotypes possible from the same mating
have been beautifully randomized by the meiotic process''
\citep[p. 332]{Fisher1951}.

While this source of randomization was originally used for eliciting
genetic causes of phenotypic variation, it was later identified as a
possible avenue for understanding causation among modifiable
phenotypes
themselves. \citet{lower79_n_acety_phenot_risk_urinar_bladd_cancer}
used N-acetylation, a phenotype of known genetic regulation and a
component of detoxification pathways for arylamine, to strengthen the
inference that arylamine exposure causes bladder
cancer. \citet{Katan2004a} proposed to address
reverse causation in the hypothesized effect of low serum cholesterol
on cancer risk via polymorphisms in the apolipoprotein E
(\textit{APOE}) gene. He argued that, if low cholesterol was indeed a
risk factor for cancer, we would expect to see higher rates of cancer
in individuals with the low cholesterol allele. Another pioneering
application of this reasoning can be found in a proposed study of the
effectiveness of bone marrow transplantation relative to chemotherapy \citep{Gray1991}, for example, in the treatment of acute myeloid leukaemia \citep{Wheatley2004}. Patients with a compatible
donor sibling were more likely to receive transplantation than
patients without. Since compatibility is a consequence of random
genetic assortment, comparing survival outcomes between the two groups
can be viewed as akin to an intention-to-treat analysis in a
randomized controlled trial. This paper appears to be the first to
use the term ``Mendelian randomization''. An earlier review of this
prehistory of MR can be found in \citet{smith07_capit_mendel_random_to_asses_effec_treat}.

It would be a dozen more years before an argument for the broader applicability of MR was put forward by \citet{DaveySmith2003}. At the time, a number of criticisms had been levelled against the state of observational epidemiology and its methods of inquiry \citep{Feinstein1988, Taubes1995, DaveySmith2001}. Several high profile results failed to be corroborated by subsequent randomized controlled trials, such as the role of beta-carotene consumption in lowering risk of cardiovascular
disease, with unobserved confounding identified as the likely culprit
\citep[p.\ 329-330]{DaveySmith2001}. This string of failures motivated
the development of a more rigorous observational design with an
explicit source of unconfounded randomization in the exposures of
interest; see \citet{DaveySmith2020} for more discussion on the
motivation for MR.

Originally, \citet{DaveySmith2003} recognized that MR is best
justified in a within-family design with parent-offspring trios. MR is
commonly described as being analogous to a randomized controlled trial
with non-compliance. This analogy is based on exact randomization in
the transmission of alleles from parents to offspring which can be
viewed as a form of treatment assignment. From its inception, it was
recognized that data limitations would largely restrict MR to be
performed in samples of unrelated individuals, which
\citet{DaveySmith2003} termed ``approximate MR''. Such approximate MR
has indeed become the norm, seen in the majority of applied studies
and methodological development to date. However, MR in unrelated individuals
lacks the explicit source of randomization offered by the
within-family design, thereby suffering potential biases from dynastic
effects, population structure and assortative mating
\citep{Davies2019, Brumpton2020, Howe2022}.

In addition to random assignment of exposure-modifying genetic
variants, another crucial assumption for MR, known as the exclusion
restriction, is that the causal effects
of these variants on the outcome must be fully mediated by the
exposure. When this assumption holds, MR can be
framed as a special case of instrumental variable analysis
\citep{Thomas2004, Didelez2007a}. Within this framework, there has
been considerable recent methodological work to replace the exclusion
restriction with more plausible assumptions, typically by placing
structure on the distribution of pleiotropic effects across
individual genetic variants
\citep{kang16_instr_variab_estim_with_some, Bowden2015,
  Zhao2020}; see also
\citet{kolesar15_ident_infer_with_many_inval_instr} for a similar idea
in the econometrics literature.

\subsection{Towards an almost exact inference for MR} \label{sec:intro:ourtest}

As parent-offspring trio data become more widely available, it is
increasingly feasible to perform MR within families, as originally
proposed by \citet{DaveySmith2003}. There has been
some recent methodological development for within-family
designs \citep{Davies2019, Brumpton2020}. Thus far this has consisted of
extensions of traditional MR techniques in which structural models for
the gene-exposure and gene-outcome relationships are proposed and
samples are assumed to be drawn according to these models from some
large population. In particular, \citet{Brumpton2020} propose a
linear structural model with parental genotype fixed effects. Their
inference is based on this model and so the role of meiotic
randomization is only implicit.

However, one of the unique advantages of MR as a natural experimental design
is that it has an explicit inferential basis, namely the exogenous
randomness in meiosis and fertilization, which has been thoroughly
studied and modelled in genetics since at least
\citet{Haldane1919}. Haldane developed a simple model for
recombination during meiosis that has demonstrated good performance on multiple
pedigrees across many species. The connection between this meiosis
model and causal inference in parent-offspring trio studies was
recently described in the context of locating causal genetic
variants \citep{Bates2020a} and was implicit in earlier pedigree-based
methods, such as the genetic linkage analysis in
\citet{morton1955sequential} and the transmission disequilibrium test
in \citet{Spielman1993}. \citet{Lauritzen2003}
attempted to represent meiosis using graphical models; however, they
focused on computatational considerations and did not explore the
potential of these models for causal inference.

The idea of the significance test (or hypothesis test, although some
authors distinguish the use of these two terms) dates back to Fisher's
original proposal for randomized experiments and is well illustrated
in his famous `lady tasting tea' example
\citep{Fisher1935}. \citet{pitman37_signif_tests_which_may_be}
appears to be the first to fully embrace the idea of randomization
tests. This mode of reasoning is
usually referred to as randomization inference or design-based
inference to contrast with model-based inference. With the aid of the
potential outcome framework \citep{Neyman1923,Rubin1974}, we can
construct an exact randomization test for the sharp null hypothesis by
conditioning on all the potential outcomes
\citep{Rubin1980,Rosenbaum1983}. Randomization tests are widely
used in a variety of settings, including genetics
\citep{Spielman1993,Bates2020a}, clinical trials
\citep{Rosenberger2019}, program evaluation
\citep{Heckman2019} and instrumental variable analysis
\citep{Rosenbaum2004, Kang2018}.

\subsection{Our contributions}
\label{sec:our-contribution}

In this article, we propose a statistical framework that enables
researchers to use meiosis models as the ``reasoned basis'' for
inference in MR. We propose a randomization test that is \emph{almost
  exact} in the sense that the test would have exactly the nominal
size if the model for meiosis and fertilization were perfect. As
detailed below, our methodological developement in
\Cref{sec:exacttest} combines several important ideas in the
literature.

Our first contribution is a theoretical
description of MR and its assumptions in
\Cref{sec:setup} via the
language of causal directed acyclic graphs (DAGs)
\citep{spirtes2000causation,Pearl2009}. These graphical tools allow us to visualize and
dissect the assumptions
imposed on an MR study. In particular, we
show how various biological and social processes, including population
stratification, gamete formation, fertilization, genetic linkage,
assortative mating, dynastic effects, and pleiotropy, can be
represented using a DAG and how they can introduce bias in MR
analyses. Furthermore, by using single world intervention graphs
(SWIGs) \citep{Richardson2013b}, we identify sufficient confounder
adjustment sets to eliminate these sources of bias in
\Cref{sec:identification}. Our results
further provide theoretical insights into a fundamental trade-off
between statistical power and eliminating pleiotropy-induced bias.

For statistical inference, we propose in \Cref{subsec:hypothesis-test}
a randomization test by
connecting two existing literatures. The first literature concerns
randomization inference for instrumental variable analyses, which
usually assumes that the instrumental variables are randomized
according to a simple design (such as random sampling of a binary
instrument without replacement) \citep{Rosenbaum2004,
  Kang2018}. However, in MR, offspring genotypes are very
high-dimensional and are randomized based on the parental
haplotypes. The second literature attempts to identify the approximate
location of (``map'') causal genetic variants by modelling the meiotic
process \citep{morton1955sequential,Spielman1993,Bates2020a}. 
In
\Cref{subsec:test-statistic,subsec:hmm-simplification,subsec:multiple-ivs},
we consider some practical issues with the randomization tests. In
particular, we show
how the hidden Markov
model for meiosis and fertilization implied by \citet{Haldane1919} can
greatly simplify the sufficient adjustment
sets and the computation of our randomization test.

In addition to the considerable conceptual advantages, our almost
exact MR approach has several practical advantages too. First, unlike
model-based approaches for within-family MR \citep{Brumpton2020},
our approach does not rely on a correctly specified phenotype
model. Nonetheless, the randomization test can take advantage of a
more accurate phenotype model to increase its
power. Second, Haldane's hidden Markov model implies a propensity
score for each genetic instrument given a sufficient adjustment set
\citep{Rosenbaum1983}. This can be used as a ``clever covariate''
\citep{Rose2008} to build powerful test statistics with attractive
robustness properties. Third,
since the randomization test is exact, it is robust to arbitrarily
weak instruments. For an ``irrelevant'' instrument which induces no
variation in the exposure, the test will simply have no
power. Finally, by taking advantage of the DAG representation and
using a sufficient confounder adjustment set, our method is also
provably robust to biases arising from population structure (including
multi-ancestry samples), assortative mating, dynastic effects and
pleiotropy by linkage.

In \Cref{sec:simulation,sec:applied-example}, we demonstrate the
practicality of the almost exact approach to MR with
a simulation study and a real data example from the Avon Longitudinal
Study of Parents and Children (ALSPAC). The simulation study confirms
that the randomization test is exact under the null and
explores the power of the test in a number of scenarios. The applied
examples consists of a negative control and a positive control. The
negative control is the effect of child's body mass index (BMI) at age
7 on mother's BMI pre-pregnancy. Although a causal effect is
temporally impossible, the existence of confounders (a.k.a.\ backdoor
paths) may lead to false rejections of the null. The positive control is
the effect of child's BMI on itself plus some noise. We compare
our results with the results from a ``standard'' MR analysis
that does not condition on parental or offspring haplotypes. We
conclude with some further discussion in \Cref{sec:discussion}.

Throughout the paper, we use $i$ to index the parent-offspring trio (or
just the offspring) and $j$ to indicate a genomic locus. Bold font is
used to represent vectors and script font is used for sets.


\section{Background} \label{sec:background}

\subsection{Causal inference preliminaries} \label{sec:causal-inference}

We will express our model and assumptions about almost exact MR using
causal diagrams, then demonstrate that a randomization test for
instrumental variables is a natural vehicle for inference in
within-family MR. As such, a good grasp of these concepts is required
to understand the remainder of the article. This section lays out some
standard notation in causal inference. A lengthier introduction
to the causal inference concepts used in this article---including
causal graphical models, single world intervention graphs,
randomization inference, and instrumental variables---can be found in
\Cref{sec:intro-to-causal-inference}.

Suppose we have a collection of $N$ individuals indexed by $i = 1, 2,
\ldots, N$ and, among these individuals, we are interested in the
effect of an exposure $D_i$ on an outcome $Y_i$. For example, the
exposure could be the level of alcohol consumption over some period of
time and the outcome could be the incidence of
cardiovascular disease. Individual $i$'s \emph{potential} (or \emph{counterfactual}) \emph{outcomes} corresponding to exposure level
$D_i = d$ are given by $Y_i(d)$. We make the \emph{consistency}
assumption \citep{Hernan2020} which states that the observed outcome
corresponds to the potential outcome at the realized exposure level
$Y_i = Y_i(D_i)$.

Note that, in denoting the potential outcomes as $Y_i(d)$, we have
implicitly made the so-called \emph{stable unit treatment value assignment} (SUTVA) assumption \citep{Rubin1980, Imbens2015}. That is,
we have assumed that there is \emph{no interference} in the sense that
the potential outcomes of each individual are unaffected by the
exposures of other individuals. We have also assumed that there are no
hidden versions of the same exposure; this could be violated, for
example, if the effect of alcohol consumption on cardiovascular
disease has a dose-response relationship but the exposure $D_i$ is
only a binary indicator of alcohol consumption.

Potential outcomes may also be defined from a nonparametric structural
equation model associated with a causal diagram using recursive
substitution \citep{Pearl2009}. In such diagrams, vertices are used to
represent random variables and directed edges are used to represent
direct causal influences. The graphical and potential outcomes
approaches to causal inference can be nicely unified via the
single world intervention graphs \citep{Richardson2013b}. A brief
review of this can be found in the Appendix.

\subsection{Genetic preliminaries} \label{sec:genetic-preliminaries}

Before proceeding to present our model of within-family MR, it is also
instructive to provide a basic overview of some relevant concepts in
human genetics, with a focus on processes in genetic inheritance
such as \emph{meiosis} and \emph{fertilization}. For a
thorough exposition on statistical models for pedigree data, see
\citet{Thompson2000}.

Human somatic cells consist of 23 pairs of chromosomes, with one in
each pair inherited from the mother and the other from the
father. To simplify the discussion we will only consider autosomal
(non-sexual) chromosomes. Each chromosome is a doubled strand of
helical DNA composed of complementary nucleotide base pairs. A base
pair which exhibits
population-level variation in its nucleotides is called a \emph{single
  nucleotide polymorphism} (SNP). DNA sequences are typically
characterized by detectable variant forms induced by different
combinations of SNPs. These variant forms are called
\emph{alleles}. In this article, we will only consider variants with
two alleles. A set of alleles on one chromosome inherited together
from the same parent is called a \emph{haplotype} and the unordered
pair of haplotypes at the same locus is called a \emph{genotype}.

Meiosis is a type of cell division that results in reproductive cells
(a.k.a.\ gametes) containing one copy of each chromosome. During this
process, homologous chromosomes line up and exchange segments of DNA
between themselves in a biochemical process called
\emph{crossover}. The recombined chromosomes are then further divided
and separated into gametes. Since recombinations are infrequent
(roughly one to four per chromosome),
SNPs located nearby on the same parental chromosome are more likely to be
transmitted together, resulting in \emph{genetic
  linkage}. Fertilization is the process by which gametes in the
father (sperm cells) and mother (egg cells) join together to form a
zygote, which will then normally develop into an embryo.

We will mainly be concerned with genetic trio studies, in which we
observe the haplotypes of the mother, father and their child at $p$
loci. Let $\mathcal{J} = \{ 1, 2, \ldots , p \}$ be the set of SNP
indices. In our discussion below we will assume that the SNPs are on a
single chromosome, as different chromosomes are usually modelled
independently. We will denote the haplotypes as follows:
\begin{description}
\item offspring's haplotypes: $\bm{Z}^m = (Z_1^m, \ldots, Z_p^m)$ and $\bm{Z}^f = (Z_1^f, \ldots, Z_p^f)$,
\item mother's haplotypes: $\bm{M}^m = (M_1^m, \ldots, M_p^m)$ and
  $\bm{M}^f = (M_1^f, \ldots, M_p^f)$, and
\item father's haplotypes: $\bm{F}^m = (F_1^m, \ldots, F_p^m)$ and
  $\bm{F}^f = (F_1^f, \ldots, F_p^f)$,
\end{description}
 where the superscript $m$ (or $f$) indicates a maternally (or
 paternally) inherited haplotype. We only consider SNPs with two alleles,
 so each of the six haplotype vectors above is in $\{0,1\}^p$.
 Furthermore, let $\bm{M}_j^{mf} = (M_j^m, M_j^f)$ denote the mother's
 haplotypes at locus $j$ and similarly
 for $\bm{F}_j^{mf}$ and $\bm{Z}_j^{mf}$. The offspring's genotype
 at locus $j \in \mathcal{J}$ is given by $Z_j = Z_j^m + Z_j^f$ and let
 $\bm{Z} = \bm{Z}^m + \bm{Z}^f \in \{ 0, 1, 2 \}^p$ denote the vector
 of offspring genotypes.

 \begin{figure}[t]
   \begin{center}
     \includegraphics[width=0.55\columnwidth]{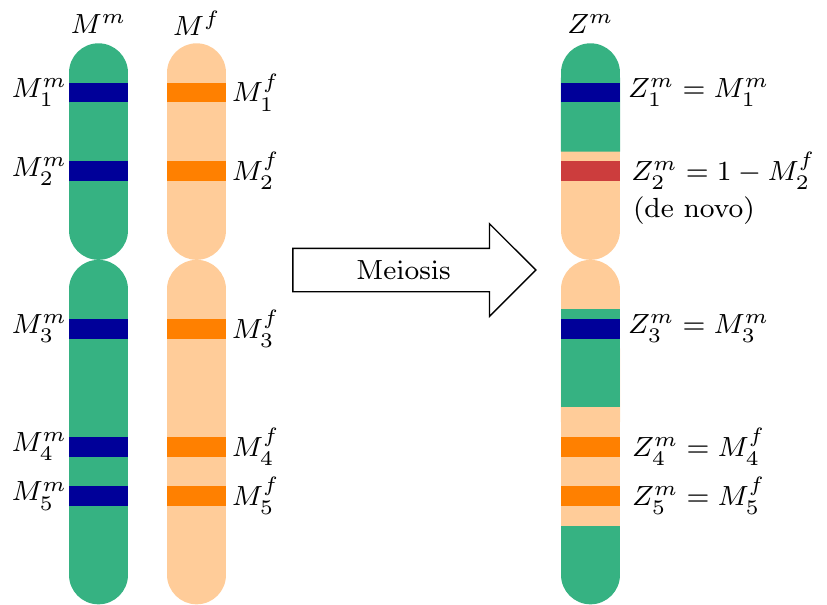}
     \caption{Illustration of the meiotic process with five loci on a chromosome.}
     \label{fig:recombination}
   \end{center}
 \end{figure}

 \Cref{fig:recombination} illustrates how an offspring's
 maternally-inherited haplotype $\bm{Z}^m$ at five loci on a
 chromosome are related to the mother's haplotypes
 $\bm{M}^m$ and $\bm{M}^f$. In a gamete produced by meiosis, the
 allele at locus $j \in \mathcal{J}$ is inherited from either the
 mother's maternal haplotype or father haplotype (excluding
 mutations). This can be formalized as an ancestry  indicator, $U_j^m
 \in \{m,f\}$. The classical meiosis model of \citet{Haldane1919}
 assumes that  $\bm U^m
 = (U_1^m,\dotsc,U_p^m)$ follows a (discretized) homogeneous Poisson
 process. Haldane's model is described in
 \Cref{sec:randomization-distribution} in detail and can simplify the
 randomization test considerably
 (\Cref{subsec:hmm-simplification}). Nonetheless, our ``almost exact''
 approach to MR is modular and does not rely on a
 specific meiosis model. In fact, it is theoretically straightforward
 to incorporate more sophisticated meiosis models that allow for
 ``interference'' between the crossovers \citep{Otto2019}. As the
 meiosis model become more accurate, our test will become closer to
 exact randomization inference.

 The description in the last paragraph does not take genetic mutation
 into account. Many meiosis models, such as the one in
 \citet{Haldane1919}, assume that there is a small probability of
 independent mutations:
 \begin{assumption} \label{assum:de-novo-distribution}
   Given mother's haplotypes $\bm M_j^{mf}$, the ancestry indicator
   $U_j^m$, and that fertilization occurs (this is represented as
   $S=1$ in \Cref{sec:exacttest}),  
   \[
     Z_j^m   = \begin{cases}
       M_j^{(U_j^m)}, & \text{with probability $1 - \epsilon$,} \\
       1-M_j^{(U_j^m)}, & \text{with probability $\epsilon$.}
     \end{cases}
   \]
   The same model holds for the paternally-inherited haplotypes.
 \end{assumption}
 The rate of \emph{de novo} mutation $\epsilon$ is about $10^{-8}$ in
 humans \citep{Acuna-Hidalgo2016}. When only one generation is
 considered (as in a genetic trip study), it often
 suffices to treat $\epsilon = 0$ (i.e.\ no mutation) for practical
 purposes, unless it is
 desirable to obtain the exact randomization distribution under a
 recombination model.

The above meiosis model assumes no \emph{transmission
  ratio distortion}. Transmission ratio distortion occurs when one of
the two parental alleles is passed on to the (surviving) offspring at
more or less than the expected Mendelian rate of 50\%
\citep{Davies2019}. Transmission ratio distortion falls into two
categories: segregation distortion, when processes during meiosis or
fertilization select certain alleles more frequently than others, and
viability selection, when the viability of zygotes themselves depend
on the offspring genotype. We sidestep this discussion for now and
return to it in \Cref{sec:discussion}.

 \section{Almost exact Mendelian randomization} \label{sec:exacttest}

Returning to the alcohol and cardiovascular disease example in
\Cref{sec:causal-inference}, observational studies suggest that
moderate alcohol consumption confers reduced risk relative to
abstinence or heavy consumption \citep{Millwood2019}. However, this
could be a result of systematic differences among people with
different drinking patterns (confounding) rather than a causally
protective effect of moderate drinking. For this reason, Mendelian
randomization has become a popular study design to investigate the
long-term health effects of alcohol drinking
\citep{chen08_alcoh_intak_blood_press}.

The \emph{ALDH2} gene regulates acetaldehyde metabolism and is
often used as an instrumental variable for alcohol consumption. In
East Asian populations, an allele of
\emph{ALDH2} produces a protein that is inactive in metabolising
acetaldehyde, causing flushing and discomfort while drinking and
thereby reducing consumption. Thus, we might like to use the random
allocation of variant copies of \emph{ALDH2} as the basis of causal
inference. To this end, we need to carefully clarify the conditions under
which this inference would be valid.

 \subsection{A causal model for Mendelian inheritance}
 \label{sec:setup}

Next, we construct a general grahical model to describe the
process of Mendelian inheritance and genotype-phenotype
relationships. This causal model allows us to identify sources of bias
in within-family MR and construct sufficient adjustment sets to
control for them. Under this causal model, the central idea behind
almost exact MR is to base statistical inference precisely on
randomness in genetic inheritance via a model for meiosis and
fertilization. 

 \begin{figure*}[htp]
   \begin{center}
     \includegraphics[width=0.9\textwidth]{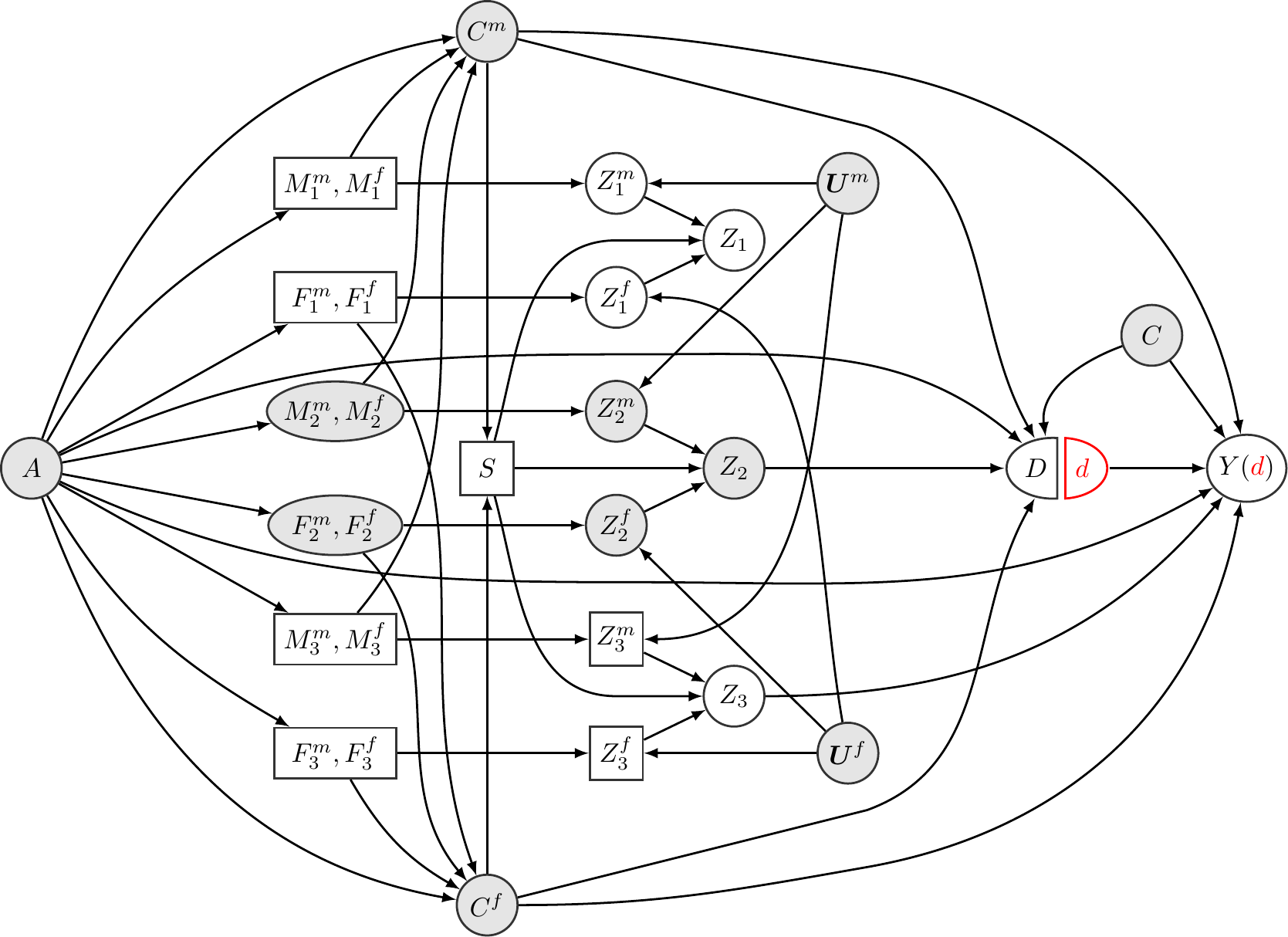}
     \caption{The single world intervention graph for a working
       example of a chromosome with $p = 3$ variants. Transparent
       nodes are observed and grey nodes are unobserved. Square nodes
       are the confounders being conditioned on in
       \Cref{prop:example-adjustment-set}. $A$ is ancestry; $\bm M^f =
       (M_1^f,M_2^f,M_3^f)$ is the mother's haplotype inherited from
       her father; $\bm M^m,\bm F^m$, and $\bm F^f$ are defined
       similarly; $C^m$ and $C^f$ are generic phenotypes of the mother
       and father; $S$ is an indicator of mating; $\bm Z^m =
       (Z_1^m,Z_2^m,Z_3^m)$ is the offspring's maternal haplotype and
       $\bm U^m$ is a meiosis indicator; $\bm Z^f$ and $\bm U^f$ are
       defined simlarly; $\bm Z = (Z_1,Z_2,Z_3)$ is the offspring's
       genotype; $D$ is the exposure of interest; $Y(d)$ is the
       potential outcome of $Y$ under the intervention that sets $D$
       to $d$; $C$ is an environmental confounder between $D$ and
       $Y$.}
     \label{fig:fammr_dag}
   \end{center}
 \end{figure*}

 \Cref{fig:fammr_dag} shows a working example of our causal model
 on a hypothetical chromosome with $p=3$ variants. The directed
 acyclic graph
 is structured in roughly chronological order from left to right,
 where $A$ describes the population structure, $S$ is an indicator for
 mating, and $C$ is any environmental confounder between the exposure
 $D$ and outcome $Y$. 


 At first glance, \Cref{fig:fammr_dag} appears to be quite complicated
 but, by the modularity of graphical models, it can be
 decomposed into a collection of much simpler subgraphs that describe different
 biological processes (\Cref{fig:subgraphs}). By
 definition, a joint distribution \emph{factorizes} according to the
 DAG in \Cref{fig:fammr_dag} if its density function can be decomposed
 as described in \Cref{tab:factorization}.


\begin{table*}[h]
  \caption{Factorization of the joint density of all variables in
    \Cref{fig:fammr_dag} ($p$ is
    used as a generic symbol for density functions).}
  \label{tab:factorization}
  \centering
  \begin{tabular}{lll}
    \toprule
    Terms & Interpretation & More detail \\
    \midrule
    $p(A) p\big(\bm U^m\big) p\big(\bm U^f\big) p(C)$ & Exogenous variables & \\[0.3em]
    $p\big(\bm M^m, \bm M^f \mid A\big) p\big(\bm F^m, \bm F^f \mid A\big)$ & Population
                                                     stratification &
                                                                      \Cref{sec:parental-genotypes}
    \\[0.3em]
    $p\big(C^m \mid A, \bm M^m, \bm M^f\big) p\big(C^f \mid A, \bm F^m, \bm F^f\big)$
          & Parental phenotypes & \Cref{sec:parental-phenotypes}
    \\[0.3em]
$p\big(\bm Z^m \mid \bm M^m, \bm M^f, \bm U^m\big) p\big(\bm Z^f \mid \bm F^m, \bm
    F^f, \bm U^f\big)$ & Meiosis & \Cref{sec:offspring-genotypes} \\[0.3em]
    $p\big(S \mid C^m, C^f\big)$  & Assortative mating &
                                                 \Cref{sec:offspring-genotypes}
    \\[0.3em]
    $p\big(\bm Z \mid \bm Z^m, \bm Z^f, S\big)$ & Fertilization &
                                           \Cref{sec:offspring-genotypes}
    \\[0.3em]
    $p\big(D \mid  A, \bm Z, C^m, C^f, C\big) p\big(Y(d) \mid A, \bm Z, C^m, C^f,
    C\big)$
    & Offspring phenotypes, &
                                                           \Cref{sec:offspring-phenotypes}
    \\
    & dynastic effect, confounding & \\
    \bottomrule
  \end{tabular}
\end{table*}

Next, we describe each term in \Cref{tab:factorization} and what this
factorization implies about our assumptions on the biological
processes. To simplify the discussion, we assume all DAGs in this
article are faithful, so conditional independence between random
variables is equivalent to d-separation in the DAG.

 \begin{figure}[htp]
\centering
\begin{subfigure}[b]{0.35\columnwidth}
 \centering
 \includegraphics[width=0.7\textwidth]{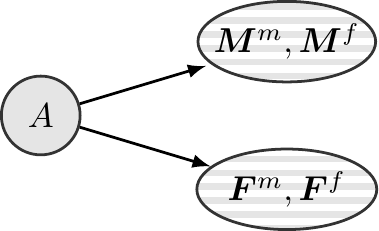}
 \caption{Population stratification (\Cref{sec:parental-genotypes}).}
 \label{fig:parental-genotypes}
\end{subfigure}
\quad
\begin{subfigure}[b]{0.6\columnwidth}
     \centering
     \includegraphics[width=\textwidth]{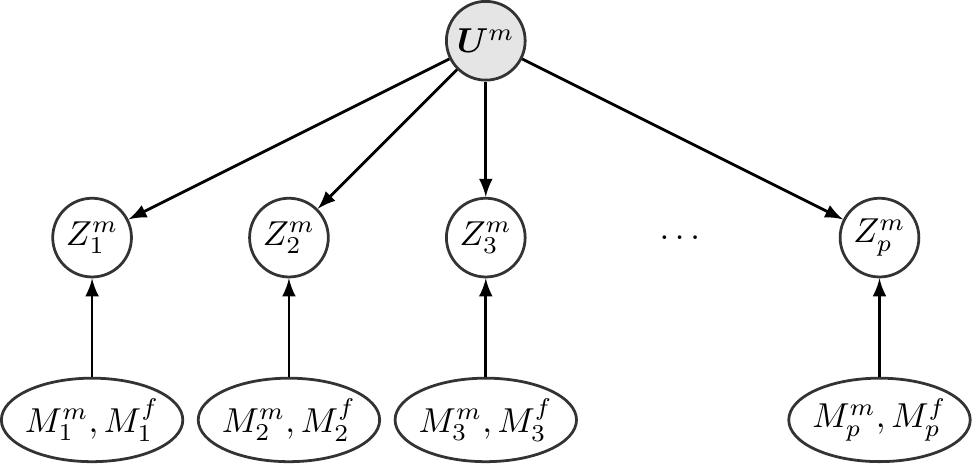}
     \caption{Recombination of mother's haplotypes (\Cref{sec:genetic-preliminaries}).}
     \label{fig:meiosis_dag}
    \end{subfigure}

    \par\bigskip

    \begin{subfigure}[b]{0.45\columnwidth}
     \centering
     \includegraphics[width=0.7\textwidth]{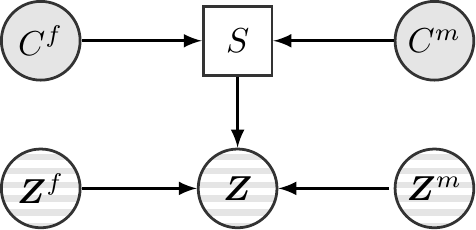}
     \caption{Assortative mating (\Cref{sec:offspring-genotypes}).}
     \label{fig:assortative-mating}
   \end{subfigure}
   \quad
   \begin{subfigure}[b]{0.45\columnwidth}
     \centering
     \includegraphics[width=0.7\textwidth]{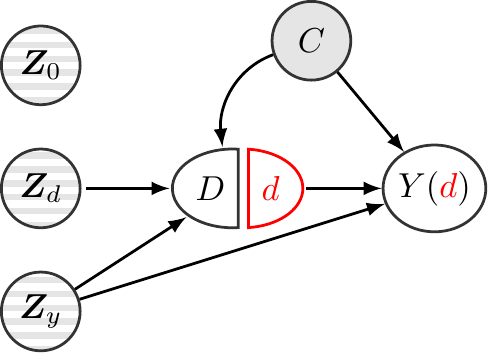}
     \caption{Offspring phenotypes. (\Cref{sec:offspring-phenotypes}).}
     \label{fig:offspring-phenotypes}
   \end{subfigure}

   \par\bigskip

    \begin{subfigure}[b]{0.5\columnwidth}
     \centering
     \includegraphics[width=0.9\textwidth]{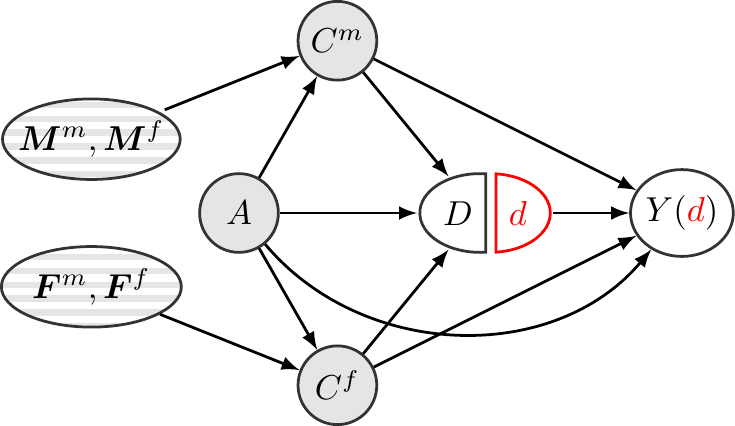}

     \caption{Dynastic effects (\Cref{sec:parental-phenotypes}).}
     \label{fig:dynasticeffects}
    \end{subfigure}%
    \caption{The constituent subgraphs of our within-family Mendelian
     randomization model. White nodes represent observed variables;
     grey nodes represent unobserved variables; and striped nodes
     represent variables for which some elements may be unobserved.}
   \label{fig:subgraphs}

 \end{figure}

\subsubsection{Parental genotypes} \label{sec:parental-genotypes}

We assume that parental genotypes originate from some arbitrary, latent
population structure. Population stratification is a phenomenon characterized by systematic differences in the distribution of alleles among subgroups of a population. These disparities
typically emerge from social and genetic mechanisms including
non-random mating, migration patterns and `founder effects'
\citep{Cardon2003} and can often be detected by principal
component analysis \citep{patterson06_popul_struc_eigen}. Population stratification
can introduce spurious associations between genetic variants and
traits \citep{Lander1994}.

We represent population structure via a latent node $A$ in
\Cref{fig:fammr_dag,fig:parental-genotypes}. The arrows from $A$ to
$\bm{M}^m,
\bm{M}^f$ and $\bm{F}^m, \bm{F}^f$ indicate that the distribution of
parental haplotypes depends on the latent population structure:
\[
 A \notindependent (\bm M^m, \bm M^f, \bm F^m, \bm F^f).
\]

The node $A$ may also capture any linkage disequilibrium in the
parental haplotypes. That is, because the parental haplotypes are
determined by the same process as the grandparental haplotypes and
so on, recombination introduces dependence among nearby genetic
variants (see \Cref{sec:offspring-genotypes} below for more
discussion). The precise
distribution of $A$ and the conditional
distribution of the parental haplotypes given $A$ are not important
below, because an appropriate
subset of the parental haplotypes will be conditioned on and any paths
that involve edges from $A$ to $\bm M^m$, $\bm M^f$, $\bm F^m$, and
$\bm F^f$ will be blocked.

\subsubsection{Parental phenotypes} \label{sec:parental-phenotypes}

We impose no assumptions on the nature and the distribution of the parental
phenotypes $C^m$ and $C^f$. They can depend arbitrarily on the
parental haplotypes $\bm{M}^m, \bm{M}^f, \bm{F}^m, \bm{F}^f$ and the
population structure $A$:
\[
 C^m \notindependent (A, \bm{M}^m, \bm{M}^f),~
 C^f \notindependent (A, \bm{F}^m, \bm{F}^f).
\]%
It is not necessary to model such dependence because the almost exact
approach to MR conditions on appropriate parental haplotypes.


\subsubsection{Offspring genotypes} \label{sec:offspring-genotypes}

There are two biological processes involved in the genesis of the
offspring's genotypes: meiosis (gamete formation) and
fertilization. The meiotic process has been reviewed in
\Cref{sec:genetic-preliminaries}, and the key
\Cref{assum:de-novo-distribution} can be represented by the causal
diagram in \Cref{fig:meiosis_dag} (for just the mother). A crucial
assumption underlying almost exact MR is the exogeneity of the meiosis
indicators $\bm U^m$ and $\bm U^f$. This is reflected in
\Cref{fig:fammr_dag,fig:meiosis_dag}, as $\bm
U^m$ and $\bm U^f$ have no parents and their only children are the
offspring's haplotypes. Formally, we assume:
\begin{assumption} \label{assum:meiosis-indicator-independence}
The meiosis indicators are independent of parental haplotypes and
phenotypes and any other confounders:
\[
 (\bm{U}^m, \bm{U}^f) \independent (A, C^m, C^f, C, \bm{M}^m, \bm{M}^f, \bm{F}^m, \bm{F}^f).
\]%
\end{assumption}

Many stochastic processes have been proposed to model the distribution
of the ancestry indicator $\bm U^m$;  see \citet{Otto2019} for a recent
review. Due to the dependence in $\bm U^m$, nearby alleles on the same
chromosome tend to be inherited together. This phenomenon is known as
\emph{genetic linkage}. In \Cref{sec:genetic-preliminaries}, we have
described the classical model of \citet{Haldane1919} which assumes
$\bm U^m$ follows a Poisson process, which was used by
\citet{Bates2020a} to map causal variants. We will see in
\Cref{subsec:hmm-simplification} that the Markov properties of a
Poisson process greatly simplify randomization inference.

Another mechanism that needs to be modeled is fertilization. In
Mendelian inheritance, it is assumed that the potential gametes
(sperms and eggs) come together at random. However, mating may not be
a random event. \emph{Assortative mating} refers to the phenomenon
where individuals are more likely to mate if they have complementary
phenotypes. For example, there is evidence in UK Biobank that
variant forms of the \emph{ADH1B} gene related to alcohol
consumption are more likely to be shared among spouses than
non-spouses \citep{Howe2019}. This suggests assortative mating on
drinking behaviour and may introduce bias to MR studies on alcohol
consumption \citep{Hartwig2018}.

The subgraph describing assortative mating is shown in
\Cref{fig:assortative-mating}, where the mating
indicator $S \in \{0, 1\}$ is dependent on the parental
phenotypes $C^m$ and $C^f$.
Any MR study necessarily conditions on $S = 1$, otherwise the
offspring would not exist and the trio data would not be
available. This is formalized in
\Cref{fig:assortative-mating} by the arrows from $S$ to $\bm{Z}$. In
particular, we may define the offspring's genotype $\bm Z$ as
\begin{align} \label{eqn:offspring-genotype-undefined}
\bm{Z} = \begin{cases}
 \bm{Z}^m + \bm{Z}^f, & \text{if $S = 1$},  \\
 \text{Undefined}, & \text{if $S = 0$.}
\end{cases}
\end{align}
Notice that the above definition recognizes the fact that the gametes
$\bm{Z}^m$ and $\bm{Z}^f$ are produced regardless of whether
fertilization actually takes place.

Finally, the assumption of no transmission
ratio distortion is formalized by the absence of arrows from
$\bm{Z}^m$ and $\bm{Z}^f$ to $S$. This is not necessarily a benign
assumption, given empirical evidence that gametes may pair up
non-randomly \citep{Nadeau2017}.

\begin{assumption} \label{assum:no-trd}
There is no transmission ratio distortion in the sense that
\[
S \independent (\bm{Z}^m, \bm Z^f) \mid (\bm{M}^{mf}, \bm{F}^{mf}).
\]
\end{assumption}

\subsubsection{Offspring phenotypes} \label{sec:offspring-phenotypes}

Finally, we describe assumptions on the offspring phenotypes. We
are interested in estimating the causal effect of an offspring
phenotype $D \in \mathcal{D} \subseteq
\mathbb{R}$ on another offspring phenotype $Y  \in
\mathcal{Y} \subseteq \mathbb{R}$. We refer to $D$ as the exposure
variable and $Y$ as the outcome variable. These phenotypes are
determined by the offspring genotypes and environmental factors
(including parental traits). For a particular realization of the genotypes $\bm z$, we denote the counterfactual exposure as $D(\bm{z})$. Furthermore, under an additional intervention that sets
$D$ to $d$, we denote the counterfactual outcome as $Y(\bm{z},
d)$. These potential outcomes are related to the observed data tuple
$(\bm{Z}, D, Y)$ by
\[\begin{array}{ll}%
 D = D(\bm{Z}),~ 
 Y = Y(\bm{Z}, D) = Y(\bm{Z},D(\bm{Z})), 
\end{array}\]%
which is a simple extension of the consistency assumption
\eqref{eq:consistency} before.

We are interested in making inference about the counterfactuals $Y(d)
= Y(\bm{Z},d)$ when the exposure is set to $d \in \mathcal{D}$. As the
exposure $D$ typically varies according to the population structure,
parental phenotypes and other environmental factors, it is not
randomized in the sense that
\[
 Y(d) \notindependent D ~ \text{for some or all $d \in
   \mathcal{D}$}.
\]
For example, if $D$ is alcohol consumption and $Y$ is cardiovascular
disease, there may exist offspring confounders such as diet or smoking
habits which are common causes of both $D$ and $Y$. The exact nature
of the confunders is not very important for MR as it tries to use
unconfounded randomness (in $\bm U^m$ and $\bm U^f$) to make causal
inference.

It will be helpful to categorize the genetic variants based on whether
they have direct causal effects on $D$ and/or $Y$.

\begin{assumption} \label{assum:snp-subsets}
The set $\mathcal{J} = \{1, \ldots, p\}$ of genetic variants can be
paritioned as $\mathcal{J} = \mathcal{J}_y \cup \mathcal{J}_d \cup
\mathcal{J}_0$, where
\begin{itemize}
\item $\mathcal{J}_y$ includes all \emph{pleiotropic} variants with a direct causal
 effect on $Y$ not mediated by $D$ (some of which may have a causal
 effect on $D$ as well).
\item $\mathcal{J}_d$ includes all causal variants for $D$ with no direct effect on $Y$.
\item $\mathcal{J}_0 = \mathcal{J} \setminus (\mathcal{J}_y \cup
  \mathcal{J}_d)$ includes all other variants.
\end{itemize}
\end{assumption}

In our working example in \Cref{fig:fammr_dag}, $\mathcal{J}_y =
\{3\}$, $\mathcal{J}_d = \{2\}$, and $\mathcal{J}_0 = \{1\}$. If the
exposure $D$ indeed has a causal effect on the outcome $Y$, the loci
of the causal variants of $Y$ are given by $\mathcal{J}_y \cup
\mathcal{J}_d$.

For subscripts $s \in \{ 0, d, y \}$, we let $\bm{Z}_s = \{
Z_j \colon j \in \mathcal{J}_s \}$ denote the corresponding genotypes,
which has support $\mathcal{Z}_s = \{ 0, 1, 2
\}^{|\mathcal{J}_s|}$. By \Cref{assum:snp-subsets}, our potential
outcomes can be written as (with an abuse of notation)
\begin{align*}
  &D(\bm z) = D(\bm z_d),~Y(\bm z, d) = Y(\bm z_y, d),
  Y(\bm z) = Y(\bm z_y, D(\bm z_d)) = Y(\bm z_y, \bm z_d),
\end{align*}
where $\bm z = (\bm z_d, \bm z_y, \bm z_0) \in \mathcal{Z}_d \times \mathcal{Z}_y \times \mathcal{Z}_0 = \mathcal{Z}$ and $d \in \mathcal{D}$.

\Cref{fig:offspring-phenotypes} provides the graphical representation
of \Cref{assum:snp-subsets}. Each element of $\bm{Z}_0$ has no
effect on $D$ or $Y(d)$, each element of $\bm{Z}_d$ has an effect on
$D$ and each element of $\bm{Z}_y$ has an effect on $Y(d)$ (some are
also causes of $D$). The vector $\bm{Z}_y$ contains the so-called
pleiotropic variants that are causally involved in the
expression of multiple phenotypes \citep{Hemani2018a}. Universal
pleiotropy is assumed in the famous infinitesimal model of
\citet{fisher19_correlation}. Currently, the widely
accepted view is that pleiotropy is at least widespread and some have
argued for a omnigenic model \citep{Boyle2017}.


\emph{Dynastic effects}, sometimes called \emph{genetic nurture}
\citep{Kong2018}, is a phenomenon characterized by parental genotypes
exerting an influence on the offspring's
phenotypes via the parental phenotypes. This is depicted in
\Cref{fig:dynasticeffects}, where paths emanate from the
parental haplotypes $\bm M^{mf}$ and $\bm F^{mf}$ to the parental
phenotypes $C^m$ and $C^m$ and on to the offspring phenotypes $D$ and
$Y$.

\subsection{Conditions for identification}
\label{sec:identification}

With the causal model outlined in \Cref{sec:setup} in mind, we now
describe some sufficient conditions under which some $Z_j \in \bm Z$
is a valid instrumental variable for estimating the causal effect of
$D$ on $Y$. Recall that an instrumental
variable induces unconfounded variation in the exposure without otherwise
affecting the outcome. Due to population stratification
(\Cref{fig:parental-genotypes}), assortative mating
(\Cref{fig:assortative-mating}), and dynastic effects
(\Cref{fig:dynasticeffects}), the offpsring genotypes $\bm Z$ as a
whole are usually not properly randomized without conditioning on the
parental haplotypes. That is,
\[
\bm Z \not \independent D(\bm z), Y(\bm z, d) ~ \text{for some or
 all $\bm z \in \mathcal{Z}$ and $d \in \mathcal{D}$.}
\]

To restore validity of genetic instruments, the key idea is to
condition on the parental haplotypes as envisioned by
\citet{DaveySmith2003} and used in the context of gene mapping by
\citet{Spielman1993} and \citet{Bates2020a}. This allows us to
use precisely the exogenous randomness in the ancestry indicators $\bm
U^m$ and $\bm U^f$ that occurs during meiosis and fertilization. This idea is formalized in the next proposition.

\begin{proposition}
\label{prop:offspring-genotype-independence}
Under the causal graphical model described in \Cref{sec:setup},
the offspring's maternal haplotype $Z_j^m$ (or genotype $Z_j$) at some site
$j \in \mathcal{J}$ is independent of all ancestral and offspring
confounders given the maternal (or parental) haplotypes at site $j$:
\begin{equation}
 \label{eq:zj-independence}
 \begin{split}
   &Z_j^m \independent (A, C^m, C^f, C) \mid (\bm M_j^{mf}, S=1),~
   Z_j \independent (A, C^m, C^f, C) \mid (\bm M_j^{mf}, \bm
     F_j^{mf}, S=1). \\
 \end{split}
\end{equation}
\end{proposition}


However, the conditional independence \eqref{eq:zj-independence}
alone does not guarantee the validity of $Z_j$ as an instrumental
variable. The main issue is that $Z_j$ might be in linkage
disequilibrium with other causal variants of $Y$. Our goal is
to find a set of variables $\bm V$ such that $Z_j$ is conditionally
independent of the potential outcome $Y(d)$. This is formalized in
the definition below.
\begin{definition} \label{def:identification-conditions}
We say a genotype $Z_j$ is a \emph{valid instrumental variable} given $\bm V$ (for estimating the
causal effect of $D$ on $Y$) if the following conditions are
satisfied:
\begin{enumerate}
\item Relevance: $Z_j \not \independent D \mid \bm V$;
\item Exogeneity: $Z_j \independent Y(d) \mid \bm V ~ \text{for all
   $d \in \mathcal{D}$}$;
\item Exclusion restriction: $Y(z_j, d) = Y(d)$ for all
   $z_j \in \{0,1,2\}$ and $d \in \mathcal{D}$.
\end{enumerate}
Simiarly, we say a haplotype $Z_j^m$ is a valid instrument given
$\bm V$ if the same conditions above hold with $Z_j$ replaced by
$Z_j^m$ and $z_j \in \{0,1,2\}$ replaced by $z_j^m \in \{0,1\}$.
\end{definition}
In our setup (\Cref{assum:snp-subsets}), the exclusion restriction is
satisfied if and only if $j \not \in \mathcal{J}_y$.

To gain intuition on how the set of variables $\bm V$ can be
selected, it is helpful to return to the working example in
\Cref{fig:fammr_dag}. We see that
$Z_3$ does not satisfy the exclusion restriction because $Z_3$ has a
direct effect on $Y$. The causal variant $Z_2$ for $D$ would be a valid
instrument if we condition on the corresponding haplotypes and $Z_3$,
but $Z_2$ is not observed in this example. This leaves us with one
remaining candidate instrument: $Z_1$ (and potentially its
haplotypes $Z_1^m$ and $Z_1^f$). The relevance
assumption is satisfied as long as $\bm V$ does not block both of the
following paths
\[\begin{array}{ll}%
 Z_1 \leftarrow Z_1^m \leftarrow \bm{U}^m \rightarrow Z_2^m \rightarrow Z_2 \rightarrow D; \\[0.3em]
 Z_1 \leftarrow Z_1^f \leftarrow \bm{U}^f \rightarrow Z_2^f \rightarrow Z_2 \rightarrow D. \\[0.3em]
\end{array}\]%
The exclusion restriction is satisfied because $Z_1$ is not a causal
variant for $Y$. Finally, exogeneity is satisfied if $\bm V$ blocks the
path
\[\begin{array}{ll}%
 Z_1 \leftarrow Z_1^m \leftarrow \bm{U}^m \rightarrow Z_3^m \rightarrow Z_3 \rightarrow Y(d); \\[0.3em]
 Z_1 \leftarrow Z_1^f \leftarrow \bm{U}^f \rightarrow Z_3^f \rightarrow Z_3 \rightarrow Y(d). \\[0.3em]
\end{array}
\]%
Thus, we have the following result:

\begin{proposition} \label{prop:example-adjustment-set}
For the example in \Cref{fig:fammr_dag}, the following conditional
independence relationships are true for all $d \in
\mathcal{D}$:
\begin{align}
 \label{eqn:adjustment-set-1} Z_1^m \independent Y(d) &\mid (\bm M_1^{mf}, \bm V_{3}^m, S = 1), \\
 \label{eqn:adjustment-set-2} Z_1 \independent Y(d) & \mid (\bm M_1^{mf}, \bm F_1^{mf}, \bm V_{3}, S=1),
\end{align}
where $\bm V_{3}^m  = (\bm M_3^{mf}, Z_3^m)$ and $\bm V_{3} = ( \bm M_3^{mf}, \bm F_3^{mf}, Z_3)$.
The adjustment variables above are minimal in the sense
that no subsets of them satisfy the same conditional independence.
\end{proposition}
\begin{table*}[tb]
    \caption{Some paths between $Z_1$ and $Y(d)$ in \Cref{fig:fammr_dag}.}
    \label{tab:dseparation}
    \begin{center}
     \begin{tabular}{lll}
       \toprule
       Name of bias & Example path & Blocking variable \\
       \midrule
       Dynastic effect & $ Z_1^m \leftarrow \bm M_1^{mf} \rightarrow
                         C^m \rightarrow Y(d)$ & $\bm M_1^{mf}$ \\[0.3em]
       Population stratification &  $ Z_1^m \leftarrow \bm M_1^{mf}
                                   \leftarrow A \rightarrow Y(d)$ &
                                                                    $\bm M_1^{mf}$ \\[0.3em]
       Pleiotropy & $Z_1^m \leftarrow \bm{U}^m  \rightarrow Z_3^m \rightarrow Z_3 \rightarrow Y(d)$ & $Z_3^m$ or $Z_3$ \\[0.3em]
       Assortative mating & $Z_1^m \leftarrow \bm M_1^{mf} \leftarrow
                              C^m \rightarrow \boxed{S} \leftarrow$
       &
                                                                      $\bm
                                                                      M_1^{mf}$,
                                                                      $Z_3$,
         or $\bm F_3^{mf}$ \\
       &                               $\leftarrow C^f
                            \leftarrow \bm F_3^{mf} \rightarrow
                            Z_3^f \rightarrow Z_3 \rightarrow Y(d)$
 & \\
       Nearly determined ancestry & $Z_1^m \leftarrow \bm{U}^m
                                    \rightarrow \boxed{Z_3^m}
                                    \leftarrow \bm M_3^{mf} \rightarrow
                                    C^m \rightarrow Y(d) $ & $\bm M_3^{mf}$ \\
       \bottomrule
     \end{tabular}
    \end{center}
\end{table*}

\begin{proof}
The conditional independence follows almost immediately from
our discussion above. It is tedious but trivial to show that $\bm V = (\bm M_1^{mf}, \bm V_{3}^m)$ is
minimal for \eqref{eqn:adjustment-set-1}.  \Cref{tab:dseparation} lists the key backdoor
paths between $Z_1^m$ and $Y(d)$, describes the corresponding biological mechanism and shows how conditioning on $\bm V$ blocks these paths. The table only includes the maternal
paths, but the same blocking also holds for the corresponding paternal
paths.
\end{proof}

\begin{remark}
To our knowledge, the bias-inducing path in the last row of
\Cref{tab:dseparation}, which we term ``nearly
determined ancestry bias'',
has not yet been identified in the literature. This is a form of
collider bias introduced because the ancestry indicator
can often be almost perfectly determined if we are given the
mother's haplotypes and the offspring's maternal haplotype. For
example, if the mother is heterozygous $M_3^m = 1,
M_3^f = 0$ and the offspring's maternal haplotype is $Z_3^m = 1$,
then we know that $U_3^m = m$ is true with an extremely high
confidence. Due to genetic linkage, there is also an exceedingly high
probability that $U_1^m = m$, inducing dependence between the
potential instrument $Z_1^m$ with maternal haplotypes $\bm
M^{mf}_3$. This further challenges the widely adopted hypothesis that
mapping causal variants is equivalent to testing conditional
independence; in our example, $Z_3$ is the only causal variant of
$Y(d)$, but $Z_1 \independent Y(d) \mid Z_3$ may not be true even if
there is no population stratification and no causal effect from $\bm
M^{mf}_{1}$ on $Y(d)$.
\end{remark}

We conclude this section with a sufficient condition for the validity
of $Z_j^m$ and $Z_j$ in our general setting. To simplify the
exposition, let $\bm V_{\mathcal{B}}^m = (\bm{M}_{\mathcal{B}}^{mf},
\bm{Z}_{\mathcal{B}}^m)$ be a set of maternal
adjustment variables, where $\mathcal{B} \subseteq \mathcal{J}
\setminus \{j\}$ is a subset of loci. Furthermore, let $\bm
V_{\mathcal{B}} = (\bm{M}_{\mathcal{B}}^{mf}, \bm
{F}_{\mathcal{B}}^{mf}, \bm{Z}_{\mathcal{B}})$.

\begin{theorem} \label{thm:sufficient-adjustment-set}
Suppose $\bm{Z} = (Z_1,\dotsc, Z_p)$ is a full chromosome. Consider
the general causal model for Mendelian randomization in
\Cref{sec:setup} and let $j \in
\mathcal{J}$ be the index of a candidate instrument. Then $Z_j^m$ is a
valid instrument conditional on $(\bm M_j^{mf}, \bm V_{\mathcal{B}}^m)$ if the following conditions
are satisfied:
\begin{enumerate}
\item $Z_j^m \not \independent \bm Z_d^m \mid (\bm M_j^{mf}, \bm V_{\mathcal{B}}^m, S = 1)$;
\item $Z_j^m \independent \bm Z_y^m \mid (\bm M_j^{mf}, \bm V_{\mathcal{B}}^m, S = 1)$;
\end{enumerate}
\end{theorem}

\begin{proof}
The relevance of $Z_j^m$ follows from the first condition, because $Z_j^m$
is dependent on some causal variants (or is itself a causal
variant) of $D$. The exclusion restriction ($j \not \in
\mathcal{J}_y$) follows directly from the second
condition. For exogeneity, paths from $Z_j^m$ to $Y(d)$ either go
through the confounders $A$, $C^f$, $C^m$, or $C$, which are blocked
by $\bm M_j^{mf}$ by \Cref{prop:offspring-genotype-independence},
or through some causal variants of the outcome as in $Z_j^m \leftarrow
\bm U^m \rightarrow \bm Z_y^m \rightarrow \bm Z_y \rightarrow Y(d)$,
which are blocked by the second condition.
\end{proof}

It is straightforward to extend
\Cref{thm:sufficient-adjustment-set} to establish validity of the
genotype $Z_j$ at locus $j$ as an instrumental variable. Details are
omitted.

Since \Cref{prop:offspring-genotype-independence} ensures that, after
conditioning on $\bm{M}_j^{mf}$, the instrument $Z_j^m$ is independent
of all ancestral and offspring confounders $(A, C^m, C^f, C)$, the
only remaining threats to the validity of $Z_j^m$ are
irrelevance and pleiotropy. The set $\mathcal{B}$ is chosen to ensure
that $Z_j^m$ is independent of all pleiotropic variants conditional on
$\bm{V}_{\mathcal{B}}^m$ (condition 2 of \Cref{thm:sufficient-adjustment-set})
but not independent of the set of causal variants (condition 1 of
\Cref{thm:sufficient-adjustment-set}).

This highlights an intrinsic trade-off in choosing the adjustment
set $\bm V_{\mathcal{B}}$: by choosing a
larger subset $\mathcal{B}$, the second condition is more likely but
the first condition is less likely to be satisfied. The reason is
that, when conditioning on more genetic variants, we are more likely
to block the pleiotropic paths to $Y$ but we are also more likely to
block the path between the instrument and the causal variant.

For now, we will continue our discussion under the assumption that an
appropriate set $\mathcal{B}$ can be chosen. In
\Cref{subsec:hmm-simplification}, we will return to this and describe
a simple construction of $\mathcal{B}$ using
Markov properties in Haldane's model for meiosis.


\subsection{Hypothesis testing}
\label{subsec:hypothesis-test}

We are now ready to describe the randomization-based inference for
within-family Mendelian randomization. 
We will focus on the simplest case where a single genetic
variant from the offspring's maternally-inherited haplotype is used as
an instrumental variable and defer the discussion on multiple instruments
to \Cref{subsec:multiple-ivs}.

Suppose we observe $N$ trios of parent and offspring and within each
trio the observed variables can be described by the causal diagram
described in \Cref{sec:setup}. Let $i \in \{1,\dotsc,N\}$ be the index
of the trio.  Following \citet{Rosenbaum1983}, we define the
\emph{propensity score} of the instrument $Z_{ij}^m$ at locus $j$ of
individual $i$ as
\begin{equation} \label{eqn:exact-propensity-score}
\pi_{ij}^m = \mathbb{P}(Z_{ij}^m = 1 \mid \bm{M}_{ij}^{mf}, \bm{V}_{i\mathcal{B}}^m)
\end{equation}
where $\mathcal{B} \subseteq \mathcal{J}$ is an appropriately chosen
set of loci that satisfies the conditions in \Cref{thm:sufficient-adjustment-set}. In words, $\pi_{ij}^m$ describes the
randomization distribution of the haplotype $Z_{ij}^m$ conditional on
a set of parental and offspring haplotypes or genotypes. For now, we
will treat $\pi_{ij}^m$ as known.

Let us consider the following model for the potential outcomes
that assumes a constant treatment effect $\beta$:
\begin{equation}
Y_i(d) = Y_i(0) + \beta d ~ \text{for all $d \in \mathcal{D}$ and $i = 1, \ldots, N$}.
\label{eqn:potential-outcomes}
\end{equation}
Let $\mathcal{F} = \{Y_i(0) \colon i = 1, \ldots, N\}$ denote the
collection of potential outcomes for all individuals $i$ under no
exposure $d = 0$. Our goal is to test null hypotheses of the form
\begin{equation} \label{eqn:hypothesis-test}
H_0 \colon \beta = \beta_0, ~ H_1 \colon \beta \neq \beta_0
\end{equation}
where $\beta_0$ is some hypothetical value of the causal effect. If
the null hypothesis is true, equation \eqref{eqn:potential-outcomes}
implies that the potential outcome under no exposure ($d = 0$) can be
identified from the observed data by
\[
Y_i(0) = Y_i(D_i) - \beta_0 D_i = Y_i - \beta_0 D_i.
\]
For ease of notation, let $Q_i(\beta_0) = Y_i - \beta_0 D_i$ be the adjusted outcome.

\Cref{thm:identification} and the model \eqref{eqn:potential-outcomes}
imply that we are essentially testing the following conditional
independence:
\begin{equation} \label{eqn:hypothesis-adjusted-outcome}
  \begin{split}
H_0 \colon Z_{ij}^m \independent Q_i(\beta_0) \mid \bm
    (\bm{M}_{ij}^{mf}, \bm{V}_{i\mathcal{B}}^m)\quad\text{vs.}\quad
    H_1 \colon Z_{ij}^m \notindependent Q_i(\beta_0) \mid
    (\bm{M}_{ij}^{mf}, \bm{V}_{i\mathcal{B}}^m).
  \end{split}
\end{equation}
Let $\bm Z_j^m = (Z_{1j}^m,\dotsc,Z_{Nj}^m)$ and similarly define
other vector-valued genotypes. Suppose we have selected a test
statistic $T(\bm{Z}_j^m \mid
\mathcal{F})$ whose dependence on $(\bm{M}_j^{mf},
\bm{V}_{\mathcal{B}}^m)$ is implicit. For example, this could be the
coefficient from a regression of the adjusted outcome on the
instrument. The randomization-based p-value for $H_0$ can then be
written as
\begin{equation}\begin{array}{ll} \label{eqn:randomization-pvalue}
                  P(\bm{Z}_j^m \mid \mathcal{F})
                  &= \tilde{\mathbb{P}}\big(T(\tilde{\bm{Z}}_j^m \mid \mathcal{F}) \leq T(\bm{Z}_j^m \mid \mathcal{F})\big) \\[1em]
&=\displaystyle \sum_{\tilde{\bm{z}}^m \in \{0,1\}^N}
  I\big\{T(\tilde{\bm{z}}_j^m \mid \mathcal{F}) \leq T(\bm{Z}_j^m \mid
  \mathcal{F})\big\} \prod_{i=1}^N (\pi_{ij}^m)^{\tilde{z}_i} (1-\pi_{ij}^m)^{1-\tilde{z}_i},
\end{array}\end{equation}
where $I\{\cdot\}$ is the indicator function, $\tilde{\bm{Z}}_j^m$
denotes an independent, random draw from
\eqref{eqn:exact-propensity-score} and $\tilde{\mathbb{P}}$ denotes
its probability distribution. Given the propensity score and the
null hypothesis, this p-value can be computed exactly by enumerating
over all $2^N$ possible values of $\tilde{\bm{Z}}^m$, or using a Monte
Carlo approximation by drawing a large sample of $\tilde{\bm{Z}}^m$ using
the propensity scores $\bm{\pi}_j^m$; see Algorithm \ref{algo:pvalue}
for the pseudo-code. It is
straightforward to replace the haplotype $Z_{ij}^m$ with the genotype
$Z_{ij}$; the randomization distribution of $Z_{ij} \in \{0,1,2\}$ is
a simple function of $\pi_{ij}^m$ and $\pi_{ij}^f$ since meioses in
the mother and father are independent.

\Cref{eqn:randomization-pvalue} highlights the importance of the vector
$\bm{\pi}_j^m$ of propensity scores in randomization
inference. However, $\bm{\pi}_j^m$ describes a biochemical
process occurring in the human body which is not precisely known to,
or controlled by, us. Therefore, the best we can do is
perform \emph{almost exact} inference by replacing $\bm{\pi}_j^m$ with
an accurate approximation. The model we use in this paper
is Haldane's hidden Markov model described in
\Cref{sec:randomization-distribution}. As discussed in
\Cref{sec:genetic-preliminaries} our method is modular in the sense
that more sophisticated meiosis models can easily be substituted as
the randomization distribution; see \citet{Broman2000,Otto2019} for
discussion and comparison of alternative models.

\begin{algorithm}[tb]
  \caption{Almost exact test}
  \label{algo:pvalue}
\SetAlgoLined
 Compute the test statistic on the observed data $t = T(\bm{Z}_j^m \mid \mathcal{F})$\;
 \For{$k = 1, \ldots, K$}{
  Sample a counterfactual instrument $\tilde{\bm{Z}}_j^m$ from the
  randomization distribution (e.g. using \Cref{thm:propensity-score} in
  \Cref{sec:randomization-distribution} based on Haldane's model)\;
  Compute the test statistic using the counterfactual instrument $\tilde{t}_k = T(\tilde{\bm{Z}}_j^m \mid \mathcal{F})$\;
 }
 Compute an approximation to the p-value in \Cref{eqn:randomization-pvalue} via the proportion of $\tilde{t}_1, \ldots, \tilde{t}_K$ which are larger than $t$:
 \[ \hat{P}(\bm{Z}_j^m \mid \mathcal{F}) = \frac{|\{k \colon t \leq \tilde{t}_k \}|}{K}. \]
\end{algorithm}

\subsection{Choice of test statistic} \label{subsec:test-statistic}

Our randomization test retains the nominal size under the null
hypothesis, regardless of the choice of the test statistic.
Nonetheless, a well chosen statistic may substantially increase the
power of the test. A practical challenge is
that the adjustment set $(\bm{M}_j^{mf}, \bm{V}_{\mathcal{B}})$ may be
high-dimensional and highly correlated, and their role in designing
the test statistic is unclear. We propose to
use the following ``clever covariate'' in the test
statistic:
\[
X_{ij}^m = \frac{Z_{ij}^m}{\pi_{ij}^m} - \frac{1-Z_{ij}^m}{1-\pi_{ij}^m}.
\]
We may then use the weighted difference-in-means statistic
\[
T(\bm{Z}_j^m \mid \mathcal{F}) = \sum_{i=1}^N Q_i(\beta_0) X_{ij}^m
\]
or the $F$-statistic in a linear regression of $Q_i(\beta_0)$ on
$Z_{ij}^m$ and $X_{ij}^m$. A simulation example in
\Cref{sec:simulation-test-statistics} shows that using this clever
covariate can increase the power of the test dramatically.

The idea of using a ``clever covariate'' is proposed in
\citet{scharfstein1999adjusting} and \citet{Rose2008} and is commonly
used in semiparametric estimators of the average treatment
effect. Heuristically, the clever covariate exploits
the so-called ``central role'' of the propensity score
\citep{Rosenbaum1983}: $Y_i(d) \independent Z_{ij}^m \mid \pi_{ij}^m$,
provided that $0 < \pi_{ij}^m < 1$. Thus, the propensity score $\pi_{ij}^m$ may be viewed as a
one-dimensional summary of the sufficient adjustment set $(\bm{M}_{ij}^{mf},
\bm{V}_{i\mathcal{B}})$ and is particularly convenient here because it
can be directly computed from a meiosis model.


\subsection{Simplification via Markovian
structure} \label{subsec:hmm-simplification}

Thus far, we have sidestepped the issue of choosing the adjustment set
$\bm V_{\mathcal{B}}$ and computing the propensity score. Conditional
independencies implied by Haldane's meiosis model allows us to greatly
simplify the sufficient confounder adjustment set. We explain this
below.

The conditions in \Cref{thm:sufficient-adjustment-set} are trivially
satisfied with $\mathcal{B} = \emptyset$ if $\mathcal{J}_y =
\emptyset$ and $\mathcal{J}_d \neq \emptyset$, i.e., all causal
variants of $Y$ on this chromosome only affect $Y$ through
$D$. However, this is a rather unlikely situation. More often, we
need to condition on some variants to block the pleiotropic paths
(such as $Z_3$ in the working example in \Cref{fig:fammr_dag}). To this
end, we can utilize the Markovian structure on the meiosis indicators
$\bm U^m$ and $\bm U^f$ implied by Haldane's model. Roughly speaking, such structure implies that
\[
  Z_j \independent Z_l \mid (\bm M_j^{mf}, \bm F_j^{mf}, \bm
  V_{k})~\text{for all}~j < k < l,
\]
if there are no mutations and mother's genotype at locus $k$ is
heterozygous (i.e.\ $M_k^f \neq M_k^m$).


More generally, let $b_1$ and $b_2$ ($b_1 < j < b_2$) be two
heterozygous loci in
the mother's genome, i.e., $M_{b_1}^f \neq M_{b_1}^m$ and
$M_{b_2}^f \neq M_{b_2}^m$. Let $\mathcal{A} =
\{b_1+1,\dotsc,b_2-1\}$ be the set of loci between $b_1$ and $b_2$,
which of course contains the locus $j$ of interest.

\begin{theorem} \label{thm:identification}
Consider the setting in \Cref{thm:sufficient-adjustment-set} and
suppose
\begin{enumerate}
\item The meiosis indicator process is a Markov chain so that $U^m_j
 \independent U^m_l \mid U^m_k$ for all $j < k < l$;
\item There are no mutations: $\epsilon = 0$.
\end{enumerate}
Then
$Z_j^m$ is a valid instrumental variable conditional on $(\bm
M_j^{mf}, \bm V^m_{\{b_1,b_2\}})$ if the following conditions are
true: 3.\ $\mathcal{A} \cap \mathcal{J}_d \neq \emptyset$; 4.\
$\mathcal{A} \cap \mathcal{J}_y = \emptyset$.
\end{theorem}
\begin{proof}
Because there are no mutations and $M_{b_1}$ and $M_{b_2}$ are
heterozygous, we can uniquely determine $U^m_{b_1}$ and $U^m_{b_2}$
from $\bm V_{\{b_1,b_2\}}^m$. By the assumed Markovian structure,
this means that
\[
 Z_j^m \independent Z_l^m \mid \bm M_j^{mf}, \bm
 V^m_{\{b_1,b_2\}}~\text{for all}~j < b_1~\text{or}~j > b_2.
\]
Thus, the last two conditions in \Cref{thm:identification} imply the
two conditions in \Cref{thm:sufficient-adjustment-set}.
\end{proof}

One can easily mirror the above result for using the paternal haplotype $Z_j^f$ as an instrument variable. Furthermore, let $b_1'$ and $b_2'$ ($b_1' < j < b_2'$) be two heterozygous loci in the father's genome. Then it is easy to see that $Z_j = Z_j^m + Z_j^f$ is a valid instrument conditional on $(\bm
M_j^{mf}, \bm F_j^{mf}, \bm V^m_{\{b_1,b_2\}},\bm
V^f_{\{b_1',b_2'\}})$ if the last two conditions hold for the
union $\mathcal{A} = \{\min(b_1,
b_1')+1,\dotsc,\max(b_2,b_2')-1\}$.


Under the setting in \Cref{thm:identification}, we can partition
the offspring genome into mutually independent subsets by conditioning
on heterozygous parental genotypes. This partition is useful for
constructing independent p-values when we have multiple
instruments. Suppose we have a collection of genomic position
$\mathcal{B} = \{b_1, \ldots, b_k\}$ that will be conditioned on and
let $\mathcal{A}_k = \{b_{k-1}+1, \ldots, b_k-1 \}$ be the loci in
between (suppose $b_0 = 0$ and $b_{k+1} = p+1$). This indues the
following partition of the chromosome:
\begin{equation}
  \label{eq:genome-partition}
  \mathcal{J} = \mathcal{A}_1 \cup \{b_1\} \cup \mathcal{A}_2 \cup
  \{b_2\} \cup  \ldots \cup \mathcal{A}_k \cup \{b_k\} \cup
  \mathcal{A}_{k+1}.
\end{equation}

\begin{proposition} \label{prop:independent-instruments}
Suppose $M_j^m \neq M_j^f$ for all $j \in \mathcal{B}$. Then, under
the first two assumptions in \Cref{thm:identification}, we have
\[
    Z_j^m \independent Z_{j^{\prime}}^m \mid (\bm{M}_j^{mf}, \bm{M}_{j^{\prime}}^{mf}, \bm{V}^m_{\mathcal{B}}).
\]
for any $j \in \mathcal{A}_{l}$ and $j^{\prime} \in \mathcal{A}_{l^{\prime}}$ such that $l \neq l^{\prime}$.
\end{proposition}
\begin{proof}
  The proof follows from an almost identical argument to
  \Cref{thm:sufficient-adjustment-set}. The assumption that $\epsilon
  = 0$ means that $U_j^m$ is uniquely determined for all $j \in
  \mathcal{B}$ from $\bm{M}^{mf}_j$ and $Z^m_j$. Therefore the assumed
  Markovian structure implies that conditioning on
  $\bm{V}_{\mathcal{B}}^m$, along with the parental haplotypes
  $\bm{M}_j^{mf}$ and $\bm{M}_{j^{\prime}}^{mf}$, then induces the
  conditional independence.
\end{proof}

\subsection{Multiple instruments} \label{subsec:multiple-ivs}

\Cref{prop:independent-instruments} allows us to formalize the
intuition that genetic instruments across the genome may provide
independent pieces of evidence.
\begin{corollary} \label{corr:independent-tests}
  Suppose $j \in \mathcal{A}_{l}$, $j^{\prime} \in
  \mathcal{A}_{l^{\prime}}$, and $l \neq l^{\prime}$.
  Then $Z_j^m$ and $Z_{j^{\prime}}^m$ are independent, valid
  instruments given $(\bm{M}_j^{mf},
  \bm{M}_{j^{\prime}}^{mf}, \bm{V}^m_{\mathcal{B}})$ if
\begin{enumerate}
\item The first two assumptions of \Cref{thm:identification} hold;
\item $\mathcal{A}_l \cap \mathcal{J}_d \neq 0$ and $\mathcal{A}_{l^{\prime}} \cap \mathcal{J}_d \neq 0$;
\item $\mathcal{A}_l \cap \mathcal{J}_y = 0$ and $\mathcal{A}_{l^{\prime}} \cap \mathcal{J}_y = 0$.
\end{enumerate}

\end{corollary}

\Cref{corr:independent-tests} says that any two instruments are valid
and independent if they lie in separate regions in the partition
\eqref{eq:genome-partition} and each region contains at least one causal
variant of the exposure and no pleiotropic variants (i.e.\ variants with a
direct effect on $Y$ not mediated by $D$).

As a direct application of this corollary, we can use standard
procedures to combine randomization p-values obtained from different
genetic instruments and test the null hypothesis
\citep[see e.g.][]{bretz16_multip_compar_using_r}. One such procedure
is Fisher's method \citep{Fisher1925}: let $\{ p_1, p_2, \ldots, p_k
\}$ be a collection of independent p-values, then $-2 \sum_{j = 1}^k
\log(p_j) \sim \chi_{2k}^2$ when all of the corresponding null
hypotheses are true. We will use this method to aggregate p-values in
the applied example in \Cref{sec:applied-example}.

As some instruments may violate the exclusion restriction, a more
robust approach is to test the partial conjunction of the null hypotheses
\citep{benjamini2008screening,Wang2019}; loosely speaking, this means
that we reject the causal null hypothesis only if quite a few genetic
instruments appear to provide evidence against it.

As a final remark, it may be impossible in practice to separate
closely linked instruments into partitions
separated by a heterozygous variant, in which case the hypothesis
\eqref{eqn:hypothesis-adjusted-outcome} can be tested using $(Z_j^m, Z_{j^{\prime}}^m)$
jointly. \Cref{cor:multiple-instruments} in \Cref{sec:technical-proofs} derives the joint randomization distribution of a collection of instruments.

\section{Simulation} \label{sec:simulation}

\subsection{Setup and illustration} \label{sec:simulation-setup}

In this section we explore the properties of our almost exact test via
a numerical simulation. The set up of the simulation is described in
detail in \Cref{sec:simulation-description}. To summarize, we generate
a dataset of 15,000 parent-offspring trios with
a null effect of an exposure on an outcome (i.e. $\beta = 0$), both of
which have variance one, and consider 5 genetic instruments on a
chromosome with $p=150$ loci.  
The instruments are non-causal markers for nearby causal
variants and there are also pleiotropic variants in linkage
disequilibrium with the instruments.

 \subsection{Power} \label{sec:simulation-test-statistics}

 We now use the simulation to study the power of the randomization
 test, assuming a correct adjustment set is used (see
\Cref{eqn:sim-adjustment-set} in
\Cref{sec:simulation-description}). As the haplotypes are simulated
according to Haldane's meiosis model, the randomization test should be
exact. This is verified by the near-uniform distributions of the
p-values for testing the correct null hypothesis $H_0:\beta = 0$ with
three different test statistics in the top panels of
\Cref{fig:p-value-histogram}.

The histograms in the bottom panels of \Cref{fig:p-value-histogram}
depict the distribution of p-values for a test of a false null hypothesis $H_0: \beta = 0.5$. The power of the test varies
significantly according to the choices of test statistic. The simple
$F$-statistic based on a linear regression of
the adjusted outcome on the instruments (test statistic 1) has almost
no power, while the test statistic obtained from the same model but with the propensity score included as a clever covariate (test statistic 2) has a reasonable power of about 0.52.

\Cref{fig:power-curves} expands upon the previous figure by plotting a
power curve for each test statistic. We can see that test statistic
1 has almost no power between $\beta_0 = 0$ and $\beta_0 =
1$. Since test statistic 1 is unconditional on the adjustment set, resampled offspring haplotypes retain their correlation with the confounders via the parental haplotypes. This can cause under-rejection of false null hypotheses around the unconditional instrumental variable estimator. In this simulation, the Anderson-Rubin 95\% confidence interval is 0.64--0.89, which aligns with the region of under-rejection.

Test statistic 2, on the other hand, conditions on the confounders via a clever covariate. It has a power curve that is centred
on the true null $\beta_0 = 0$ and has significantly improved power in
the region between $\beta_0 = 0$ and $\beta_0 = 1$. However, it should be noted that test statistic 2 is not uniformly more
powerful than test statistic 1.

 \begin{figure*}[p]
   \centering
   \begin{subfigure}[b]{0.32\textwidth}
     \centering
     \includegraphics[width=\textwidth]{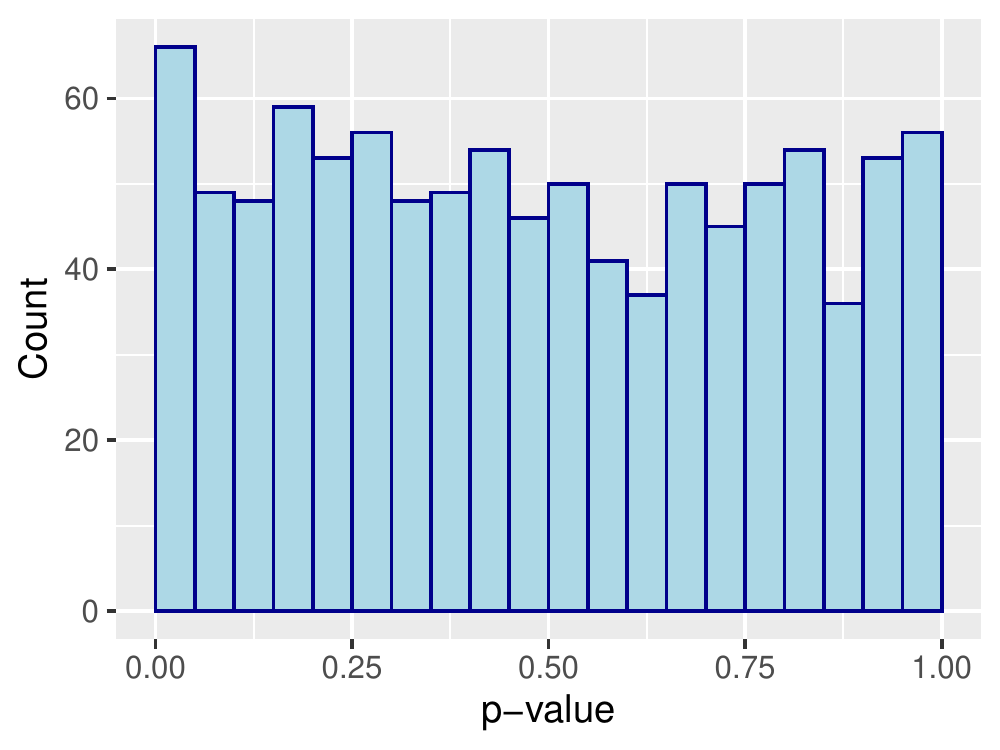}
     \caption{$H_0: \beta = 0$ (true) and test statistic 1.}
     \label{fig:null-1-test-statistic-1}
   \end{subfigure}
   ~
   \begin{subfigure}[b]{0.32\textwidth}
     \centering
     \includegraphics[width=\textwidth]{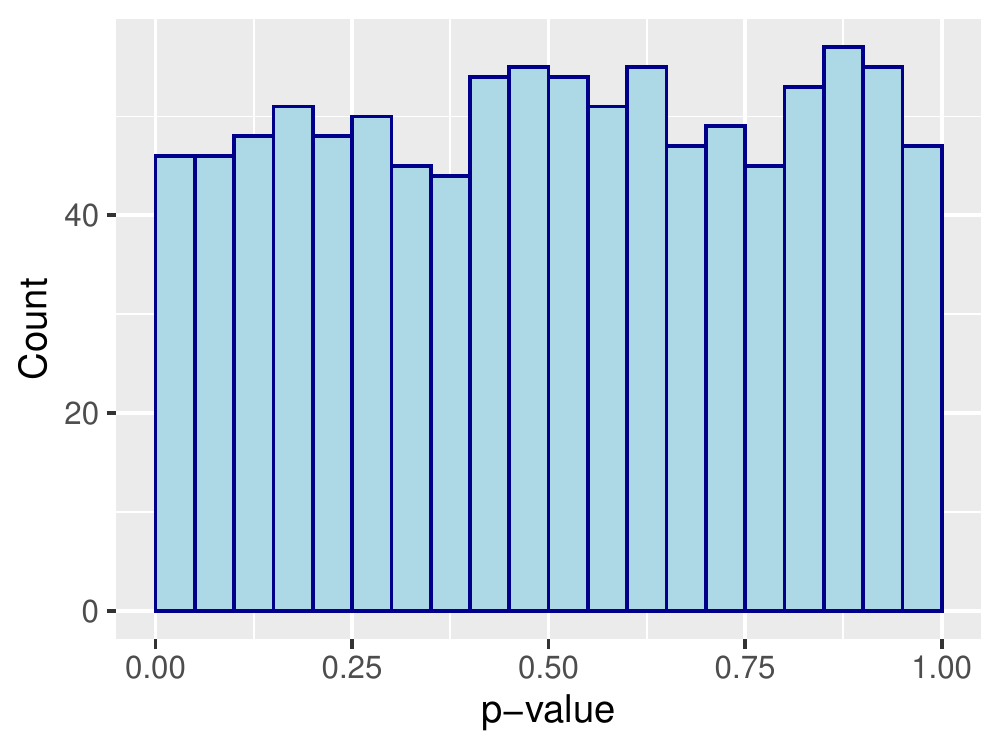}
     \caption{$H_0: \beta = 0$ (true) and test statistic 2.}
     \label{fig:null-1-test-statistic-2}
   \end{subfigure}
   ~
   \begin{subfigure}[b]{0.32\textwidth}
     \centering
     \includegraphics[width=\textwidth]{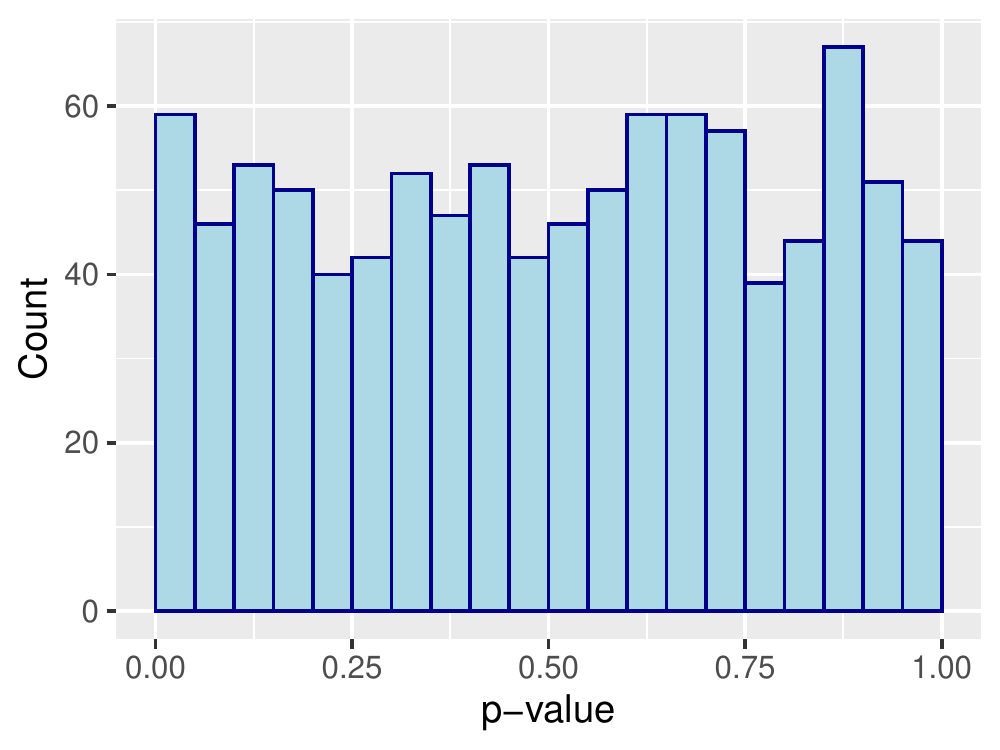}
     \caption{$H_0: \beta = 0$ (true) and test statistic 3.}
     \label{fig:null-1-test-statistic-3}
   \end{subfigure}

   \par\bigskip

   \begin{subfigure}[b]{0.32\textwidth}
     \centering
     \includegraphics[width=\textwidth]{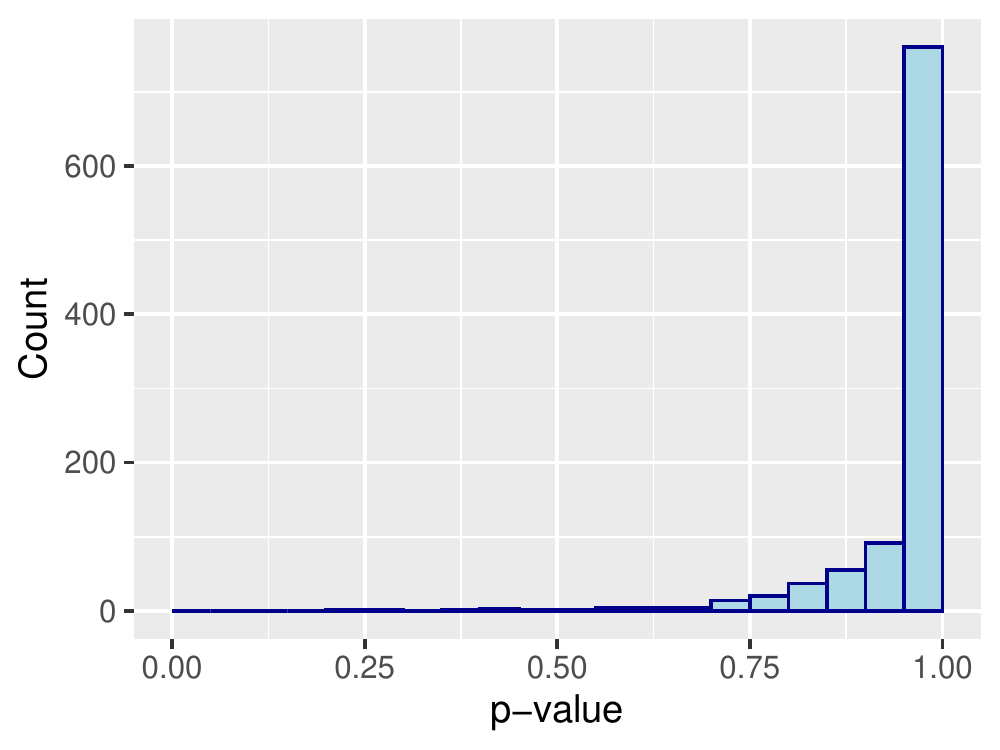}
     \caption{$H_0: \beta = 0.5$ (false) and test statistic 1.}
     \label{fig:null-2-test-statistic-1}
   \end{subfigure}%
   ~
   \begin{subfigure}[b]{0.32\textwidth}
     \centering
     \includegraphics[width=\textwidth]{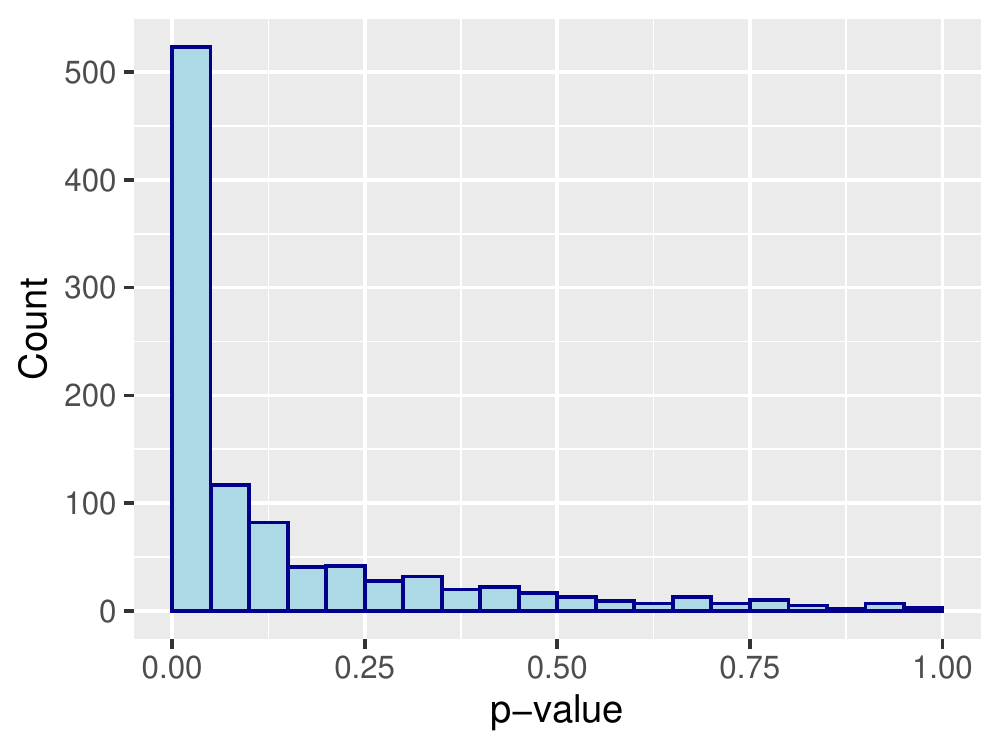}
     \caption{$H_0: \beta = 0.5$ (false) and test statistic 2.}
     \label{fig:null-2-test-statistic-2}
   \end{subfigure}%
   ~
   \begin{subfigure}[b]{0.32\textwidth}
     \centering
     \includegraphics[width=\textwidth]{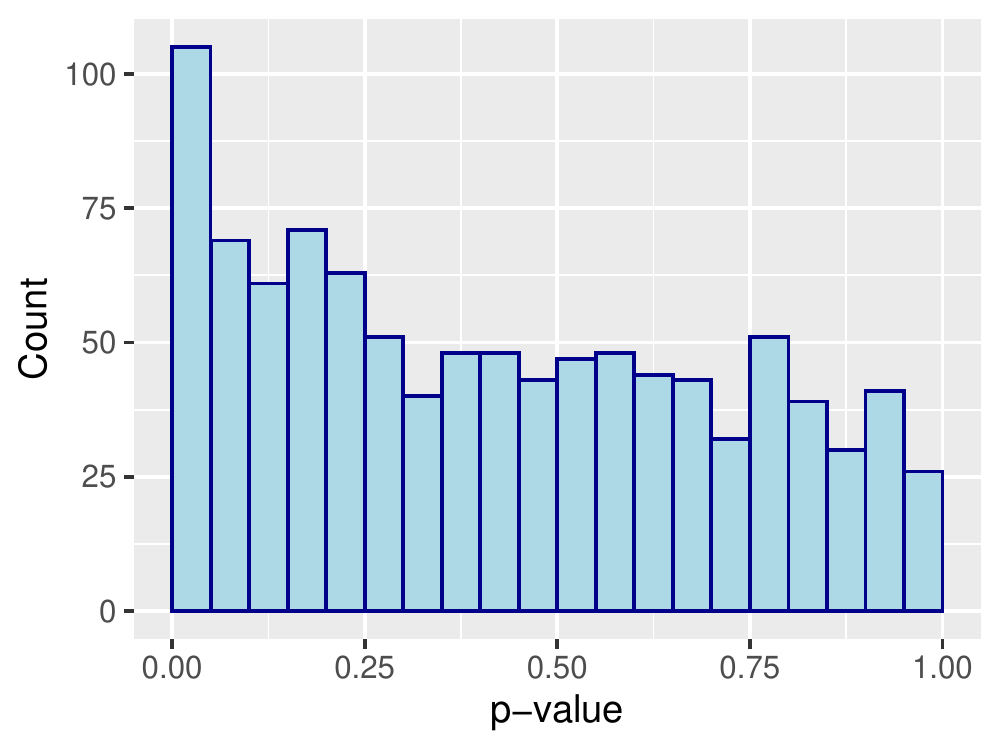}
     \caption{$H_0: \beta = 0.5$ (false) and test statistic 3.}
     \label{fig:null-2-test-statistic-3}
   \end{subfigure}%
   \caption{Histograms of 1,000 p-values for several null hypotheses
     and test statistics. Test statistic 1 is the F-statistic from a linear regression of the
       adjusted outcome on the instruments. Test statistic 2 includes
       the propensity scores for each instrument
       as covariates. Test statistic 3 includes the parental haplotypes as covariates.} \label{fig:p-value-histogram}
 \end{figure*}

  \begin{figure*}[hp]
   \begin{center}
     \includegraphics[width = 0.9\textwidth]{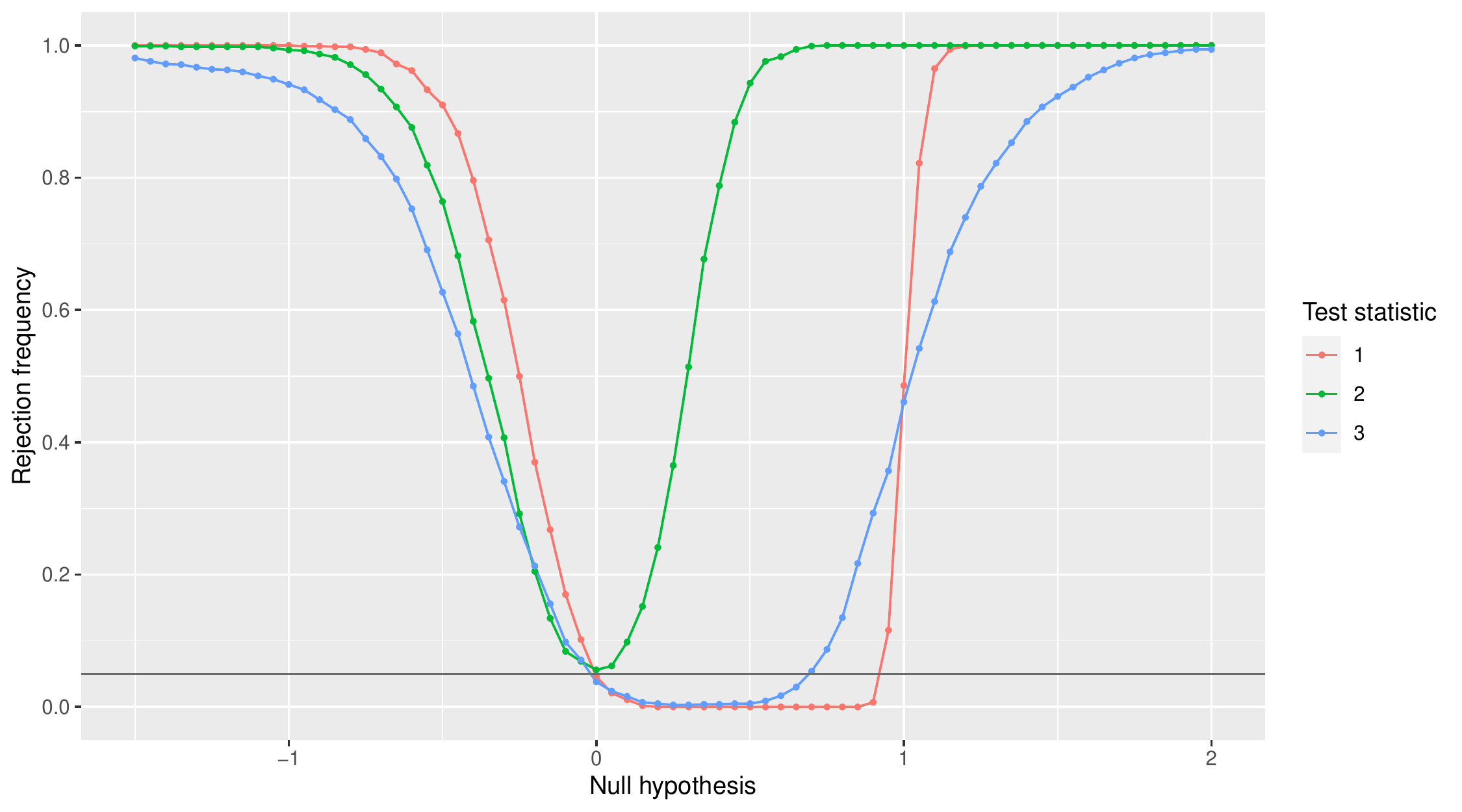}
     \caption{Power curves for the three choices of test statistic. Test statistic 1 is the F-statistic from a linear regression of the
       adjusted outcome on the instruments. Test statistic 2 includes
       the propensity scores for each instrument
       as covariates. Test statistic 3 includes the parental haplotypes as covariates. Each point on the figure is the rejection
       frequency over 1,000 replications.}
     \label{fig:power-curves}
   \end{center}
 \end{figure*}


 \section{Applied example} \label{sec:applied-example}
 \subsection{Preliminaries} \label{sec:applied-example-preliminaries}
We illustrate our approach with a pair of negative and positive
controls using the Avon Longitudinal Study of Parents and Children
(ALSPAC). Our dataset consists of 6,222 mother-child duos from ALSPAC,
a longitudinal cohort initially comprising pregnant
women resident in Avon, UK with expected dates of delivery from 1
April 1991 to 31 December 1992. The initial sample consisted of 14,676
fetuses, resulting in 14,062 live births and 13,988 children who were
alive at 1 year of age. In subsequent years, mothers, children and
occasionally partners attended several waves of questionnaires and
clinic visits, including genotyping. For a more thorough cohort
description, see \cite{Boyd2013} and \cite{Fraser2013}.\footnote{Please note
that the study website contains details of all the data that is
available through a fully searchable data dictionary and variable
search tool
(\url{https://www.bristol.ac.uk/alspac/researchers/our-data/}).}

The negative control is the effect of child's BMI at age 7 on mother's
BMI pre-pregnancy. Dynastic effects, as depicted in
\Cref{fig:dynasticeffects}, could induce a spurious correlation
between child's BMI-associated variants and their mother's BMI
pre-pregnancy. Blocking this backdoor path is crucial for reliable
causal inference. The positive control is the effect of child's BMI at
age 7 on a simulated, noisy version of itself. We vary the proportion
of the outcome that is attributable to noise to assess the power of
our test.

  \subsection{Data processing} \label{sec:applied-example-data-processing}

We use ALSPAC genotype data generated using the Illumina HumanHap550 chip (for children) and Illumina human660W chip (for mothers) and imputed to the 1000 Genomes reference panel. We remove SNPs with missingness of more than 5\% and minor allele frequency of less than 1\%. Haplotypes are phased using the \texttt{SHAPEIT2} software with the \texttt{duoHMM} flag, which ensures that phased haplotypes are consistent with known pedigrees in the sample. We obtain recombination probabilities from the 1000 Genomes genetic map file on Genome Reference Consortium Human Build 37.

 Our instruments are selected from the genome-wide association study
 (GWAS) of \cite{Vogelezang2020}, which identifies 25 genetic variants
 for childhood BMI, including 2 novel loci located close to
 \emph{NEDD4L} and \emph{SLC45A3}. Of the genome-wide significant
 variants in the discovery sample, we select 11 with a p-value of less
 than $0.001$ in the replication sample. ALSPAC is included in the
 discovery sample, so independent replication is important for
 avoiding spurious associations with the exposure. Two of our
 instruments, rs571312 and rs76227980, are located close together near
 \emph{MC4R} and need to be tested jointly. We exclude rs62107261
 because it is not contained in the 1000 Genomes genetic map
 file. We condition on all variants outside of a 500 kilobase window around each instrument.

 \subsection{Results} \label{sec:applied-example-results}

 \Cref{tab:negative-control-results,tab:positive-control-results} show
 results for the negative and positive controls, respectively. The
 last row of each table shows the p-value from Fisher's method
 aggregated across all independent p-values. The aggregated p-value
 for the negative control is 0.21, indicating little evidence against
 the null. The aggregated p-values for the positive control range from
 0.03 (when 10\% of the variance of the simulated outcome is noise) to 0.16 (when 50\%
 of the variance of the simulated outcome is noise). This indicates weak evidence
 against the null even when the effect is quite strong. 

  \begin{table}[t]
   \caption{Results from the ALSPAC negative control example.}
   \begin{center}
     \begin{tabular}{lllc}
       \toprule
       Instrument (rsID) & Chromosome & Proximal gene & P-value \\
       \midrule
       rs11676272 & 2 & \emph{ADCY3} & 0.45 \\
       rs7138803 & 12 & \emph{BCDIN3D} & 0.55 \\
       rs939584 & 2 & \emph{TMEM18} & 0.39 \\
       rs17817449 & 16 & \emph{FTO} & 0.06 \\
       rs12042908 & 1 & \emph{TNNI3K} & 0.35 \\
       rs543874 & 1 & \emph{SEC16B} & 0.07 \\
       rs56133711 & 11 & \emph{BDNF} & 0.59 \\
       rs571312, rs76227980 & 18 & \emph{MC4R} & 0.48 \\
       rs12641981 & 4 & \emph{GNPDA2} & 0.62 \\
       rs1094647 & 1 & \emph{SLC45A3} & 0.19 \\
       \midrule
       \multicolumn{3}{l}{Fisher's method} & 0.21 \\
       \bottomrule
     \end{tabular}
     \label{tab:negative-control-results}
   \end{center}
 \end{table}
   \begin{table}[t]
   \caption{Results from the ALSPAC positive control example. (Chr.\ =
     chromosome)}
   \begin{center}
     \begin{tabular}{lllccc}
       \toprule
       Instrument (rsID) & Chr.\ & Gene &  \multicolumn{3}{c}{P-value for noise of} \\
       & & & 10\% & 20\% & 50\% \\
       \midrule
       rs11676272 & 2 & \emph{ADCY3} & 0.01 & 0.01 & 0.01 \\
       rs7138803 & 12 & \emph{BCDIN3D} & 0.01 & 0.01 & 0.01 \\
       rs939584 & 2 & \emph{TMEM18} & 0.98 & 0.95 & 0.88 \\
       rs17817449 & 16 & \emph{FTO} & 0.33 & 0.35 & 0.44 \\
       rs12042908 & 1 & \emph{TNNI3K} & 0.77 & 0.79 & 0.85 \\
       rs543874 & 1 & \emph{SEC16B} & 0.48 & 0.64 & 0.92 \\
       rs56133711 & 11 & \emph{BDNF} & 0.12 & 0.14 & 0.25 \\
       rs571312, rs76227980 & 18 & \emph{MC4R} & 0.31 & 0.39 & 0.63 \\
       rs12641981 & 4 & \emph{GNPDA2} & 0.49 & 0.56 & 0.76 \\
       rs1094647 & 1 & \emph{SLC45A3} & 0.23 & 0.25 & 0.35 \\
       \midrule
       \multicolumn{3}{l}{Fisher's method} & 0.03 & 0.05 & 0.16 \\
       \bottomrule
     \end{tabular}
     \label{tab:positive-control-results}
   \end{center}
 \end{table}

 We can also compare the results in Tables
 \ref{tab:negative-control-results} and
 \ref{tab:positive-control-results} with per-instrument p-values obtained from two-stage least squares (2SLS) using the same offspring haplotypes as instruments, unconditional on parental or other offspring
 haplotypes. For the negative control, the p-value from Fisher's
 method is 0.02, indicating some evidence against the null. This is
 expected, given that the backdoor paths remain unblocked. For the positive control, the
 p-values from Fisher's method range from less than $10^{-20}$ (when
 10\% of the variance of the simulated outcome is noise) to $4.5 \times 10^{-11}$
 (when 50\% of the variance of the simulated outcome is noise). This indicates that
 the unconditional analysis has significantly more power to detect
 non-zero effects compared to our ``almost exact'' test.  We discuss
 potential reasons for, and implications of, this low power in
 \Cref{sec:discussion}

 \section{Discussion} \label{sec:discussion}

 We have presented an almost exact approach to within-family MR, which
 has a number of conceptual and practical advantages over model-based
 approaches to MR using population GWAS data. However, the applied
 example in \Cref{sec:applied-example} demonstrates that power may be
 limited relative to conventional MR analyses in
 unrelated individuals. Since our test leverages the precise amount
 of information available in a single meiosis, this suggests that
 MR in unrelated individuals is drawing power
 from elsewhere.

 Besides the obvious distinction that conventional MR analyses are
 model-based (and thus are not robust to model misspecification),
 another likely reason for the large difference in the empirical
 results is that MR in unrelated individuals use randomness in meioses
 across many generations. For example, an offspring with parents who are
 homozygous for the non-effect allele offers no power in our test,
 since their genotype will not vary across meioses. However, if we
 assume that genotypes
 are randomly distributed at the population level (as in MR studies
 with unrelated individuals), that same offspring can act as a
 comparator for individuals with the effect
 allele. \cite{Brumpton2020} corroborate
 this loss of power for their within-family method, but do not
 elaborate on the broader implications for how Mendelian randomization
 is typically justified. It would be extremely valuable for the MR
 literature to discuss the extent to which Mendelian inheritance
 across multiple generations is driving the power behind existing
 results, as such uncontrolled randomness may introduce bias when
 there are strong dynastic effects and natural selection.

Continuing the discussion on multiple instruments in
\Cref{subsec:multiple-ivs}, our approach closely resembles the usage
of evidence factors in observational studies as advocated by
\citet{rosenbaum2010evidence,rosenbaum2021replication} and \citet{zhao22_eviden_factor_from_multip_possib}. Using
(conditionally) independent instruments in different genomic regions
may also be viewed as a form of triangulation to improve causal
inference \citep{lawlor17_trian_aetiol_epidem}. Although all the
p-values are obtained using the same study design, different genetic
variants may influence the exposure through different biological
mechanisms and the fact that they provide corroborating evidence
strengthens the causal conclusion.

We must also return to the problem of transmission ratio distortion (TRD) discussed in \Cref{sec:genetic-preliminaries}. TRD violates the assumptions of our meiosis model that alleles are (unconditionally) passed from parents to offspring at the Mendelian rate of 50\%. We could represent TRD in our causal model in \Cref{fig:fammr_dag} via an arrow from the gametes $(\bm{Z}^m, \bm{Z}^f)$ to the mating indicator $S$. This indicates that the gametes themselves influence survival of their corresponding zygote to term. If our putative instrument $Z_1^m$ is in linkage with any variant exhibiting TRD, then this invalidates it as an instrument. Suppose $Z_3^m$ exhibits TRD, then this opens collider paths via the parental phenotypes $C^m$ and $C^f$, for example, $Y(d) \leftarrow C^m \rightarrow \boxed{S} \leftarrow Z_3^m \leftarrow \bm{U}^m \rightarrow Z_1^m$. The intuition is that parental phenotypes related to the likelihood of mating become associated with offspring variants related to the likelihood of offspring survival. Within our causal model, this pathway can be closed by conditioning on $Z_3^m$, with unconditioned variants obeying the meiosis model. If any unconditioned variants exhibit TRD, then this bias will remain and our meiosis model will incorrectly describe the inheritance patterns of any linked variants, resulting in an erroneous randomization distribution. Expanding resources of parent-offspring data may allow us to test the prevalence of transmission ratio distortion, which will help to inform the reasonableness of maintaining Mendel's First Law in our meiosis and fertilization model.

\begin{acks}[Acknowledgments]
The authors thank Kate Tilling, Rachael A Hughes, Jack Bowden, Gibran Hemani, Neil M Davies, Ben Brumpton, and Nianqiao Ju for their helpful feedback. In addition, the authors are extremely grateful to all the families who took part in the ALSPAC cohort, the midwives for their help in recruiting them, and the whole ALSPAC team, which includes interviewers, computer and laboratory technicians, clerical workers, research scientists, volunteers, managers, receptionists and nurses. Ethical approval for our applied example was obtained from the ALSPAC Ethics and Law Committee and the Local Research Ethics Committees. This publication is the work of the authors who will serve as guarantors for the contents of this paper. For the purpose of Open Access, the authors have applied a CC BY public copyright licence to any Author Accepted Manuscript version arising from this submission.
\end{acks}

\begin{funding}
  This research was supported in part by the Wellcome Trust (grant
  number 220067/Z/20/Z) and EPSRC (grant number EP/V049968/1). The UK
  Medical Research Council and Wellcome (grant number 217065/Z/19/Z)
  and the University of Bristol provide core support for ALSPAC. GWAS
  data was generated by Sample Logistics and Genotyping Facilities at
  Wellcome Sanger Institute and LabCorp (Laboratory Corporation of
  America) using support from 23andMe.
\end{funding}


\bibliographystyle{imsart-nameyear}
\bibliography{fammr_bib}

\begin{thebibliography}{77}

\bibitem[\protect\citeauthoryear{Acuna-Hidalgo, Veltman and
  Hoischen}{2016}]{Acuna-Hidalgo2016}
\begin{barticle}[author]
\bauthor{\bsnm{Acuna-Hidalgo},~\bfnm{Rocio}\binits{R.}},
  \bauthor{\bsnm{Veltman},~\bfnm{Joris~A.}\binits{J.~A.}} \AND
  \bauthor{\bsnm{Hoischen},~\bfnm{Alexander}\binits{A.}}
(\byear{2016}).
\btitle{{New insights into the generation and role of de novo mutations in
  health and disease}}.
\bjournal{Genome Biology}
\bvolume{17}
\bpages{1--19}.
\bdoi{10.1186/s13059-016-1110-1}
\end{barticle}
\endbibitem

\bibitem[\protect\citeauthoryear{Bates et~al.}{2020}]{Bates2020a}
\begin{barticle}[author]
\bauthor{\bsnm{Bates},~\bfnm{Stephen}\binits{S.}},
  \bauthor{\bsnm{Sesia},~\bfnm{Matteo}\binits{M.}},
  \bauthor{\bsnm{Sabatti},~\bfnm{Chiara}\binits{C.}} \AND
  \bauthor{\bsnm{Cand{\`{e}}s},~\bfnm{Emmanuel}\binits{E.}}
(\byear{2020}).
\btitle{{Causal inference in genetic trio studies}}.
\bjournal{Proceedings of the National Academy of Sciences of the United States
  of America}
\bvolume{117}
\bpages{24117--24126}.
\bdoi{10.1073/pnas.2007743117}
\end{barticle}
\endbibitem

\bibitem[\protect\citeauthoryear{Belmont et~al.}{2005}]{Belmont2005}
\begin{barticle}[author]
\bauthor{\bsnm{Belmont},~\bfnm{John~W.}\binits{J.~W.}},
  \bauthor{\bsnm{Boudreau},~\bfnm{Andrew}\binits{A.}},
  \bauthor{\bsnm{Leal},~\bfnm{Suzanne~M.}\binits{S.~M.}},
  \bauthor{\bsnm{Hardenbol},~\bfnm{Paul}\binits{P.}} \betal{et~al.}
(\byear{2005}).
\btitle{{A haplotype map of the human genome}}.
\bjournal{Nature}
\bvolume{437}
\bpages{1299--1320}.
\bdoi{10.1038/nature04226}
\end{barticle}
\endbibitem

\bibitem[\protect\citeauthoryear{Benjamini and
  Heller}{2008}]{benjamini2008screening}
\begin{barticle}[author]
\bauthor{\bsnm{Benjamini},~\bfnm{Yoav}\binits{Y.}} \AND
  \bauthor{\bsnm{Heller},~\bfnm{Ruth}\binits{R.}}
(\byear{2008}).
\btitle{Screening for partial conjunction hypotheses}.
\bjournal{Biometrics}
\bvolume{64}
\bpages{1215-1222}.
\end{barticle}
\endbibitem

\bibitem[\protect\citeauthoryear{Bherer, Campbell and Auton}{2017}]{Bherer2017}
\begin{barticle}[author]
\bauthor{\bsnm{Bherer},~\bfnm{Claude}\binits{C.}},
  \bauthor{\bsnm{Campbell},~\bfnm{Christopher~L.}\binits{C.~L.}} \AND
  \bauthor{\bsnm{Auton},~\bfnm{Adam}\binits{A.}}
(\byear{2017}).
\btitle{{Refined genetic maps reveal sexual dimorphism in human meiotic
  recombination at multiple scales}}.
\bjournal{Nature Communications}
\bvolume{8}.
\bdoi{10.1038/ncomms14994}
\end{barticle}
\endbibitem

\bibitem[\protect\citeauthoryear{Bowden, Davey~Smith and
  Burgess}{2015}]{Bowden2015}
\begin{barticle}[author]
\bauthor{\bsnm{Bowden},~\bfnm{Jack}\binits{J.}},
  \bauthor{\bsnm{Davey~Smith},~\bfnm{George}\binits{G.}} \AND
  \bauthor{\bsnm{Burgess},~\bfnm{Stephen}\binits{S.}}
(\byear{2015}).
\btitle{{Mendelian randomization with invalid instruments: effect estimation
  and bias detection through Egger regression}}.
\bjournal{International Journal of Epidemiology}
\bvolume{44}
\bpages{512--525}.
\bdoi{10.1093/ije/dyv080}
\end{barticle}
\endbibitem

\bibitem[\protect\citeauthoryear{Boyd et~al.}{2013}]{Boyd2013}
\begin{barticle}[author]
\bauthor{\bsnm{Boyd},~\bfnm{Andy}\binits{A.}},
  \bauthor{\bsnm{Golding},~\bfnm{Jean}\binits{J.}},
  \bauthor{\bsnm{Macleod},~\bfnm{John}\binits{J.}},
  \bauthor{\bsnm{Lawlor},~\bfnm{Deborah~A}\binits{D.~A.}},
  \bauthor{\bsnm{Fraser},~\bfnm{Abigail}\binits{A.}},
  \bauthor{\bsnm{Henderson},~\bfnm{John}\binits{J.}},
  \bauthor{\bsnm{Molloy},~\bfnm{Lynn}\binits{L.}},
  \bauthor{\bsnm{Ness},~\bfnm{Andy}\binits{A.}},
  \bauthor{\bsnm{Ring},~\bfnm{Susan}\binits{S.}} \AND
  \bauthor{\bsnm{Davey~Smith},~\bfnm{George.}\binits{G.}}
(\byear{2013}).
\btitle{{Cohort Profile: The `Children of the 90s'—the index offspring of the
  Avon Longitudinal Study of Parents and Children}}.
\bjournal{International Journal of Epidemiology}
\bvolume{42}
\bpages{111--127}.
\bdoi{10.1093/ije/dys064}
\end{barticle}
\endbibitem

\bibitem[\protect\citeauthoryear{Boyle, Li and Pritchard}{2017}]{Boyle2017}
\begin{barticle}[author]
\bauthor{\bsnm{Boyle},~\bfnm{Evan~A}\binits{E.~A.}},
  \bauthor{\bsnm{Li},~\bfnm{Yang~I}\binits{Y.~I.}} \AND
  \bauthor{\bsnm{Pritchard},~\bfnm{Jonathan~K}\binits{J.~K.}}
(\byear{2017}).
\btitle{{An expanded view of complex traits: From polygenic to omnigenic}}.
\bjournal{Cell}
\bvolume{169}
\bpages{1177--1186}.
\bdoi{10.1016/j.cell.2017.05.038.An}
\end{barticle}
\endbibitem

\bibitem[\protect\citeauthoryear{Bretz, Hothorn and
  Westfall}{2016}]{bretz16_multip_compar_using_r}
\begin{bbook}[author]
\bauthor{\bsnm{Bretz},~\bfnm{Frank}\binits{F.}},
  \bauthor{\bsnm{Hothorn},~\bfnm{Torsten}\binits{T.}} \AND
  \bauthor{\bsnm{Westfall},~\bfnm{Peter}\binits{P.}}
(\byear{2016}).
\btitle{Multiple Comparisons Using R}.
\bseries{[]}.
\bpublisher{Chapman and Hall/CRC}.
\bdoi{10.1201/9781420010909}
\end{bbook}
\endbibitem

\bibitem[\protect\citeauthoryear{Broman and Weber}{2000}]{Broman2000}
\begin{barticle}[author]
\bauthor{\bsnm{Broman},~\bfnm{Karl~W.}\binits{K.~W.}} \AND
  \bauthor{\bsnm{Weber},~\bfnm{James~L.}\binits{J.~L.}}
(\byear{2000}).
\btitle{{Characterization of human crossover interference}}.
\bjournal{American Journal of Human Genetics}
\bvolume{66}
\bpages{1911--1926}.
\bdoi{10.1086/302923}
\end{barticle}
\endbibitem

\bibitem[\protect\citeauthoryear{Brumpton et~al.}{2020}]{Brumpton2020}
\begin{barticle}[author]
\bauthor{\bsnm{Brumpton},~\bfnm{Ben}\binits{B.}},
  \bauthor{\bsnm{Sanderson},~\bfnm{Eleanor}\binits{E.}},
  \bauthor{\bsnm{Heilbron},~\bfnm{Karl}\binits{K.}},
  \bauthor{\bsnm{Hartwig},~\bfnm{Fernando~Pires}\binits{F.~P.}} \betal{et~al.}
(\byear{2020}).
\btitle{{Avoiding dynastic, assortative mating, and population stratification
  biases in Mendelian randomization through within-family analyses}}.
\bjournal{Nature Communications}
\bvolume{11}
\bpages{1--13}.
\bdoi{10.1038/s41467-020-17117-4}
\end{barticle}
\endbibitem

\bibitem[\protect\citeauthoryear{Cardon and Palmer}{2003}]{Cardon2003}
\begin{barticle}[author]
\bauthor{\bsnm{Cardon},~\bfnm{Lon~R.}\binits{L.~R.}} \AND
  \bauthor{\bsnm{Palmer},~\bfnm{Lyle~J.}\binits{L.~J.}}
(\byear{2003}).
\btitle{{Population stratification and spurious allelic association}}.
\bjournal{Lancet}
\bvolume{361}
\bpages{598--604}.
\bdoi{10.1016/S0140-6736(03)12520-2}
\end{barticle}
\endbibitem

\bibitem[\protect\citeauthoryear{Chen
  et~al.}{2008}]{chen08_alcoh_intak_blood_press}
\begin{barticle}[author]
\bauthor{\bsnm{Chen},~\bfnm{Lina}\binits{L.}},
  \bauthor{\bsnm{Smith},~\bfnm{George~Davey}\binits{G.~D.}},
  \bauthor{\bsnm{Harbord},~\bfnm{Roger~M}\binits{R.~M.}} \AND
  \bauthor{\bsnm{Lewis},~\bfnm{Sarah~J}\binits{S.~J.}}
(\byear{2008}).
\btitle{Alcohol Intake and Blood Pressure: a Systematic Review Implementing a
  Mendelian Randomization Approach}.
\bjournal{PLoS Medicine}
\bvolume{5}
\bpages{e52}.
\bdoi{10.1371/journal.pmed.0050052}
\end{barticle}
\endbibitem

\bibitem[\protect\citeauthoryear{{Davey Smith}}{2001}]{DaveySmith2001}
\begin{barticle}[author]
\bauthor{\bsnm{{Davey Smith}},~\bfnm{George}\binits{G.}}
(\byear{2001}).
\btitle{{Reflections on the limitations to epidemiology}}.
\bjournal{Journal of Clinical Epidemiology}
\bvolume{54}
\bpages{325--331}.
\bdoi{10.1016/S0895-4356(00)00334-6}
\end{barticle}
\endbibitem

\bibitem[\protect\citeauthoryear{Davey~Smith}{2007}]{smith07_capit_mendel_random_to_asses_effec_treat}
\begin{barticle}[author]
\bauthor{\bsnm{Davey~Smith},~\bfnm{George}\binits{G.}}
(\byear{2007}).
\btitle{Capitalizing on Mendelian Randomization To Assess the Effects of
  Treatments}.
\bjournal{Journal of the Royal Society of Medicine}
\bvolume{100}
\bpages{432-435}.
\bdoi{10.1177/014107680710000923}
\end{barticle}
\endbibitem

\bibitem[\protect\citeauthoryear{{Davey Smith} and
  Ebrahim}{2003}]{DaveySmith2003}
\begin{barticle}[author]
\bauthor{\bsnm{{Davey Smith}},~\bfnm{George}\binits{G.}} \AND
  \bauthor{\bsnm{Ebrahim},~\bfnm{Shah}\binits{S.}}
(\byear{2003}).
\btitle{{'Mendelian randomization': Can genetic epidemiology contribute to
  understanding environmental determinants of disease?}}
\bjournal{International Journal of Epidemiology}
\bvolume{32}
\bpages{1--22}.
\bdoi{10.1093/ije/dyg070}
\end{barticle}
\endbibitem

\bibitem[\protect\citeauthoryear{Davey~Smith et~al.}{2020}]{DaveySmith2020}
\begin{barticle}[author]
\bauthor{\bsnm{Davey~Smith},~\bfnm{George}\binits{G.}},
  \bauthor{\bsnm{Holmes},~\bfnm{Michael~V}\binits{M.~V.}},
  \bauthor{\bsnm{Davies},~\bfnm{Neil~M}\binits{N.~M.}} \AND
  \bauthor{\bsnm{Ebrahim},~\bfnm{Shah}\binits{S.}}
(\byear{2020}).
\btitle{{Mendel's laws, Mendelian randomization and causal inference in
  observational data: substantive and nomenclatural issues}}.
\bjournal{European Journal of Epidemiology}
\bvolume{35}
\bpages{99--111}.
\bdoi{10.1007/s10654-020-00622-7}
\end{barticle}
\endbibitem

\bibitem[\protect\citeauthoryear{Davies et~al.}{2019}]{Davies2019}
\begin{barticle}[author]
\bauthor{\bsnm{Davies},~\bfnm{Neil~M.}\binits{N.~M.}},
  \bauthor{\bsnm{Howe},~\bfnm{Laurence~J.}\binits{L.~J.}},
  \bauthor{\bsnm{Brumpton},~\bfnm{Ben}\binits{B.}},
  \bauthor{\bsnm{Havdahl},~\bfnm{Alexandra}\binits{A.}},
  \bauthor{\bsnm{Evans},~\bfnm{David~M.}\binits{D.~M.}} \AND
  \bauthor{\bsnm{{Davey Smith}},~\bfnm{George}\binits{G.}}
(\byear{2019}).
\btitle{{Within family Mendelian randomization studies}}.
\bjournal{Human Molecular Genetics}
\bvolume{28}
\bpages{R170--R179}.
\bdoi{10.1093/hmg/ddz204}
\end{barticle}
\endbibitem

\bibitem[\protect\citeauthoryear{Didelez and Sheehan}{2007}]{Didelez2007a}
\begin{barticle}[author]
\bauthor{\bsnm{Didelez},~\bfnm{Vanessa}\binits{V.}} \AND
  \bauthor{\bsnm{Sheehan},~\bfnm{Nuala}\binits{N.}}
(\byear{2007}).
\btitle{{Mendelian randomization as an instrumental variable approach to causal
  inference}}.
\bjournal{Statistical Methods in Medical Research}
\bvolume{16}
\bpages{309--330}.
\bdoi{10.1177/0962280206077743}
\end{barticle}
\endbibitem

\bibitem[\protect\citeauthoryear{Feinstein}{1988}]{Feinstein1988}
\begin{barticle}[author]
\bauthor{\bsnm{Feinstein},~\bfnm{Alvan~R}\binits{A.~R.}}
(\byear{1988}).
\btitle{{Scientific standards in epidemiologic studies of the menace of daily
  life}}.
\bjournal{Science}
\bvolume{242}
\bpages{1257--1263}.
\end{barticle}
\endbibitem

\bibitem[\protect\citeauthoryear{Fisher}{1918}]{fisher19_correlation}
\begin{barticle}[author]
\bauthor{\bsnm{Fisher},~\bfnm{R.~A.}\binits{R.~A.}}
(\byear{1918}).
\btitle{The Correlation Between Relatives on the Supposition of Mendelian
  Inheritance.}
\bjournal{Transactions of the Royal Society of Edinburgh}
\bvolume{52}
\bpages{399-433}.
\bdoi{10.1017/s0080456800012163}
\end{barticle}
\endbibitem

\bibitem[\protect\citeauthoryear{Fisher}{1925}]{Fisher1925}
\begin{bbook}[author]
\bauthor{\bsnm{Fisher},~\bfnm{Ronald~A.}\binits{R.~A.}}
(\byear{1925}).
\btitle{{Statistical methods for research workers}}.
\bpublisher{Oliver {\&} Boyd}, \baddress{Edinburgh}.
\end{bbook}
\endbibitem

\bibitem[\protect\citeauthoryear{Fisher}{1926}]{fisher26_arran_field_exper}
\begin{barticle}[author]
\bauthor{\bsnm{Fisher},~\bfnm{Ronald~Aylmer}\binits{R.~A.}}
(\byear{1926}).
\btitle{The Arrangement of Field Experiments}.
\bjournal{Journal of the Ministry of Agriculture}
\bvolume{33}
\bpages{503-513}.
\bdoi{10.23637/ROTHAMSTED.8V61Q}
\end{barticle}
\endbibitem

\bibitem[\protect\citeauthoryear{Fisher}{1935}]{Fisher1935}
\begin{bbook}[author]
\bauthor{\bsnm{Fisher},~\bfnm{Ronald~A.}\binits{R.~A.}}
(\byear{1935}).
\btitle{{The design of experiments}}.
\bpublisher{Oliver {\&} Boyd}, \baddress{Edinburgh}.
\end{bbook}
\endbibitem

\bibitem[\protect\citeauthoryear{Fisher}{1951}]{Fisher1951}
\begin{barticle}[author]
\bauthor{\bsnm{Fisher},~\bfnm{Ronald~A.}\binits{R.~A.}}
(\byear{1951}).
\btitle{{Statistical methods in genetics}}.
\bjournal{Heredity}
\bvolume{6}
\bpages{1--12}.
\bdoi{10.1093/ije/dyp379}
\end{barticle}
\endbibitem

\bibitem[\protect\citeauthoryear{Fraser et~al.}{2013}]{Fraser2013}
\begin{barticle}[author]
\bauthor{\bsnm{Fraser},~\bfnm{Abigail}\binits{A.}},
  \bauthor{\bsnm{Macdonald-Wallis},~\bfnm{Corrie}\binits{C.}},
  \bauthor{\bsnm{Tilling},~\bfnm{Kate}\binits{K.}},
  \bauthor{\bsnm{Boyd},~\bfnm{Andy}\binits{A.}},
  \bauthor{\bsnm{Golding},~\bfnm{Jean}\binits{J.}},
  \bauthor{\bsnm{Davey~Smith},~\bfnm{George}\binits{G.}},
  \bauthor{\bsnm{Henderson},~\bfnm{John}\binits{J.}},
  \bauthor{\bsnm{Macleod},~\bfnm{John}\binits{J.}},
  \bauthor{\bsnm{Molloy},~\bfnm{Lynn}\binits{L.}},
  \bauthor{\bsnm{Ness},~\bfnm{Andy}\binits{A.}},
  \bauthor{\bsnm{Ring},~\bfnm{Susan}\binits{S.}},
  \bauthor{\bsnm{Nelson},~\bfnm{Scott~M}\binits{S.~M.}} \AND
  \bauthor{\bsnm{Lawlor},~\bfnm{Debbie~A}\binits{D.~A.}}
(\byear{2013}).
\btitle{{Cohort Profile: The Avon Longitudinal Study of Parents and Children:
  ALSPAC mothers cohort}}.
\bjournal{International Journal of Epidemiology}
\bvolume{42}
\bpages{97--110}.
\bdoi{10.1093/ije/dys066}
\end{barticle}
\endbibitem

\bibitem[\protect\citeauthoryear{Gray and Wheatley}{1991}]{Gray1991}
\begin{barticle}[author]
\bauthor{\bsnm{Gray},~\bfnm{Richard}\binits{R.}} \AND
  \bauthor{\bsnm{Wheatley},~\bfnm{Keith}\binits{K.}}
(\byear{1991}).
\btitle{{How to avoid bias when comparing bone marrow transplantation with
  chemotherapy}}.
\bjournal{Bone Marrow Transplantation}
\bvolume{7}
\bpages{9--12}.
\end{barticle}
\endbibitem

\bibitem[\protect\citeauthoryear{Haldane}{1919}]{Haldane1919}
\begin{barticle}[author]
\bauthor{\bsnm{Haldane},~\bfnm{John B~S}\binits{J.~B.~S.}}
(\byear{1919}).
\btitle{{The combination of linkage values and the calculation of distances
  between the loci of linked factors}}.
\bjournal{Journal of Genetics}
\bvolume{8}
\bpages{299--309}.
\end{barticle}
\endbibitem

\bibitem[\protect\citeauthoryear{Hartwig, Davies and {Davey
  Smith}}{2018}]{Hartwig2018}
\begin{barticle}[author]
\bauthor{\bsnm{Hartwig},~\bfnm{Fernando~Pires}\binits{F.~P.}},
  \bauthor{\bsnm{Davies},~\bfnm{Neil~Martin}\binits{N.~M.}} \AND
  \bauthor{\bsnm{{Davey Smith}},~\bfnm{George}\binits{G.}}
(\byear{2018}).
\btitle{{Bias in Mendelian randomization due to assortative mating}}.
\bjournal{Genetic Epidemiology}
\bvolume{42}
\bpages{608--620}.
\bdoi{10.1002/gepi.22138}
\end{barticle}
\endbibitem

\bibitem[\protect\citeauthoryear{Heckman and Karapakula}{2019}]{Heckman2019}
\begin{barticle}[author]
\bauthor{\bsnm{Heckman},~\bfnm{James~J}\binits{J.~J.}} \AND
  \bauthor{\bsnm{Karapakula},~\bfnm{Ganesh}\binits{G.}}
(\byear{2019}).
\btitle{{The Perry Preschoolers at Late Midlife: A Study in Design-Specific
  Inference}}.
\bjournal{National Bureau of Economic Research Working Paper Series}
\bvolume{No. 25888}
\bpages{14--21}.
\end{barticle}
\endbibitem

\bibitem[\protect\citeauthoryear{Hemani, Bowden and {Davey
  Smith}}{2018}]{Hemani2018a}
\begin{barticle}[author]
\bauthor{\bsnm{Hemani},~\bfnm{Gibran}\binits{G.}},
  \bauthor{\bsnm{Bowden},~\bfnm{Jack}\binits{J.}} \AND \bauthor{\bsnm{{Davey
  Smith}},~\bfnm{George}\binits{G.}}
(\byear{2018}).
\btitle{{Evaluating the potential role of pleiotropy in Mendelian randomization
  studies}}.
\bjournal{Human Molecular Genetics}
\bvolume{27}
\bpages{R195--R208}.
\bdoi{10.1093/hmg/ddy163}
\end{barticle}
\endbibitem

\bibitem[\protect\citeauthoryear{Hern{\'{a}}n and Robins}{2020}]{Hernan2020}
\begin{bbook}[author]
\bauthor{\bsnm{Hern{\'{a}}n},~\bfnm{Miguel~A}\binits{M.~A.}} \AND
  \bauthor{\bsnm{Robins},~\bfnm{James~M}\binits{J.~M.}}
(\byear{2020}).
\btitle{{Causal inference: what if}}.
\bpublisher{Chapman {\&} Hall/CRC}, \baddress{Boca Raton}.
\end{bbook}
\endbibitem

\bibitem[\protect\citeauthoryear{Holland}{1986}]{Holland1986a}
\begin{barticle}[author]
\bauthor{\bsnm{Holland},~\bfnm{Paul~W.}\binits{P.~W.}}
(\byear{1986}).
\btitle{{Statistics and causal inference}}.
\bjournal{Journal of the American Statistical Association}
\bvolume{81}
\bpages{945--960}.
\bdoi{10.1080/01621459.1986.10478354}
\end{barticle}
\endbibitem

\bibitem[\protect\citeauthoryear{Howe et~al.}{2019}]{Howe2019}
\begin{barticle}[author]
\bauthor{\bsnm{Howe},~\bfnm{Laurence~J.}\binits{L.~J.}},
  \bauthor{\bsnm{Lawson},~\bfnm{Daniel~J.}\binits{D.~J.}},
  \bauthor{\bsnm{Davies},~\bfnm{Neil~M.}\binits{N.~M.}}, \bauthor{\bsnm{{St.
  Pourcain}},~\bfnm{Beate}\binits{B.}},
  \bauthor{\bsnm{Lewis},~\bfnm{Sarah~J.}\binits{S.~J.}}, \bauthor{\bsnm{{Davey
  Smith}},~\bfnm{George}\binits{G.}} \AND
  \bauthor{\bsnm{Hemani},~\bfnm{Gibran}\binits{G.}}
(\byear{2019}).
\btitle{{Genetic evidence for assortative mating on alcohol consumption in the
  UK Biobank}}.
\bjournal{Nature Communications}
\bvolume{10}.
\bdoi{10.1038/s41467-019-12424-x}
\end{barticle}
\endbibitem

\bibitem[\protect\citeauthoryear{Howe et~al.}{2022}]{Howe2022}
\begin{barticle}[author]
\bauthor{\bsnm{Howe},~\bfnm{Laurence~J.}\binits{L.~J.}},
  \bauthor{\bsnm{Nivard},~\bfnm{Michel~G.}\binits{M.~G.}},
  \bauthor{\bsnm{Morris},~\bfnm{Tim~T.}\binits{T.~T.}},
  \bauthor{\bsnm{Hansen},~\bfnm{Ailin~F.}\binits{A.~F.}} \betal{et~al.}
(\byear{2022}).
\btitle{{Within-sibship genome-wide association analyses decrease bias in
  estimates of direct genetic effects}}.
\bjournal{Nature Genetics}
\bvolume{54}
\bpages{581--592}.
\bdoi{10.1038/s41588-022-01062-7}
\end{barticle}
\endbibitem

\bibitem[\protect\citeauthoryear{Imbens and Rubin}{2015}]{Imbens2015}
\begin{bbook}[author]
\bauthor{\bsnm{Imbens},~\bfnm{Guido~W}\binits{G.~W.}} \AND
  \bauthor{\bsnm{Rubin},~\bfnm{Donald~B}\binits{D.~B.}}
(\byear{2015}).
\btitle{{Causal Inference for Statistics, Social, and Biomedical Sciences: An
  Introduction}}.
\bpublisher{Cambridge University Press}, \baddress{Cambridge}.
\bdoi{10.1017/CBO9781139025751}
\end{bbook}
\endbibitem

\bibitem[\protect\citeauthoryear{Kang, Peck and Keele}{2018}]{Kang2018}
\begin{barticle}[author]
\bauthor{\bsnm{Kang},~\bfnm{Hyunseung}\binits{H.}},
  \bauthor{\bsnm{Peck},~\bfnm{Laura}\binits{L.}} \AND
  \bauthor{\bsnm{Keele},~\bfnm{Luke}\binits{L.}}
(\byear{2018}).
\btitle{{Inference for instrumental variables: A randomization inference
  approach}}.
\bjournal{Journal of the Royal Statistical Society. Series A: Statistics in
  Society}
\bvolume{181}
\bpages{1231--1254}.
\bdoi{10.1111/rssa.12353}
\end{barticle}
\endbibitem

\bibitem[\protect\citeauthoryear{Kang
  et~al.}{2016}]{kang16_instr_variab_estim_with_some}
\begin{barticle}[author]
\bauthor{\bsnm{Kang},~\bfnm{Hyunseung}\binits{H.}},
  \bauthor{\bsnm{Zhang},~\bfnm{Anru}\binits{A.}},
  \bauthor{\bsnm{Cai},~\bfnm{T.~Tony}\binits{T.~T.}} \AND
  \bauthor{\bsnm{Small},~\bfnm{Dylan~S.}\binits{D.~S.}}
(\byear{2016}).
\btitle{Instrumental Variables Estimation With Some Invalid Instruments and Its
  Application To Mendelian Randomization}.
\bjournal{Journal of the American Statistical Association}
\bvolume{111}
\bpages{132-144}.
\bdoi{10.1080/01621459.2014.994705}
\end{barticle}
\endbibitem

\bibitem[\protect\citeauthoryear{Katan}{1986}]{Katan2004a}
\begin{barticle}[author]
\bauthor{\bsnm{Katan},~\bfnm{Martjin~B.}\binits{M.~B.}}
(\byear{1986}).
\btitle{{Apolipoprotein E isoforms, serum cholesterol, and cancer}}.
\bjournal{International Journal of Epidemiology}
\bvolume{33}
\bpages{9}.
\bdoi{10.1093/ije/dyh312}
\end{barticle}
\endbibitem

\bibitem[\protect\citeauthoryear{Koles{\'a}r
  et~al.}{2015}]{kolesar15_ident_infer_with_many_inval_instr}
\begin{barticle}[author]
\bauthor{\bsnm{Koles{\'a}r},~\bfnm{Michal}\binits{M.}},
  \bauthor{\bsnm{Chetty},~\bfnm{Raj}\binits{R.}},
  \bauthor{\bsnm{Friedman},~\bfnm{John}\binits{J.}},
  \bauthor{\bsnm{Glaeser},~\bfnm{Edward}\binits{E.}} \AND
  \bauthor{\bsnm{Imbens},~\bfnm{Guido~W.}\binits{G.~W.}}
(\byear{2015}).
\btitle{Identification and Inference With Many Invalid Instruments}.
\bjournal{Journal of Business \& Economic Statistics}
\bvolume{33}
\bpages{474-484}.
\bdoi{10.1080/07350015.2014.978175}
\end{barticle}
\endbibitem

\bibitem[\protect\citeauthoryear{Kong et~al.}{2018}]{Kong2018}
\begin{barticle}[author]
\bauthor{\bsnm{Kong},~\bfnm{Augustine}\binits{A.}},
  \bauthor{\bsnm{Thorleifsson},~\bfnm{Gudmar}\binits{G.}},
  \bauthor{\bsnm{Frigge},~\bfnm{Michael~L.}\binits{M.~L.}},
  \bauthor{\bsnm{Vilhj{\'{a}}lmsson},~\bfnm{Bjarni~J.}\binits{B.~J.}},
  \bauthor{\bsnm{Young},~\bfnm{Alexander~I.}\binits{A.~I.}},
  \bauthor{\bsnm{Thorgeirsson},~\bfnm{Thorgeir~E.}\binits{T.~E.}},
  \bauthor{\bsnm{Benonisdottir},~\bfnm{Stefania}\binits{S.}},
  \bauthor{\bsnm{Oddsson},~\bfnm{Asmundur}\binits{A.}},
  \bauthor{\bsnm{Halld{\'{o}}rsson},~\bfnm{Bjarni~V.}\binits{B.~V.}},
  \bauthor{\bsnm{Masson},~\bfnm{G{\'{i}}sli}\binits{G.}},
  \bauthor{\bsnm{Gudbjartsson},~\bfnm{Daniel~F.}\binits{D.~F.}},
  \bauthor{\bsnm{Helgason},~\bfnm{Agnar}\binits{A.}},
  \bauthor{\bsnm{Bjornsdottir},~\bfnm{Gyda}\binits{G.}},
  \bauthor{\bsnm{Thorsteinsdottir},~\bfnm{Unnur}\binits{U.}} \AND
  \bauthor{\bsnm{Stefansson},~\bfnm{Kari}\binits{K.}}
(\byear{2018}).
\btitle{{The nature of nurture: Effects of parental genotypes}}.
\bjournal{Science}
\bvolume{359}
\bpages{424--428}.
\bdoi{10.1101/219261}
\end{barticle}
\endbibitem

\bibitem[\protect\citeauthoryear{Lander and Schork}{1994}]{Lander1994}
\begin{barticle}[author]
\bauthor{\bsnm{Lander},~\bfnm{Eric~S}\binits{E.~S.}} \AND
  \bauthor{\bsnm{Schork},~\bfnm{Nicholas~J}\binits{N.~J.}}
(\byear{1994}).
\btitle{{Genetic dissection of complex traits}}.
\bjournal{Science}
\bvolume{265}
\bpages{2037--2048}.
\bdoi{10.1038/ng0496-355}
\end{barticle}
\endbibitem

\bibitem[\protect\citeauthoryear{Lauritzen and Sheehan}{2003}]{Lauritzen2003}
\begin{barticle}[author]
\bauthor{\bsnm{Lauritzen},~\bfnm{Steffen~L.}\binits{S.~L.}} \AND
  \bauthor{\bsnm{Sheehan},~\bfnm{Nuala~A.}\binits{N.~A.}}
(\byear{2003}).
\btitle{{Graphical models for genetic analyses}}.
\bjournal{Statistical Science}
\bvolume{18}
\bpages{489--514}.
\bdoi{10.1214/ss/1081443232}
\end{barticle}
\endbibitem

\bibitem[\protect\citeauthoryear{Lauritzen et~al.}{1990}]{Lauritzen1990}
\begin{barticle}[author]
\bauthor{\bsnm{Lauritzen},~\bfnm{Steffen~L}\binits{S.~L.}},
  \bauthor{\bsnm{Dawid},~\bfnm{Phillip~A}\binits{P.~A.}},
  \bauthor{\bsnm{Larsen},~\bfnm{B~N}\binits{B.~N.}} \AND
  \bauthor{\bsnm{Leimer},~\bfnm{H~G}\binits{H.~G.}}
(\byear{1990}).
\btitle{{Independence properties of directed markov fields}}.
\bjournal{Networks}
\bvolume{20}
\bpages{491--505}.
\bdoi{10.1002/net.3230200503}
\end{barticle}
\endbibitem

\bibitem[\protect\citeauthoryear{Lawlor, Tilling and
  Davey~Smith}{2017}]{lawlor17_trian_aetiol_epidem}
\begin{barticle}[author]
\bauthor{\bsnm{Lawlor},~\bfnm{Debbie~A.}\binits{D.~A.}},
  \bauthor{\bsnm{Tilling},~\bfnm{Kate}\binits{K.}} \AND
  \bauthor{\bsnm{Davey~Smith},~\bfnm{George}\binits{G.}}
(\byear{2017}).
\btitle{Triangulation in Aetiological Epidemiology}.
\bjournal{International Journal of Epidemiology}
\bvolume{nil}
\bpages{dyw314}.
\bdoi{10.1093/ije/dyw314}
\end{barticle}
\endbibitem

\bibitem[\protect\citeauthoryear{Lower
  et~al.}{1979}]{lower79_n_acety_phenot_risk_urinar_bladd_cancer}
\begin{barticle}[author]
\bauthor{\bsnm{Lower},~\bfnm{G.~M.}\binits{G.~M.}},
  \bauthor{\bsnm{Nilsson},~\bfnm{T.}\binits{T.}},
  \bauthor{\bsnm{Nelson},~\bfnm{C.~E.}\binits{C.~E.}},
  \bauthor{\bsnm{Wolf},~\bfnm{H.}\binits{H.}},
  \bauthor{\bsnm{Gamsky},~\bfnm{T.~E.}\binits{T.~E.}} \AND
  \bauthor{\bsnm{Bryan},~\bfnm{G.~T.}\binits{G.~T.}}
(\byear{1979}).
\btitle{N-Acetyltransferase Phenotype and Risk in Urinary Bladder Cancer:
  Approaches in Molecular Epidemiology. Preliminary Results in Sweden and
  Denmark}.
\bjournal{Environmental Health Perspectives}
\bvolume{29}
\bpages{71-79}.
\bdoi{10.1289/ehp.792971}
\end{barticle}
\endbibitem

\bibitem[\protect\citeauthoryear{Millwood et~al.}{2019}]{Millwood2019}
\begin{barticle}[author]
\bauthor{\bsnm{Millwood},~\bfnm{Iona~Y.}\binits{I.~Y.}},
  \bauthor{\bsnm{Walters},~\bfnm{Robin~G.}\binits{R.~G.}},
  \bauthor{\bsnm{Mei},~\bfnm{Xue~W.}\binits{X.~W.}},
  \bauthor{\bsnm{Guo},~\bfnm{Yu}\binits{Y.}},
  \bauthor{\bsnm{Yang},~\bfnm{Ling}\binits{L.}},
  \bauthor{\bsnm{Bian},~\bfnm{Zheng}\binits{Z.}},
  \bauthor{\bsnm{Bennett},~\bfnm{Derrick~A.}\binits{D.~A.}},
  \bauthor{\bsnm{Chen},~\bfnm{Yiping}\binits{Y.}},
  \bauthor{\bsnm{Dong},~\bfnm{Caixia}\binits{C.}},
  \bauthor{\bsnm{Hu},~\bfnm{Ruying}\binits{R.}},
  \bauthor{\bsnm{Zhou},~\bfnm{Gang}\binits{G.}},
  \bauthor{\bsnm{Yu},~\bfnm{Bo}\binits{B.}},
  \bauthor{\bsnm{Jia},~\bfnm{Weifang}\binits{W.}},
  \bauthor{\bsnm{Parish},~\bfnm{Sarah}\binits{S.}},
  \bauthor{\bsnm{Clarke},~\bfnm{Robert}\binits{R.}}, \bauthor{\bsnm{{Davey
  Smith}},~\bfnm{George}\binits{G.}},
  \bauthor{\bsnm{Collins},~\bfnm{Rory}\binits{R.}},
  \bauthor{\bsnm{Holmes},~\bfnm{Michael~V.}\binits{M.~V.}},
  \bauthor{\bsnm{Li},~\bfnm{Liming}\binits{L.}},
  \bauthor{\bsnm{Peto},~\bfnm{Richard}\binits{R.}} \AND
  \bauthor{\bsnm{Chen},~\bfnm{Zhengming}\binits{Z.}}
(\byear{2019}).
\btitle{{Conventional and genetic evidence on alcohol and vascular disease
  aetiology: a prospective study of 500 000 men and women in China}}.
\bjournal{The Lancet}
\bvolume{393}
\bpages{1831--1842}.
\bdoi{10.1016/S0140-6736(18)31772-0}
\end{barticle}
\endbibitem

\bibitem[\protect\citeauthoryear{Morton}{1955}]{morton1955sequential}
\begin{barticle}[author]
\bauthor{\bsnm{Morton},~\bfnm{Newton~E}\binits{N.~E.}}
(\byear{1955}).
\btitle{Sequential tests for the detection of linkage}.
\bjournal{American Journal of Human Genetics}
\bvolume{7}
\bpages{277-318}.
\end{barticle}
\endbibitem

\bibitem[\protect\citeauthoryear{Nadeau}{2017}]{Nadeau2017}
\begin{barticle}[author]
\bauthor{\bsnm{Nadeau},~\bfnm{Joseph~H.}\binits{J.~H.}}
(\byear{2017}).
\btitle{{Do gametes woo? Evidence for their nonrandom union at fertilization}}.
\bjournal{Genetics}
\bvolume{207}
\bpages{369--387}.
\bdoi{10.1534/genetics.117.300109}
\end{barticle}
\endbibitem

\bibitem[\protect\citeauthoryear{Neyman}{1990}]{Neyman1923}
\begin{barticle}[author]
\bauthor{\bsnm{Neyman},~\bfnm{Jerzy}\binits{J.}}
(\byear{1990}).
\btitle{{On the application of probability theory to agricultural experiments.
  Essay on principles. Section 9.}}
\bjournal{Statistical Science}
\bvolume{5}
\bpages{465--480}.
\end{barticle}
\endbibitem

\bibitem[\protect\citeauthoryear{Otto and Payseur}{2019}]{Otto2019}
\begin{barticle}[author]
\bauthor{\bsnm{Otto},~\bfnm{Sarah~P.}\binits{S.~P.}} \AND
  \bauthor{\bsnm{Payseur},~\bfnm{Bret~A.}\binits{B.~A.}}
(\byear{2019}).
\btitle{{Crossover interference: Shedding light on the evolution of
  recombination}}.
\bjournal{Annual Review of Genetics}
\bvolume{53}
\bpages{19--44}.
\bdoi{10.1146/annurev-genet-040119-093957}
\end{barticle}
\endbibitem

\bibitem[\protect\citeauthoryear{Patterson, Price and
  Reich}{2006}]{patterson06_popul_struc_eigen}
\begin{barticle}[author]
\bauthor{\bsnm{Patterson},~\bfnm{Nick}\binits{N.}},
  \bauthor{\bsnm{Price},~\bfnm{Alkes~L.}\binits{A.~L.}} \AND
  \bauthor{\bsnm{Reich},~\bfnm{David}\binits{D.}}
(\byear{2006}).
\btitle{Population Structure and Eigenanalysis}.
\bjournal{PLoS Genetics}
\bvolume{2}
\bpages{e190}.
\bdoi{10.1371/journal.pgen.0020190}
\end{barticle}
\endbibitem

\bibitem[\protect\citeauthoryear{Pearl}{2009}]{Pearl2009}
\begin{bbook}[author]
\bauthor{\bsnm{Pearl},~\bfnm{Judea}\binits{J.}}
(\byear{2009}).
\btitle{{Causality}},
\bedition{2} ed.
\bpublisher{Cambridge University Press}, \baddress{New York}.
\end{bbook}
\endbibitem

\bibitem[\protect\citeauthoryear{Pitman}{1937}]{pitman37_signif_tests_which_may_be}
\begin{barticle}[author]
\bauthor{\bsnm{Pitman},~\bfnm{E.~J.~G.}\binits{E.~J.~G.}}
(\byear{1937}).
\btitle{Significance Tests Which May Be Applied To Samples From Any
  Populations}.
\bjournal{Supplement to the Journal of the Royal Statistical Society}
\bvolume{4}
\bpages{119-130}.
\bdoi{10.2307/2984124}
\end{barticle}
\endbibitem

\bibitem[\protect\citeauthoryear{Richardson and Robins}{2013}]{Richardson2013b}
\begin{bunpublished}[author]
\bauthor{\bsnm{Richardson},~\bfnm{Tom~S}\binits{T.~S.}} \AND
  \bauthor{\bsnm{Robins},~\bfnm{James~M}\binits{J.~M.}}
(\byear{2013}).
\btitle{{Single world intervention graphs (SWIGs): A unification of the
  counterfactual and graphical approaches to causality}}.
\end{bunpublished}
\endbibitem

\bibitem[\protect\citeauthoryear{Rose and van~der Laan}{2008}]{Rose2008}
\begin{barticle}[author]
\bauthor{\bsnm{Rose},~\bfnm{Sherri}\binits{S.}} \AND
  \bauthor{\bparticle{van~der} \bsnm{Laan},~\bfnm{Mark~J}\binits{M.~J.}}
(\byear{2008}).
\btitle{{Simple optimal weighting of cases and controls in case-controls
  studies}}.
\bjournal{The International Journal of Biostatistics}
\bvolume{4}.
\end{barticle}
\endbibitem

\bibitem[\protect\citeauthoryear{Rosenbaum}{2004}]{Rosenbaum2004}
\begin{bmisc}[author]
\bauthor{\bsnm{Rosenbaum},~\bfnm{Paul~R}\binits{P.~R.}}
(\byear{2004}).
\btitle{{Randomization inference with an instrumental variable}}.
\end{bmisc}
\endbibitem

\bibitem[\protect\citeauthoryear{Rosenbaum}{2010}]{rosenbaum2010evidence}
\begin{barticle}[author]
\bauthor{\bsnm{Rosenbaum},~\bfnm{Paul~R}\binits{P.~R.}}
(\byear{2010}).
\btitle{Evidence factors in observational studies}.
\bjournal{Biometrika}
\bvolume{97}
\bpages{333--345}.
\end{barticle}
\endbibitem

\bibitem[\protect\citeauthoryear{Rosenbaum}{2021}]{rosenbaum2021replication}
\begin{bbook}[author]
\bauthor{\bsnm{Rosenbaum},~\bfnm{Paul}\binits{P.}}
(\byear{2021}).
\btitle{Replication and evidence factors in observational studies}.
\bpublisher{CRC Press}.
\end{bbook}
\endbibitem

\bibitem[\protect\citeauthoryear{Rosenbaum and Rubin}{1983}]{Rosenbaum1983}
\begin{barticle}[author]
\bauthor{\bsnm{Rosenbaum},~\bfnm{Paul~R.}\binits{P.~R.}} \AND
  \bauthor{\bsnm{Rubin},~\bfnm{Donald~B.}\binits{D.~B.}}
(\byear{1983}).
\btitle{{The central role of the propensity score in observational studies for
  causal effects}}.
\bjournal{Biometrika}
\bvolume{70}
\bpages{41--55}.
\bdoi{10.1017/CBO9780511810725.016}
\end{barticle}
\endbibitem

\bibitem[\protect\citeauthoryear{Rosenberger, Uschner and
  Wang}{2019}]{Rosenberger2019}
\begin{barticle}[author]
\bauthor{\bsnm{Rosenberger},~\bfnm{William~F.}\binits{W.~F.}},
  \bauthor{\bsnm{Uschner},~\bfnm{Diane}\binits{D.}} \AND
  \bauthor{\bsnm{Wang},~\bfnm{Yanying}\binits{Y.}}
(\byear{2019}).
\btitle{{Randomization: The forgotten component of the randomized clinical
  trial}}.
\bjournal{Statistics in Medicine}
\bvolume{38}
\bpages{1--12}.
\bdoi{10.1002/sim.7901}
\end{barticle}
\endbibitem

\bibitem[\protect\citeauthoryear{Rubin}{1974}]{Rubin1974}
\begin{barticle}[author]
\bauthor{\bsnm{Rubin},~\bfnm{Donald~B}\binits{D.~B.}}
(\byear{1974}).
\btitle{{Estimating causal effects of treatment in randomized and nonrandomized
  studies}}.
\bjournal{Journal of Educational Psychology}
\bvolume{66}
\bpages{688--701}.
\end{barticle}
\endbibitem

\bibitem[\protect\citeauthoryear{Rubin}{1980}]{Rubin1980}
\begin{barticle}[author]
\bauthor{\bsnm{Rubin},~\bfnm{Donald~B}\binits{D.~B.}}
(\byear{1980}).
\btitle{{Comment: 'Randomization analysis of experimental data: The Fisher
  randomization test'}}.
\bjournal{Journal of the American Statistical Association}
\bvolume{75}
\bpages{591--593}.
\bdoi{10.1080/01621459.1980.10477512}
\end{barticle}
\endbibitem

\bibitem[\protect\citeauthoryear{Sanderson et~al.}{2022}]{Sanderson2022}
\begin{barticle}[author]
\bauthor{\bsnm{Sanderson},~\bfnm{Eleanor}\binits{E.}},
  \bauthor{\bsnm{Glymour},~\bfnm{Maria~M}\binits{M.~M.}},
  \bauthor{\bsnm{Holmes},~\bfnm{Michael~V}\binits{M.~V.}},
  \bauthor{\bsnm{Kang},~\bfnm{Hyunseung}\binits{H.}},
  \bauthor{\bsnm{Morrison},~\bfnm{Jean}\binits{J.}},
  \bauthor{\bsnm{Munaf{\`{o}}},~\bfnm{Marcus~R}\binits{M.~R.}},
  \bauthor{\bsnm{Palmer},~\bfnm{Tom}\binits{T.}},
  \bauthor{\bsnm{Schooling},~\bfnm{Mary~C}\binits{M.~C.}},
  \bauthor{\bsnm{Wallace},~\bfnm{Chris}\binits{C.}},
  \bauthor{\bsnm{Zhao},~\bfnm{Qingyuan}\binits{Q.}} \AND \bauthor{\bsnm{{Davey
  Smith}},~\bfnm{George}\binits{G.}}
(\byear{2022}).
\btitle{{Mendelian randomization}}.
\bjournal{Nature Reviews Methods Primers}
\bvolume{2}.
\bdoi{10.1038/s43586-021-00092-5}
\end{barticle}
\endbibitem

\bibitem[\protect\citeauthoryear{Scharfstein, Rotnitzky and
  Robins}{1999}]{scharfstein1999adjusting}
\begin{barticle}[author]
\bauthor{\bsnm{Scharfstein},~\bfnm{Daniel~O}\binits{D.~O.}},
  \bauthor{\bsnm{Rotnitzky},~\bfnm{Andrea}\binits{A.}} \AND
  \bauthor{\bsnm{Robins},~\bfnm{James~M}\binits{J.~M.}}
(\byear{1999}).
\btitle{Adjusting for nonignorable drop-out using semiparametric nonresponse
  models}.
\bjournal{Journal of the American Statistical Association}
\bvolume{94}
\bpages{1096--1120}.
\end{barticle}
\endbibitem

\bibitem[\protect\citeauthoryear{Spielman, McGinnis and
  Ewens}{1993}]{Spielman1993}
\begin{barticle}[author]
\bauthor{\bsnm{Spielman},~\bfnm{R.~S.}\binits{R.~S.}},
  \bauthor{\bsnm{McGinnis},~\bfnm{R.~E.}\binits{R.~E.}} \AND
  \bauthor{\bsnm{Ewens},~\bfnm{W.~J.}\binits{W.~J.}}
(\byear{1993}).
\btitle{{Transmission test for linkage disequilibrium: The insulin gene region
  and insulin-dependent diabetes mellitus (IDDM)}}.
\bjournal{American Journal of Human Genetics}
\bvolume{52}
\bpages{506--516}.
\end{barticle}
\endbibitem

\bibitem[\protect\citeauthoryear{Spirtes, Glymour and
  Scheines}{2000}]{spirtes2000causation}
\begin{bbook}[author]
\bauthor{\bsnm{Spirtes},~\bfnm{Peter}\binits{P.}},
  \bauthor{\bsnm{Glymour},~\bfnm{Clark~N}\binits{C.~N.}} \AND
  \bauthor{\bsnm{Scheines},~\bfnm{Richard}\binits{R.}}
(\byear{2000}).
\btitle{Causation, prediction, and search}.
\bpublisher{MIT press}.
\end{bbook}
\endbibitem

\bibitem[\protect\citeauthoryear{Taubes}{1995}]{Taubes1995}
\begin{barticle}[author]
\bauthor{\bsnm{Taubes},~\bfnm{Gary}\binits{G.}}
(\byear{1995}).
\btitle{{Epidemiology faces its limits}}.
\bjournal{Science}
\bvolume{269}
\bpages{164--169}.
\end{barticle}
\endbibitem

\bibitem[\protect\citeauthoryear{Thomas and Conti}{2004}]{Thomas2004}
\begin{barticle}[author]
\bauthor{\bsnm{Thomas},~\bfnm{Duncan~C.}\binits{D.~C.}} \AND
  \bauthor{\bsnm{Conti},~\bfnm{David~V.}\binits{D.~V.}}
(\byear{2004}).
\btitle{{Commentary: The concept of 'Mendelian randomization'}}.
\bjournal{International Journal of Epidemiology}
\bvolume{33}
\bpages{21--25}.
\bdoi{10.1093/ije/dyh048}
\end{barticle}
\endbibitem

\bibitem[\protect\citeauthoryear{Thompson}{2000}]{Thompson2000}
\begin{binproceedings}[author]
\bauthor{\bsnm{Thompson},~\bfnm{Elizabeth~A}\binits{E.~A.}}
(\byear{2000}).
\btitle{{Statistical inference from genetic data on pedigrees}}.
In \bbooktitle{NSF-CBMS Regional Conference Series in Probability and
  Statistics}
\bvolume{6}
\bpages{1--169}.
\end{binproceedings}
\endbibitem

\bibitem[\protect\citeauthoryear{Vogelezang et~al.}{2020}]{Vogelezang2020}
\begin{barticle}[author]
\bauthor{\bsnm{Vogelezang},~\bfnm{Suzanne}\binits{S.}},
  \bauthor{\bsnm{Bradfield},~\bfnm{Jonathan~P}\binits{J.~P.}},
  \bauthor{\bsnm{Ahluwalia},~\bfnm{Tarunveer~S}\binits{T.~S.}},
  \bauthor{\bsnm{Curtin},~\bfnm{John~A}\binits{J.~A.}} \betal{et~al.}
(\byear{2020}).
\btitle{{Novel loci for childhood body mass index and shared heritability with
  adult cardiometabolic traits}}.
\bjournal{PLoS Genetics}
\bpages{1--26}.
\bdoi{10.1371/journal.pgen.1008718}
\end{barticle}
\endbibitem

\bibitem[\protect\citeauthoryear{Wang and Owen}{2019}]{Wang2019}
\begin{barticle}[author]
\bauthor{\bsnm{Wang},~\bfnm{Jingshu}\binits{J.}} \AND
  \bauthor{\bsnm{Owen},~\bfnm{Art~B.}\binits{A.~B.}}
(\byear{2019}).
\btitle{{Admissibility in Partial Conjunction Testing}}.
\bjournal{Journal of the American Statistical Association}
\bvolume{114}
\bpages{158--168}.
\bdoi{10.1080/01621459.2017.1385465}
\end{barticle}
\endbibitem

\bibitem[\protect\citeauthoryear{Wheatley and Gray}{2004}]{Wheatley2004}
\begin{barticle}[author]
\bauthor{\bsnm{Wheatley},~\bfnm{Keith}\binits{K.}} \AND
  \bauthor{\bsnm{Gray},~\bfnm{Richard}\binits{R.}}
(\byear{2004}).
\btitle{{Commentary: Mendelian randomization - An update on its use to evaluate
  allogeneic stem cell transplantation in leukemia}}.
\bjournal{International Journal of Epidemiology}
\bvolume{33}
\bpages{15--17}.
\bdoi{10.1093/ije/dyg313}
\end{barticle}
\endbibitem

\bibitem[\protect\citeauthoryear{Wright}{1920}]{Wright1920}
\begin{barticle}[author]
\bauthor{\bsnm{Wright},~\bfnm{Sewall}\binits{S.}}
(\byear{1920}).
\btitle{{The relative importance of heredity: determining the piebald pattern
  of guinea pigs}}.
\bjournal{Proceedings of the National Academy of Sciences}
\bvolume{6}
\bpages{320--332}.
\end{barticle}
\endbibitem

\bibitem[\protect\citeauthoryear{Wright}{1923}]{Wright1923}
\begin{barticle}[author]
\bauthor{\bsnm{Wright},~\bfnm{Sewall}\binits{S.}}
(\byear{1923}).
\btitle{{The theory of path coefficients: a reply to Niles' criticism}}.
\bjournal{Genetics}
\bvolume{8}
\bpages{239--255}.
\end{barticle}
\endbibitem

\bibitem[\protect\citeauthoryear{Zhao et~al.}{2020}]{Zhao2020}
\begin{barticle}[author]
\bauthor{\bsnm{Zhao},~\bfnm{Qingyuan}\binits{Q.}},
  \bauthor{\bsnm{Wang},~\bfnm{Jingshu}\binits{J.}},
  \bauthor{\bsnm{Hemani},~\bfnm{Gibran}\binits{G.}},
  \bauthor{\bsnm{Bowden},~\bfnm{Jack}\binits{J.}} \AND
  \bauthor{\bsnm{Small},~\bfnm{Dylan~S.}\binits{D.~S.}}
(\byear{2020}).
\btitle{{Statistical inference in two-sample summary-data Mendelian
  randomization using robust adjusted profile score}}.
\bjournal{Annals of Statistics}
\bvolume{48}
\bpages{1742--1769}.
\bdoi{10.1214/19-AOS1866}
\end{barticle}
\endbibitem

\bibitem[\protect\citeauthoryear{Zhao
  et~al.}{2022}]{zhao22_eviden_factor_from_multip_possib}
\begin{barticle}[author]
\bauthor{\bsnm{Zhao},~\bfnm{Anqi}\binits{A.}},
  \bauthor{\bsnm{Lee},~\bfnm{Youjin}\binits{Y.}},
  \bauthor{\bsnm{Small},~\bfnm{Dylan~S.}\binits{D.~S.}} \AND
  \bauthor{\bsnm{Karmakar},~\bfnm{Bikram}\binits{B.}}
(\byear{2022}).
\btitle{Evidence Factors From Multiple, Possibly Invalid, Instrumental
  Variables}.
\bjournal{The Annals of Statistics}
\bvolume{50}
\bpages{nil}.
\bdoi{10.1214/21-aos2148}
\end{barticle}
\endbibitem

\end{thebibliography}

\clearpage

\appendix
\onecolumn
\section{Introduction to causal inference} \label{sec:intro-to-causal-inference}

This section introduces some basic concepts in causal inference, including the potential outcomes framework \citep{Neyman1923, Rubin1974}, randomization
inference \citep{Fisher1935,Rubin1980}, instrumental variables \citep{Wright1923, Imbens2015}, and single world intervention graphs \citep{Richardson2013b}, which are essential for describing our methodology. Interested readers can find a more thorough coverage of these topics in other texts \citep{Imbens2015,Hernan2020}.

\subsection{Treatment assignment and potential outcomes}
\label{sec:potential-outcomes}

In the typical setup of a randomized experiment with non-compliance, we have a sample of $N$ individuals indexed by $i = 1, 2, \ldots , N$ and each individual is randomly assigned to receive a binary treatment $Z_i \in \{0,1\}$. The common convention is that $Z_i = 1$ denotes assignment to an experimental treatment and $Z_i = 0$ a control treatment. However, individuals might not comply with their assigned
treatment, and we denote the treatment that the individual actually takes as $D_i \in \{0,1\}$. Finally, we observe an outcome variable $Y_i$ for each individual. As the treatment uptake $D_i$ is not randomized, there may exist a confounding variable $C_i$ that is a common cause of both $D_i$ and $Y_i$.

Individual $i$ has two \emph{potential outcomes} (also called
\emph{counterfactuals}) of her treatment uptake,
$D_i(0)$ and $D_i(1)$. If she is randomized to the experimental (or control)
treatment, she will take treatment $D_i(1)$ (or $D_i(0)$). For
example, some individuals will take the experimental treatment
regardless of their assigned treatment, so $D_i(1) = D_i(0) =
1$. Each individual also has four potential
outcomes $Y_i(z,d)$ from the experiment corresponding to each
combination of treatment assignment $z \in \{0,1\}$ and treatment
uptake $d \in \{0,1\}$. Similarly, we may define the potential outcome
with just $D$ being intervened on by $Y_i(d) = Y_i(Z_i, d)$, where the
assigned treatment takes its ``natural'' value $Z_i$.

In defining these potential outcomes, we have
implicitly made the \emph{no interference} assumption which states that
individual $j$'s treatment is independent of individual $i$'s outcome
when $i \neq j$ \citep{Rubin1980}. To simplify the exposition, we
further make the \emph{exclusion restriction} assumption in this section. That
is, we assume that the treatment assignment has no causal effect on the
outcome except via treatment uptake, so
\begin{equation}
  \label{eq:exclusion-restriction}
  Y_i(1,d) = Y_i(0,d) = Y_i(d) \, \text{for $d \in \{0,1\}$}.
\end{equation}
Let $\mathcal{F} = \{(D_i(1), D_i(0), Y_i(0), Y_i(1)), \, i = 1,
\ldots, N \}$ denote the collection of potential outcomes for all
the individuals.

We define a \emph{causal effect} of the treatment as a contrast of
potential outcomes. When the treatment is binary, the causal effect
for individual $i$ is given by $\beta_i = Y_i(1) - Y_i(0)$, the
difference in individual $i$'s outcomes between the two possible
treatments. However, inference for the individual treatment effect
$\beta_i$ is difficult because we do not observe both potential
outcomes of the same individual simultaneously. This has been famously
described as the ``fundamental problem of causal inference''
\citep{Holland1986a}. Indeed, we only observe the potential
outcome corresponding to the treatment that is actually received, such
that
\begin{equation}
  \label{eq:consistency}
  D_i = \begin{cases}
    D_i(1) & \text{if } Z_i = 1, \\
    D_i(0) & \text{if } Z_i = 0;
  \end{cases} ~ \text{and} ~ Y_i = \begin{cases}
    Y_i(1) & \text{if } D_i = 1, \\
    Y_i(0) & \text{if } D_i = 0.
  \end{cases}
\end{equation}
The above equation is sometimes called the \emph{consistency}
assumption since it ensures that the observed outcomes and potential
outcomes are consistent with one another
\citep{Hernan2020}.

\begin{table}[h]
  \caption{Observed data from a hypothetical experiment for a LDL
    cholesterol-lowering drug.}
  \begin{center}
    \begin{tabular}{c|c|ccc|ccc}
      \toprule
      $i$ & $Z_i$ & $D_i$ & $D_i(1)$ & $D_i(0)$ & $Y_i$ & $Y_i(1)$ & $Y_i(0)$ \\
      \midrule
      1 & 1 & 1 & 1 & \textcolor{red}{?} & 120 & 120 & \textcolor{red}{?}  \\
      2 & 1 & 1 & 1 & \textcolor{red}{?} & 120 & 120 & \textcolor{red}{?} \\
      3 & 1 & 0 & 0 & \textcolor{red}{?} & 75 & \textcolor{red}{?} & 75 \\
      4 & 0 & 0 & \textcolor{red}{?} & 1 & 165 & 165 & \textcolor{red}{?} \\
      5 & 0 & 0 & \textcolor{red}{?} & 0 & 135 & \textcolor{red}{?} & 135 \\
      6 & 0 & 0 & \textcolor{red}{?} & 0 & 105 & \textcolor{red}{?} & 105 \\
      \bottomrule
    \end{tabular}
    \label{tab:hypomiss}
  \end{center}
\end{table}

From this perspective, causal inference can be regarded as a missing
data problem. Consider a simple hypothetical experiment in
\Cref{tab:hypomiss} consisting of $N = 6$ individuals, 3 of whom
are randomized to take an experimental LDL cholesterol-lowering drug
and 3 of whom are randomized to take a placebo. However, not everyone
adheres to the assigned treatment. The outcome variable is LDL
cholesterol measured in grams per litre (mg/dL). As discussed above,
we can only observe the potential outcomes corresponding to the
observed treatment assignment and uptake; all other potential outcomes
are missing.

\subsection{Causal graphical models}
\label{sec:caus-graph-models}

\begin{figure}[t]
  \centering
  \begin{subfigure}[b]{0.32\textwidth}
    \centering
    \includegraphics[width=\textwidth]{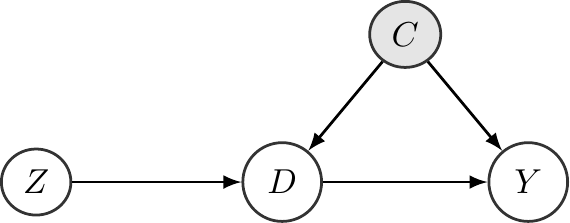}
    \caption{Directed acyclic graph}
    \label{fig:dag}
  \end{subfigure}
  \hspace{15mm}
  \begin{subfigure}[b]{0.32\textwidth}
    \centering
    \includegraphics[width=\textwidth]{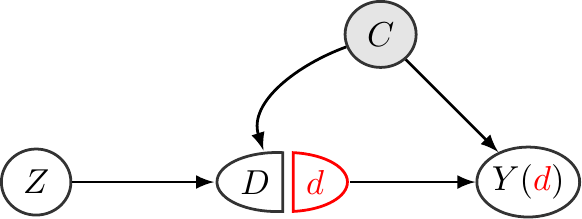}
    \caption{Single world intervention graph}
    \label{fig:swig}
  \end{subfigure}%
  \caption{Causal diagram of the example in \Cref{sec:intro-to-causal-inference}} \label{fig:example_graphs}
\end{figure}

The setting above can be described by a directed acyclic graph (DAG) as
shown in \Cref{fig:dag}. Below we will use some basic concepts in DAG
models such as Markov properties and d-separation, which are described
in detail in other texts
\citep{Lauritzen1990,Pearl2009,Hernan2020}.

We will use single world intervention graphs (SWIG) to unify the counterfactual
and graphical descriptions of the causal inference problem
\citep{Richardson2013b}. The SWIG representation of our setting above
is given in \Cref{fig:swig}. Since we are interested in the causal
effect of an intervention on $D$, the SWIG splits that node into two
halves: $D$, representing the randomized treatment, and $d$, representing a fixed intervention value. The random half $D$ inherits all incoming arrows in the original DAG and
the fixed half $d$ inherits all outgoing arrows. Descendants of the
intervention node (in this case $Y$) are replaced with the potential
outcomes $Y(d)$ under the intervention value $d$.

It has been shown that SWIGs define a graphical model for the
potential outcomes \citep{Richardson2013b}, so we can apply
d-separation to obtain conditional independence between
counterfactuals. For example, \Cref{fig:swig} implies \emph{exchangeability} (or \emph{ignorability}),
\begin{equation} \label{eqn:random-assignment}
  Z_i \independent Y_i(d) ~ \text{for all $d \in \{0,1\}$.}
\end{equation}
However, $D_i \independent Y_i(d)$ is generally not true due to the
confounder $C_i$.

\subsection{Randomization inference for instrumental variables}
\label{sec:rand-infer}

To construct an exact randomization test, the key idea is to
base the inference precisely on the randomness introduced
by the experimenter. To this end, we must characterize the treatment
assignment mechanism.

Let $\bm{Z} = (Z_1, \ldots, Z_N)^\intercal$ denote the $N$-vector of
treatment assignments. To simplify the exposition, we will assume that
the experiment is completely randomized, such that a fixed number of
individuals $N_t$ are assigned to the experimental treatment and $N_c
= N - N_t$ are assigned to the control treatment. The same method
below can be applied to more sophisticated assignment
mechanisms (such as the ones we describe later for within-family MR).
Let $\Omega = \{ (z_1, \ldots, z_N) \in \{0,1\}^N
\colon \sum_{i=1}^N z_i = N_t\}$ denote the set of feasible assignment
vectors. By assumption, all assignment vectors in $\Omega$ are
realized with equal probability. Stated formally, the randomization
distribution can be written as
\begin{equation}
  \label{eq:randomization-distribution}
  \mathbb{P}(\bm{Z} = \bm{z} \mid \mathcal{F}) =
  \begin{cases}
    \displaystyle \binom{N}{N_t}^{-1}, & \text{for all $\bm{z} \in \Omega$}, \\
    0, & \text{otherwise}.
  \end{cases}
\end{equation}

To illustrate randomization inference, consider the hypothetical
experiment in \Cref{tab:hypomiss}. Suppose we are interested in
evaluating the effectiveness of this drug at lowering LDL
cholesterol. However, although the drug is initially randomly
assigned, the treatment uptake is not randomized. In particular,
non-compliance may be driven by a confounder $C_i \in \{0,1\}$. This
might be an underlying comorbidity such that those with $C_i = 1$ have
a higher baseline outcome $Y_i(0)$ but experience negative side
effects from the experimental treatment. Due to the side effects,
individuals with the comorbidity (such as $i=3$ in \Cref{tab:hypomiss})
may be inclined to switch to the control treatment when they are
assigned to the experimental drug. Due to the systematic shift of high
baseline individuals from the experimental treatment to the control
treatment, a simple intention-to-treat estimate (by regressing $Y_i$
on $Z_i$) will underestimate the causal effect.

To address unobserved confounding such as systematic non-compliance,
one approach is to use an instrumental variable. An instrumental
variable induces unconfounded variation in the treatment without
otherwise affecting the outcome. In our example, the randomized treatment
assignment $Z_i \in \{0, 1\}$ is a good instrument for treatment
uptake $D_i$, because it will change the outcome $Y_i$ only through
$D_i$ by the exclusion restriction assumption
\eqref{eq:exclusion-restriction}. Furthermore, it is independent of
the underlying comorbidity status and the counterfactual outcomes,
as shown by \eqref{eqn:random-assignment}.

Randomization inference for instrumental variables
\citep{Rosenbaum2004, Kang2018} tests sharp null
hypotheses of the form
\[
  H_0: Y_i(d) - Y_i(0) = \beta_0 d,~\text{for all}~d \in \{0,1\}.
\]
This implies a constant additive treatment effect $\beta_0$ across
individuals. Under this hypothesis and the consistency assumption
\eqref{eq:consistency}, the baseline potential outcome can be
written in terms of the observable data $(Z_i, D_i, Y_i)$ as
\[
  Y_i(0) = Y_i - \beta_0 D_i =
  \begin{cases}
    Y_i,& \text{if}~D_i = 0, \\
    Y_i - \beta_0,& \text{if}~D_i = 1,
  \end{cases}
\]
which is termed as the ``adjusted response'' by
\citet{Rosenbaum2004}. Therefore, when the null hypothesis is true,
the randomization of $Z_i$, namely \eqref{eqn:random-assignment},
implies that
\[
  Z_i \independent Y_i - \beta_0 D_i.
\]
Consequently, testing the null hypothesis $H_0$ that the causal effect is
a constant $\beta_0$ is equivalent to testing the independence of
$Z_i$ and $Y_i - \beta_0 D_i$. To this end, a simple test statistic is
the difference in outcomes between the two groups,
\[
  T(\bm{Z} \mid \mathcal{F}) = \sum_{i=1}^N Z_i (Y_i - \beta_0 D_i) - \sum_{i=1}^N (1 - Z_i) (Y_i - \beta_0 D_i) \overset{H_0}{=} \sum_{i:Z_i=1} Y_i(0) - \sum_{i:Z_i=0} Y_i(0).
\]
The randomization test then rejects $H_0$ at significance level
$\alpha$, if the p-value
\[
  P(\bm Z \mid \mathcal{F}) = \tilde{\mathbb{P}}(T(\tilde{\bm{Z}} \mid
  \mathcal{F}) \leq T(\bm{Z} \mid \mathcal{F}))
\]
is less than or equal to $\alpha$. Here $\tilde{\bm Z}$ is an
independent copy of $\bm Z$ and $\tilde{\mathbb{P}}$ means that the
probability is taken over $\tilde{\bm Z}$ according to the
randomization distribution \eqref{eq:randomization-distribution}. In plain terms, we are asking: if we re-ran the experiment many times under the null hypothesis (i.e. $Z_i$ and $Y_i - \beta_0 D_i$ are independent), how often would we observe a test statistic more extreme than our observed test statistic? If this probability is lower than $\alpha$, then we have little confidence in the null hypothesis.

This p-value has size $\alpha$ in the sense that
\[
  \mathbb{P}(P(\bm{Z} \mid \mathcal{F}) \leq \alpha \mid H_0) = \alpha.
\]
for any significance level $0 \leq \alpha \leq 1$ and test statistic $T(\cdot \mid \mathcal{F})$. For continuously distributed test statistics the proof relies on the idea that $T(\tilde{\bm{Z}} \mid \mathcal{F}) \overset{d}{=} T(\bm{Z} \mid \mathcal{F})$ under $H_0$ which means that $P(\bm{Z} \mid \mathcal{F})$ is the cumulative distribution of $T(\tilde{\bm{Z}} \mid \mathcal{F})$ at $T(\bm{Z} \mid \mathcal{F})$. Since cumulative distributions are uniformly distributed the result follows.

Next, we illustrate the randomization test using the hypothetical
experiment in \Cref{tab:hypomiss} and the null hypothesis $H_0:
Y_i(0) = Y_i(1)$ for all $i$ (i.e. $\beta_0 = 0$). For the realized experiment in
\Cref{tab:hypomiss}, the difference in outcomes between the experimental and placebo groups is
$(120 + 120)/2  - (75 + 165 + 135 + 105)/4 = 0$. In other words, average LDL cholesterol appears to be identical in the experimental and control arms. However, it is unclear whether this should be interpreted as evidence of a null causal effect. As discussed above, it is possible that individuals with high baseline outcomes are more inclined to switch from the experimental treatment to the control treatment.

Our observed test statistic for this experiment is $T(\bm{Z} \mid \mathcal{F}) = (120 + 120 + 75)/3 - (165+135+105)/3 = -30$. Since we know the missing potential outcomes under the null hypothesis and we know the mechanism by which treatment was randomly assigned, we
can also consider the results of counterfactual experiments. Table
\ref{tab:hyponull} shows a counterfactual experiment that could have
occurred with missing potential outcomes imputed under the null. The counterfactual
treatment assignment is given by $\tilde{Z}_i$ to distinguish it from
the factual $Z_i$. We can compute the difference in outcomes from this
counterfactual experiment, equal to $T(\tilde{\bm{Z}} \mid
\mathcal{F}) = 420/3 - 300/3 = 40$. Indeed, we could enumerate the
counterfactual results from all 20 equally possible experiments, shown
in Figure \ref{fig:hypohist}. The bars highlighted in red are
comprised of 4 counterfactual experiments with an average outcome
difference less than or equal to that observed in our actual
experiment. Therefore, under the null hypothesis, the one-sided
probability of observing a result more extreme than our observed
result is 4/20 or 20\%.

\begin{table}[ht]
  \caption{Imputed data from a counterfactual experiment under the exact null hypothesis}
  \begin{center}
    \begin{tabular}{c|c|ccc|ccc}
      \toprule
      $i$ & $\tilde{Z}_i$ & $D_i$ & $D_i(1)$ & $D_i(0)$ & $Y_i$ & $Y_i(1)$ & $Y_i(0)$ \\
      \midrule
      1 & \textcolor{red}{1} & 1 & 1 & \textcolor{red}{?} & 120 & 120 & \textcolor{red}{120}  \\
      2 & \textcolor{red}{0} & 1 & 1 & \textcolor{red}{?} & 120 & 120 & \textcolor{red}{120} \\
      3 & \textcolor{red}{0} & 0 & 0 & \textcolor{red}{?} & 75 & \textcolor{red}{75} & 75 \\
      4 & \textcolor{red}{1} & 0 & \textcolor{red}{?} & 1 & 165 & 165 & \textcolor{red}{165} \\
      5 & \textcolor{red}{1} & 0 & \textcolor{red}{?} & 0 & 135 & \textcolor{red}{135} & 135 \\
      6 & \textcolor{red}{0} & 0 & \textcolor{red}{?} & 0 & 105 & \textcolor{red}{105} & 105 \\
      \bottomrule
    \end{tabular}
    \label{tab:hyponull}
  \end{center}
\end{table}

\begin{figure}[ht]
  \begin{center}
    \caption{Histogram of outcome differences for the exact null hypothesis}
    \includegraphics[scale=0.8]{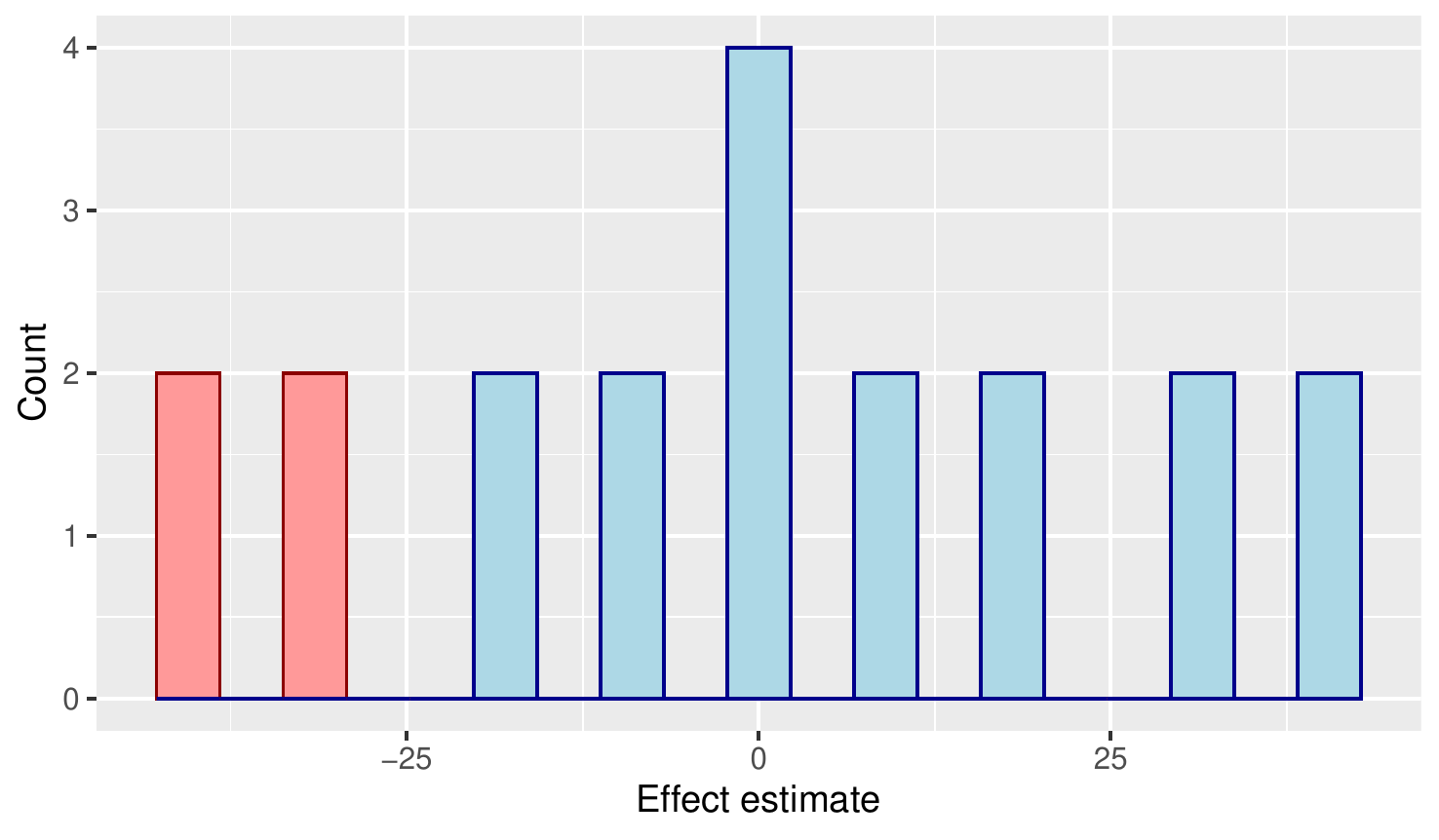}
    \label{fig:hypohist}
  \end{center}
\end{figure}

\section{Randomization distribution of offspring alleles}
\label{sec:randomization-distribution}

The distribution of offspring haplotypes is often approximated by a first order hidden Markov model (HMM) \citep{Haldane1919, Thompson2000, Bates2020a}. This model assumes ``no interference'', such that the location of crossover events are independent and the likelihood of an offspring inheriting a SNP from a given maternal or paternal haplotype depends only on the inheritance at adjacent loci. This induces a Poisson renewal process for the distribution of distances between crossovers, however, it should be noted that there is evidence of positive crossover interference in human meioses which results in a more even spread of crossovers than would be expected with random placement. Recent literature has therefore suggested that a Gamma renewal process may be a more appropriate model, although we do not provide this extension here \citep{Otto2019}.

The randomness in our randomization distribution arises from both the location of crossover events (i.e. the transition distribution) and the small probability of independent de novo mutations (i.e. the emission distribution). Without loss of generality, we describe the distribution of offspring alleles inherited from the mother $\bm{Z}^m$ given maternal haplotypes $\bm{M}^m$ and $\bm{M}^f$. Inheritance from the father is an independent instance of the same model. The transition distribution for the meiosis indicator at site $j$ is assumed to be Poisson with mean equal to the genetic distance in centimorgans $r_j$ between site $j-1$ and $j$:
\[\begin{array}{rl}%
 \mathbb{P}(U_j^m = u_{j-1}^m \mid U_{j-1}^m = u_{j-1}^m) &= \mathbb{P}(\text{even number of recombinations between $j-1$ and $j$}) \\[0.3em]
 &= \frac{1}{2}(1 + e^{-2r_j}); \\[1em]
 \mathbb{P}(U_j^m = U_{j^{\prime}}^m) &= \frac{1}{2}(1 + e^{-2(d_j + \ldots + d_{j^{\prime}})})
\end{array}\]%

where $u_{j-1}^m \in \{ m, f \}$ and $j < j^{\prime}$. Genetic distance is not proportional to physical distance on the chromosome due to the presence of recombination hotspots where crossover events are more likely to occur \citep{Belmont2005, Bherer2017}. As $r_j$ becomes large, the likelihood of an even number of recombinations approaches one half since genetically distant sites are transmitted almost independently.

The emission distribution is characterized by the probability of independent de novo single nucleotide mutations. A de novo mutation is said to occur when the base pair at some offspring SNP differs from the base pair they inherited from the parental haplotype. Within the context of the model, conditional on $U_j^m = u_j^m \in \{ m, f \}$, each $Z_j^m$ is sampled according to
\begin{equation}
    \mathbb{P}(Z_j^m = M_j^{(u_j^m)} \mid U_j^m = u_j^m) = 1 - \epsilon
    \label{eqn:emission}
\end{equation}
The probability of a de novo mutation $\epsilon$ is approximately $1 \cdot 10^{-8}$ in humans \citep{Acuna-Hidalgo2016}.

The graphical structure of the hidden Markov model is shown in \Cref{fig:hidden-markov-model}. This graph differs from the more general structure shown in \Cref{fig:meiosis_dag} in that each meiosis indicator $U_j^m$ depends only on the previous indicator $U_{j-1}^m$. \Cref{fig:fammr-dag-hmm} embeds the hidden Markov model within the complete causal model used throughout \Cref{sec:identification}.

Our primary use of the Markovian structure described above is to derive propensity scores for offspring haplotypes $Z_j^m \in \{0, 1\}$. In particular, our goal is to express the propensity score of some SNP $Z_j^m$ given the adjustment set $(\bm{M}_j^{mf}, \bm{V}_{\mathcal{B}}^m)$ of \Cref{thm:sufficient-adjustment-set}, where $\bm{V}_{\mathcal{B}}^m = (\bm{M}_{\mathcal{B}}^{mf}, \bm{Z}_{\mathcal{B}}^m)$ and $\mathcal{B} \subseteq \mathcal{J} \setminus \{j\}$. Throughout this section we will assume that $\mathcal{B} = \{1, \ldots, l\} \cup \{h, \ldots p\}$ for $l < j < h$. Suppressing conditioning on $\bm{M}_j^{mf}$ and $\bm{M}_{\mathcal{B}}^{mf}$ for ease of notation, the propensity score for $Z_j^m$ can be written as
 \begin{equation}
   \mathbb{P}(Z_j^m = 1 \mid \bm{Z}_{\mathcal{B}}^m = \bm{z}_{\mathcal{B}}^m) = \sum_{u \in \{ m, f \}} \mathbb{P}(Z_j^m = 1 \mid U_j^m = u) \mathbb{P}(U_j^m = u \mid \bm{Z}_{\mathcal{B}}^m = \bm{z}_{\mathcal{B}}^m).
   \label{eqn:samplehaplotype}
 \end{equation}
 It is therefore more convenient to consider the conditional probability of $U_j^m$. We state the following theorem:
 \begin{theorem} \label{thm:propensity-score}
   Using the conditional independence properties implied by \Cref{fig:meiosis_dag}, the conditional probability of $U_j^m = m$ can be factorized as
   \[\begin{array}{ll}
     &\mathbb{P}(U_j^m = m \mid \bm{Z}_{\mathcal{B}}^m = \bm{z}_{\mathcal{B}}^m) \\[0.5em] \propto& \displaystyle \biggl[ \sum_{u \in \{m,f\}} \beta_{h-1}^m(u) \mathbb{P}(U_{h-1}^m = u \mid U_j^m = m) \biggr] \biggl[ \sum_{u \in \{ m, f \}} \mathbb{P}(U_j^m = m \mid U_l^m = u) \, \alpha_l^m(u) \biggr]. \end{array}\]
   The forward weights are defined recursively as
   \[\begin{array}{ll}%
       \alpha_1^m(u_1^m) &= \begin{cases}
         \frac{1}{2}(1 - \epsilon) & \text{if $M_1^{u_1^m} = z_1^m$}  \\
         \frac{1}{2} \epsilon & \text{if $M_1^{u_1^m} \neq z_1^m$}
       \end{cases} \\[1.3em]
       \alpha_k^m(u_k^m) &= \displaystyle \sum_{u \in \{ m, f \}} \mathbb{P}(Z_k^m = z_k^m \mid U_k^m = u_k^m) \mathbb{P}(U_k^m = u_k^m \mid U_{k-1}^m = u) \alpha_{k-1}^m(u) , ~ k = 2, \ldots, p
     \end{array}\]%
   and the backward weights are defined recursively as
   \[\begin{array}{ll}%
       \beta_p^m(u_p^m) &= 1 \\[0.3em]
       \beta_k^m(u_k^m) &= \displaystyle \sum_{u \in \{ m, f \}} \beta_{k+1}^m(u) \mathbb{P}(U_{k+1}^m = u \mid U_k^m = u_k^m) \mathbb{P}(Z_{k+1}^m = z_{k+1}^m \mid U_{k+1}^m = u), ~ k = 1, \ldots, p-1,
     \end{array}\]%
   for $u_k^m \in \{ m, f \}$ and $j, k \in \mathcal{J}$.
 \end{theorem}

If we impose the simplifying assumption that $\epsilon = 0$, so that there is zero probability of de novo mutations, then the distribution of $U_j^m$ derived in Theorem \ref{thm:propensity-score} can be simplified further.

 \begin{corollary} \label{cor:heterozygous-hmm}
   Suppose the probability of a single nucleotide de novo mutation is $\epsilon = 0$ and suppose that the maternal haplotypes at $b_1, b_2 \in \mathcal{J}$ are heterozygous, where $b_1 < l < j < h < b_2$. That is, $M_{b_1}^m \neq M_{b_1}^f$ and $M_{b_2}^m \neq M_{b_2}^f$. Then the propensity score in \Cref{thm:propensity-score} can equivalently be written as
    \[\label{eqn:propensity-score-simple} \begin{array}{ll}
     &\mathbb{P}(U_j^m = m \mid \bm{Z}_{\mathcal{B}}^m = \bm{z}_{\mathcal{B}}^m) \\[0.5em] \propto& \displaystyle \biggl[ \sum_{u \in \{m,f\}} \tilde{\beta}_{h-1}^m(u) \mathbb{P}(U_{h-1}^m = u \mid U_j^m = m) \biggr] \biggl[ \sum_{u \in \{ m, f \}} \mathbb{P}(U_j^m = m \mid U_l^m = u) \, \tilde{\alpha}_l^m(u) \biggr]. \end{array}\]
    where
   \[\begin{array}{ll}%
       \tilde{\alpha}_{b_1+1}^m(u_{b_1+1}^m) = \mathbb{P}(Z_{b_1+1}^m = z_{b_1+1}^m \mid U_{b_1+1}^m = u_{b_1+1}^m) \mathbb{P}(U_{b_1+1}^m = u_{b_1+1}^m \mid U_{b_1}^m = u_{b_1}^m) \\[0.3em]
       \tilde{\alpha}_k^m(u_k^m) = \displaystyle \sum_{u \in \{ m, f \}} \mathbb{P}(Z_k^m = z_k^m \mid U_k^m = u_k^m) \mathbb{P}(U_k^m = u_k^m \mid U_{k-1}^m = u) \tilde{\alpha}_{k-1}^m(u) , ~ k = b_1+2, \ldots, p;
     \end{array}\]%
    and
  \[\begin{array}{ll}%
       \tilde{\beta}_{b_2-1}^m(u_{b_2-1}^m) = \mathbb{P}(U_{b_2}^m = u_{b_2}^m \mid U_{b_2-1}^m = u_{b_2-1}^m) \mathbb{P}(Z_{b_2}^m = z_{b_2}^m \mid U_{b_2}^m = u_{b_2}^m) \\[0.3em]
       \beta_k^m(u_k^m) = \displaystyle \sum_{u \in \{ m, f \}} \tilde{\beta}_{k+1}^m(u) \mathbb{P}(U_{k+1}^m = u \mid U_k^m = u_k^m) \mathbb{P}(Z_{k+1}^m = z_{k+1}^m \mid U_{k+1}^m = u), ~ k = 1, \ldots, b_2-2.
     \end{array}\]%
 \end{corollary}

We will occasionally have multiple instruments lying in the same window. We will then need to compute a multivariate propensity score. We state the following corollary without proof because it follows almost immediately from \Cref{thm:propensity-score}.
 \begin{corollary} \label{cor:multiple-instruments}
Suppose we have a collection of instruments $\mathcal{J} = \{j_1, j_2, \ldots, j_r\}$ such that $l < j_1 < j_2 < \ldots < j_r < h$. Then the propensity score can be written as
\begin{align}
& \mathbb{P}(U_{j_1}^m = u_{j-1}^m, U_{j_2}^m = u_{j_2}^m, \ldots, U_{j_r}^m = u_{j_r}^m \mid \bm{Z}_{\mathcal{B}}^m = \bm{z}_{\mathcal{B}}^m) \\ = ~ &  \mathbb{P}(U_{j_1}^m = u_{j_1}^m \mid \bm{Z}_{\mathcal{B}}^m = \bm{z}_{\mathcal{B}}^m) \prod_{k=2}^r \mathbb{P}(U_{j_k}^m = u_{j_k}^m \mid U_{j_{k-1}}^m = u_{j_{k-1}}^m, \bm{Z}_{\mathcal{B}}^m = \bm{z}_{\mathcal{B}}^m)
\end{align}
The first propensity score $\mathbb{P}(U_{j_1}^m = m \mid \bm{Z}_{\mathcal{B}}^m = \bm{z}_{\mathcal{B}}^m)$ takes the form in \Cref{thm:propensity-score} and
\begin{align}
& \mathbb{P}(U_{j_k}^m = m \mid U_{j_{k-1}}^m = u_{j_{k-1}}^m, \bm{Z}_{\mathcal{B}}^m = \bm{z}_{\mathcal{B}}^m) \\
\propto ~ & \mathbb{P}(U_{j_k}^m = m \mid U_{j_{k-1}}^m = u_{j_{k-1}}^m) \biggl[ \sum_{u \in \{m,f\}} \beta_{h-1}^m(u) \mathbb{P}(U_{h-1}^m = u \mid U_{j_k}^m = m) \biggr]
\end{align}
where $\beta_{h-1}^m(u)$ is the backward weight defined in \Cref{thm:propensity-score}.
 \end{corollary}

 \begin{figure}[ht]
   \begin{center}
     \includegraphics[width=0.65\textwidth]{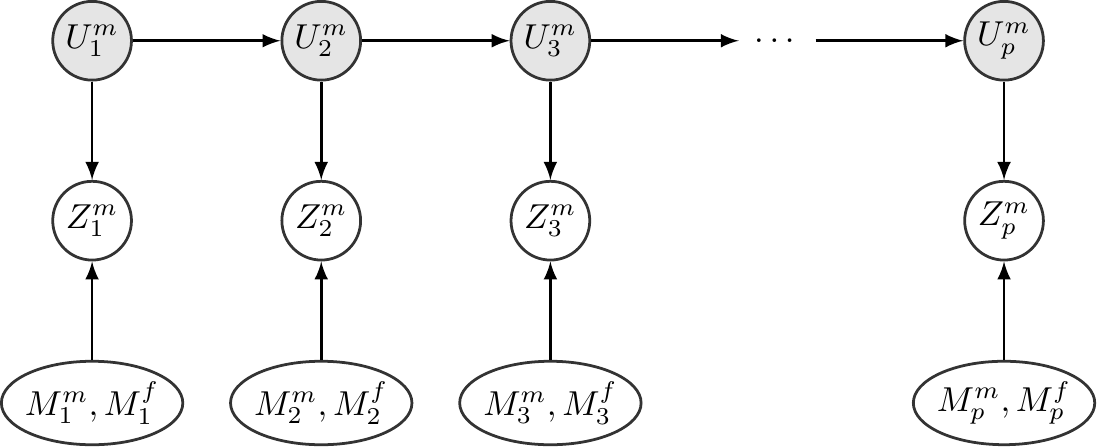}
     \caption{Graphical representation of Haldane's hidden Markov model.}
     \label{fig:hidden-markov-model}
   \end{center}
 \end{figure}

 \begin{figure}[ht]
   \begin{center}
     \includegraphics[width=0.8\textwidth]{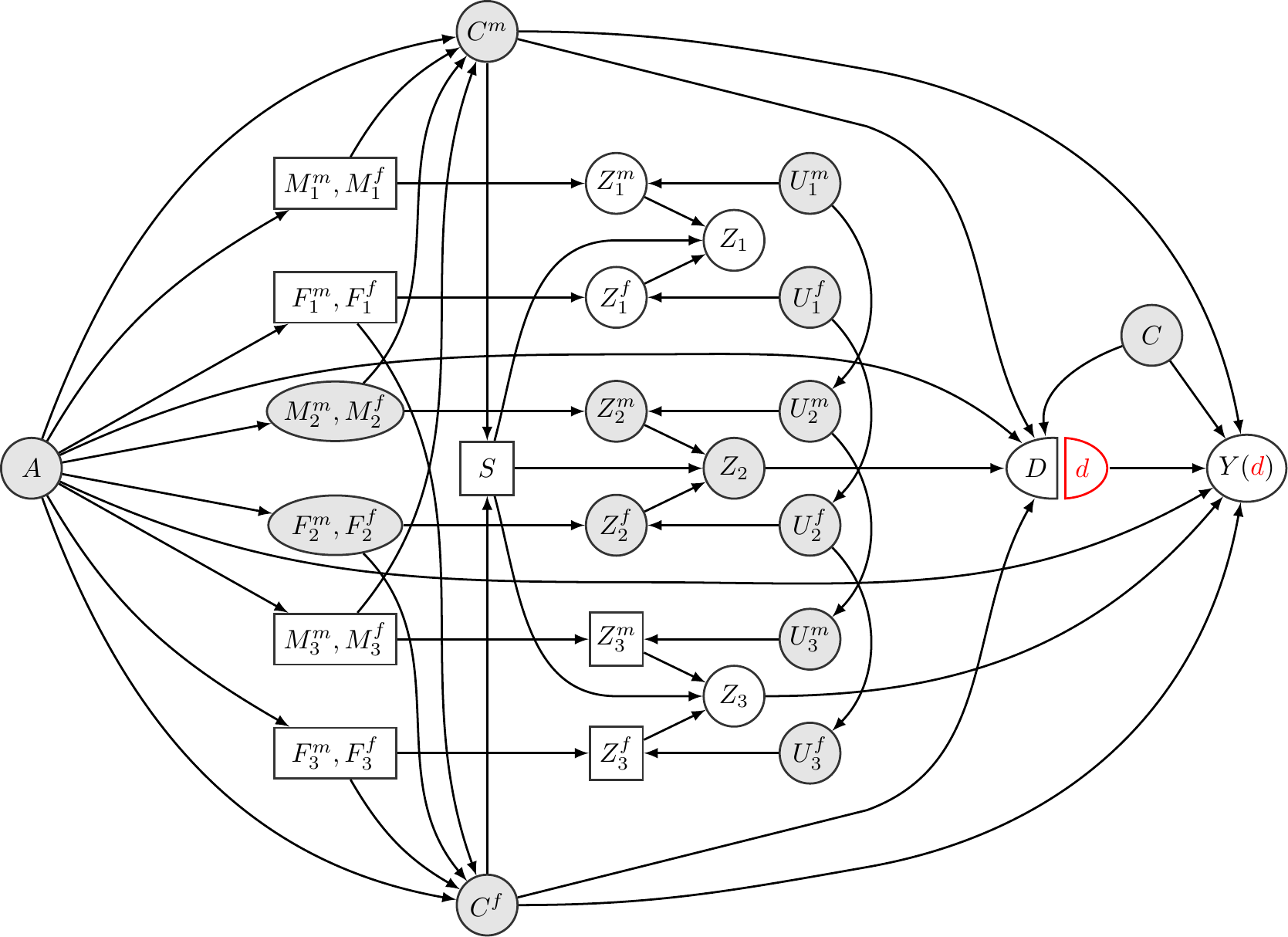}
     \caption{Haldane's hidden Markov model embedded in the working
       example in \Cref{fig:fammr_dag}.}
     \label{fig:fammr-dag-hmm}
   \end{center}
 \end{figure}

 \section{Technical proofs}
 \label{sec:technical-proofs}

 \subsection*{\Cref{prop:offspring-genotype-independence}}
 \label{sec:prop:-genotype-indep}

 \begin{proof}
   From
   \Cref{assum:de-novo-distribution} we know that, conditional on
   $(M_j^m, M_j^f, F_j^m, F_j^f)$, $Z_j^m$ and $Z_j^f$ only depend on
   $\bm{U}^m$ and $\bm{U}^f$, respectively, and exogenous mutation
   events. By \eqref{eqn:offspring-genotype-undefined}, $Z_j = Z_j^m +
   Z_j^f$ given that $S = 1$ (fertilization occurs).  Finally, by
   \Cref{assum:meiosis-indicator-independence}, the
   meiosis indicators $\bm U^m$ and $\bm U^f$ are independent of all confounders $(A, C^m,
   C^f, C)$. Therefore, the conditional independence statement
   immediately follows.
 \end{proof}




 \subsection*{\Cref{thm:propensity-score}}
 \label{sec:thm-propensity-score}

 \begin{proof}
   The conditional probability of $U_j^m$ can be factorized as
   \begin{align*}
       & \mathbb{P}(U_j^m = m \mid \bm{Z}_{\mathcal{B}}^m = \bm{z}_{\mathcal{B}}^m) \\[0.3em]
       \propto & ~ \mathbb{P}(U_j^m = m, \bm{Z}_{\mathcal{B}}^m = \bm{z}_{\mathcal{B}}^m) \\[0.3em]
       =& ~ \mathbb{P}(\bm{Z}_{h:p}^m = \bm{z}_{h:p}^m \mid U_j^m = m) \mathbb{P}(U_j^m = m, \bm{Z}_{1:l}^m = \bm{z}_{1:l}^m) \\[0.3em]
       =& ~ \displaystyle \biggl[ \sum_{u \in \{m,f\}}
          \mathbb{P}(\bm{Z}_{(j+1):p}^m = \bm{z}_{(j+1):p}^m,
          U_{h-1}^m = u \mid U_j^m = m) \biggr] \\
     & \quad \displaystyle \biggl[ \sum_{u \in \{ m, f \}} \mathbb{P}(U_j^m = m, U_l^m = u, \bm{Z}_{1:l}^m = \bm{z}_{1:l}^m) \biggr] \\[0.3em]
      =& ~ \displaystyle \biggl[ \sum_{u \in \{m,f\}} \mathbb{P}(\bm{Z}_{(j+1):p}^m = \bm{z}_{(j+1):p}^m \mid U_{h-1}^m = u) \mathbb{P}(U_{h-1}^m = u \mid U_j^m = m) \biggr] \\
       & \quad \displaystyle \biggl[ \sum_{u \in \{ m, f \}} \mathbb{P}(U_j^m = m \mid U_l^m = u) \mathbb{P}(U_l^m = u, \bm{Z}_{1:l}^m = \bm{z}_{1:l}^m) \biggr] \\[0.3em]
      =& ~ \displaystyle \biggl[ \sum_{u \in \{m,f\}} \beta_{h-1}^m(u) \mathbb{P}(U_{h-1}^m = u \mid U_j^m = m) \biggr] \biggl[ \sum_{u \in \{ m, f \}} \mathbb{P}(U_j^m = m \mid U_l^m = u) \alpha_l^m(u) \biggr].
    \end{align*}

   The forward weight $\alpha_1^m(u_1^m)$ for some $u_1^m \in \{ m, f \}$ can be derived as
   \[\begin{array}{ll}%
       \alpha_1^m(u_1^m) &= \mathbb{P}(U_1^m = u_1^m, Z_1^m = z_1^m) \\[0.3em]

       &= \mathbb{P}(Z_1^m = z_1^m \mid U_1^m = u_1^m) \mathbb{P}(U_1^m = u_1^m) \\[0.3em]

       &= \frac{1}{2} \mathbb{P}(Z_1^m = z_1^m \mid U_1^m = u_1^m)
     \end{array}\]%
   where the emission probability is known. A recursive expression for the forward weight $\alpha_j^m(u_j^m)$ for $j = 2, \ldots, p$ can be derived as
   \[\begin{array}{ll}%
       \alpha_j^m(u_j^m) &= \mathbb{P}(U_j^m = u_j^m, \bm{Z}_{1:j}^m = \bm{z}_{1:j}^m) \\[0.3em]

       &= \displaystyle \sum_{u \in \{ m, f \}} \mathbb{P}(U_j^m = u_j^m, U_{j-1}^m = u_{j-1}^m, \bm{Z}_{1:j}^m = \bm{z}_{1:j}^m) \\[0.3em]

       &= \displaystyle \sum_{u \in \{ m, f \}} \mathbb{P}(Z_j^m = z_j^m \mid U_j^m = u_j^m) \mathbb{P}(U_j^m = u_j^m \mid U_{j-1}^m = u) \mathbb{P}(U_{j-1}^m = u, \bm{Z}_{1:(j-1)}^m = \bm{z}_{1:(j-1)}^m) \\[0.3em]

       &= \displaystyle \sum_{u \in \{ m, f \}} \mathbb{P}(Z_j^m = z_j^m \mid U_j^m = u_j^m) \mathbb{P}(U_j^m = u_j^m \mid U_{j-1}^m = u) \, \alpha_{j-1}^m(u).

     \end{array}\]%

   The backward weight $\beta_j^m(u_j^m)$ for some $u_j^m \in \{ m, f \}$ and $j = 1, \ldots, p-1$ can be derived as
   \[\begin{array}{ll}%
       &\beta_j^m(u_j^m) = \mathbb{P}(\bm{Z}_{(j+1):p}^m = \bm{z}_{(j+1):p}^m \mid U_j^m = u_j^m) \\[0.3em]

       &= \displaystyle \sum_{u \in \{ m, f \}} \mathbb{P}(\bm{Z}_{(j+1):p}^m = \bm{z}_{(j+1):p}^m, U_{j+1}^m = u \mid U_j^m = u_j^m) \\[0.3em]

       &= \displaystyle \sum_{u \in \{ m, f \}} \mathbb{P}(\bm{Z}_{(j+2):p}^m = \bm{z}_{(j+2):p}^m \mid U_{j+1}^m = u) \mathbb{P}(Z_{j+1}^m = z_{j+1}^m \mid U_{j+1}^m = u) \mathbb{P}(U_{j+1}^m = u \mid U_j^m = u_j^m).
     \end{array}\]%
   Writing the probability of $U_p^m$ shows that $\beta_p^m(u) = 1$ for all $u \in \{ m, f \}$.
 \end{proof}

 \subsection*{\Cref{cor:heterozygous-hmm}}
 \label{sec:crefc}

 \begin{proof}
   The proof involves some manipulation of conditional independencies. We simplify the probability with respect to $b_1$ and omit simplification with respect to $b_2$ for brevity. As with the proof of Theorem \ref{thm:propensity-score} we begin by factorising the conditional probability of $U_j^m$.
   \begin{equation}%
     \mathbb{P}(U_j^m = m \mid \bm{Z}_{\mathcal{B}}^m = \bm{z}_{\mathcal{B}}^m) = \frac{ \mathbb{P}(\bm{Z}_{h:p}^m = \bm{z}_{h:p}^m \mid U_j^m = m) \mathbb{P}(U_j^m = m, \bm{Z}_{1:l}^m = \bm{z}_{1:l}^m)}{\mathbb{P}(\bm{Z}_{\mathcal{B}}^m = \bm{z}_{\mathcal{B}}^m)}.
     \label{eqn:propensity-score}
   \end{equation}%

   Since $b_1 < j$ we are concerned with simplifying the second probability in the numerator of equation \eqref{eqn:propensity-score}.
   \[\begin{array}{ll}%
       & \mathbb{P}(U_j^m = m, \bm{Z}_{1:l}^m = \bm{z}_{1:l}^m) \\[0.3em]

       =& \displaystyle \sum_{u \in \{ m, f \}} \mathbb{P}(U_j^m = m, U_{b_1}^m = u, \bm{Z}_{1:l}^m = \bm{z}_{1:l}^m) \\[0.3em]

       =& \displaystyle \sum_{u \in \{ m, f \}} \mathbb{P}(U_j^m = m, \bm{Z}_{(b_1+1):l}^m = \bm{z}_{(b_1+1):l}^m \mid U_{b_1}^m = u) \mathbb{P}(U_{b_1}^m = u, \bm{Z}_{1:b_1}^m = \bm{z}_{1:b_1}^m) \\[0.3em]

       =& \mathbb{P}(U_j^m = m, \bm{Z}_{(b_1+1):l}^m = \bm{z}_{(b_1+1):l}^m \mid U_{b_1}^m = m) \mathbb{P}(U_{b_1}^m = m, \bm{Z}_{1:b_1}^m = \bm{z}_{1:b_1}^m) \\[0.3em]

       =& \displaystyle \mathbb{P}(U_{b_1}^m = m, \bm{Z}_{1:b_1}^m = \bm{z}_{1:b_1}^m) \sum_{u \in \{ m, f \}} \mathbb{P}(U_j^m = m \mid U_{j-1}^m = u) \mathbb{P}(U_{j-1}^m = u, \bm{Z}_{(b_1+1):(j-1)}^m = \\ & \qquad \bm{z}_{(b_1+1):(j-1)}^m \mid U_{b_1}^m = m) \\[0.3em]

       =& \displaystyle \mathbb{P}(U_{b_1}^m = m, \bm{Z}_{1:b_1}^m = \bm{z}_{1:b_1}^m) \sum_{u \in \{ m, f \}} \mathbb{P}(U_j^m = m \mid U_{j-1}^m = u) \tilde{\alpha}_{j-1}^m(u).
     \end{array}\]%
   where
   \[\begin{array}{rl}%
       \tilde{\alpha}_{b_1+1}^m(u_{b_1+1}^m) &= \mathbb{P}(Z_{b_1+1}^m = z_{b_1+1}^m \mid U_{b_1+1}^m = u_{b_1+1}^m) \mathbb{P}(U_{b_1+1}^m = u_{b_1+1}^m \mid U_{b_1}^m = m) \\[0.3em]
       \tilde{\alpha}_k^m(u_k^m) &= \displaystyle \sum_{u \in \{ m, f \}} \mathbb{P}(Z_k^m = z_k^m \mid U_k^m = u_k^m) \mathbb{P}(U_k^m = u_k^m \mid U_{k-1}^m = u) \tilde{\alpha}_{k-1}^m(u), \\
       & \quad \text{for } k = b_1+2, \ldots, j-1.
     \end{array}\]%

   We now factorize the denominator of equation \eqref{eqn:propensity-score}.
   \[
     \mathbb{P}(\bm{Z}_{\mathcal{B}}^m = \bm{z}_{\mathcal{B}}^m) = \mathbb{P}(\bm{Z}_{(b_1+1):l} = \bm{z}_{(b_1+1):l}, \bm{Z}_{h:p} = \bm{z}_{h:p} \mid U_{b_1}^m = m) \mathbb{P}(U_{b_1}^m = m, \bm{Z}_{1:b_1} = \bm{z}_{1:b_1}).
   \]%
   Substituting these simplified expressions back in equation \eqref{eqn:propensity-score} we obtain
   \begin{equation}%
     \begin{array}{ll}%
       & \mathbb{P}(U_j^m = m \mid \bm{Z}_{\mathcal{B}}^m = \bm{z}_{\mathcal{B}}^m) \\[0.5em]

       =& \displaystyle \frac{ \mathbb{P}(\bm{Z}_{h:p}^m = \bm{z}_{h:p}^m \mid U_j^m = m) \mathbb{P}(U_{b_1}^m = m, \bm{Z}_{1:b_1}^m = \bm{z}_{1:b_1}^m) \sum_{u \in \{ m, f \}} \mathbb{P}(U_j^m = m \mid U_{j-1}^m = u) \tilde{\alpha}_{j-1}^m(u)}{\mathbb{P}(\bm{Z}_{(b_1+1):l} = \bm{z}_{(b_1+1):l}, \bm{Z}_{h:p} = \bm{z}_{h:p} \mid U_k^m = m) \mathbb{P}(U_{b_1}^m = m, \bm{Z}_{1:b_1} = \bm{z}_{1:b_1})} \\[1.5em]

       =& \displaystyle \frac{ \mathbb{P}(\bm{Z}_{h:p}^m = \bm{z}_{h:p}^m \mid U_j^m = m) \sum_{u \in \{ m, f \}} \mathbb{P}(U_j^m = m \mid U_{j-1}^m = u) \tilde{\alpha}_{j-1}^m(u)}{\mathbb{P}(\bm{Z}_{(b_1+1):l} = \bm{z}_{(b_1+1):l}, \bm{Z}_{h:p} = \bm{z}_{h:p} \mid U_{b_1}^m = m)}.
     \end{array}%
   \end{equation}%
   which does not depend on $\bm{Z}_{1:k}^m$.
 \end{proof}

 \newpage

 \section{Simulation} \label{sec:simulation-description}

 \subsection{Further description of the simulation setting}

 \begin{longtable}{lp{7cm}p{3cm}}
   \caption{Description of the simulation variables and
   parameters} \label{tab:simulation-description} \\
   \toprule
   Variable & Description of how the variable is constructed & Parameters \\
   \midrule
   $\bm{M}_i^m, \bm{M}_i^f, \bm{F}_i^m, \bm{F}_i^f$ &  The parental haplotypes are constructed to allow linkage disequilibrium in nearby SNPs. For each parental haplotype we first sample from a $p$-variate normal such that $X_{ij} \sim \mathcal{N}(0,1)$ and $Cov(X_{ij},X_{ik}) = \rho^{|j - k|}$, $0 < \rho < 1$, $j,k \in \mathcal{J}$. Thresholds $V_{ij} \sim Unif(a,b)$ are sampled and the haplotypes are defined as $M_{ij}^m = I\{X_{ij} > V_{ij}\}$ where $I\{\cdot\}$ is the indicator function (and similarly for the other haplotypes). & \[\begin{array}{l} \rho=0.75\\ a = \Phi^{-1}(0.6) \\ b = \Phi^{-1}(0.95) \end{array}\]
   where $\Phi^{-1}(\cdot)$ is the inverse normal CDF. \\ \midrule
   $C_i^m, C_i^f$ & We first define a variable \[ \hat{\mu}_i^m = \frac{1}{p} \sum_{j=1}^p (M_{ij}^m + M_{ij}^f). \] It follows from our construction of the parental haplotypes that \[ \mu^m = E[\hat{\mu}^m] = 2\biggl(1 - \frac{1}{b-a} \int_a^b \Phi(x) dx \biggr). \] where $\Phi(\cdot)$ is the normal CDF. For each individual $i$ we sample the parental confounder such that \[ C_i^m \sim \mathcal{N}(\hat{\mu}_i^m - \mu^m, 1).\] We follow an identical procedure for $C_i^f$. & N/A \\[0.3em] \midrule
   $C_i$ & We construct the offspring confounder as \[ C_i \sim \mathcal{N}(0,1). \] & N/A \\ \midrule
   $\bm{Z}^m_i, \bm{Z}^f_i$ & We sample the offspring haplotypes using Algorithm 1 in \citet{Bates2020a}. This algorithm unconditionally samples a full haplotype $\bm{Z}_i^m$ or  $\bm{Z}_i^f$ according to the hidden Markov model described in \Cref{sec:randomization-distribution}. It depends on the genetic distances $\bm{r}$ and de novo mutation rate $\epsilon$. We  sample $r_j \sim Unif(c,d)$ and set $r_k = \infty$ for $k =  37,62,86,112$ so that the instruments are unconditionally independent. From these haplotypes we choose a subset $\mathcal{J}_g \subset \mathcal{J}$ to be instruments. & \[\begin{array}{l} \epsilon=10^{-8} \\ c = 0 \\ d = 0.75 \\ \mathcal{J}_g = \{25,50,75,\\ 100, 125\} \end{array}\] \\ \midrule
   $D_i$ & The exposure follows a linear structural equation model \[ D_i = \gamma^{\intercal} \bm{Z}_i + \theta^m C_i^m + \theta^f C_i^f + \theta^c C_i + \nu_i \] where $\nu_i \sim \mathcal{N}(0,0.7)$. We choose $\gamma$ so that it is zero everywhere except for $\gamma_{24}$, $\gamma_{49}$, $\gamma_{74}$, $\gamma_{99}$ and $\gamma_{124}$ which represent causal variants. The parameters are chosen so that $Var(D_i) = 1$. & \[\begin{array}{l}  \theta^m = \theta^f = \sqrt{0.3} \\ \theta^c = \sqrt{0.75} \\ \gamma_j = \sqrt{0.1} \end{array}\] for $j = 24$, $49$, $74$, $99$, $124$. \\
   \midrule
   $Y_i$ & The outcome follows a linear structural equation model \[ Y_i = \beta D_i + \delta^{\intercal} \bm{Z}_i + \phi^m C_i^m + \phi^f C_i^f + \phi^c C_i + \upsilon_i \] where $\upsilon_i \sim \mathcal{N}(0,0.7)$. We choose $\delta$ so that it is zero everywhere except for $\delta_{23}$, $\delta_{27}$, $\delta_{48}$, $\delta_{52}$, $\delta_{73}$, $\delta_{77}$, $\delta_{98}$, $\delta_{102}$, $\delta_{123}$ and $\delta_{127}$ which represent pleiotropic variants. The parameters are chosen so that $Var(Y_i) = 1$. &  \[\begin{array}{l} \beta = 0 \\ \phi^m = \phi^f = \sqrt{0.3} \\ \phi^c = \sqrt{0.75} \\ \delta_j = \sqrt{0.05} \end{array}\] for $j = 23$, $27$, $48$, $52$, $73$, $77$, $98$, $102$, $123$, $127$. \\
   \hline
 \end{longtable}

 \Cref{thm:sufficient-adjustment-set} implies that a sufficient adjustment set for this simulation is
 \begin{equation} \label{eqn:sim-adjustment-set}
   \bm (\bm{M}^{mf}_{\mathcal{B}_g}, \bm{F}^{mf}_{\mathcal{B}_g}, \bm{Z}_{\mathcal{B}})
 \end{equation}
 where
 \[
   \mathcal{B} = \mathcal{J} \setminus \{24,25,26,49,50,51,99,74,75,76,99,100,101,124,125,126\}
 \]
 and
 \[
   \mathcal{B}_g = \mathcal{B} \cup \{25,50,75,100,125\}.
 \]

\subsection{An illustration of the simulation setup}
\label{sec:an-illustr-simul}

To make our setup
more tangible, \Cref{tab:counterfactual-simulation-data} shows the first 6
lines of observed and counterfactual data (in red) from the simulation for one of the instruments
and corresponding parental haplotypes. We can see that individual $4$
will provide almost no information for a test of the null hypothesis;
both of her parents are homozygous so there is no randomization in her
genotype outside of de novo mutations. Conversely, both of individual
$1$'s parents are heterozygous so she could receive both major
alleles, both minor alleles or one of each.

Suppose we wish to test the null hypothesis $H_0: \beta =
-0.3$. Column $\tilde{Z}_i$ in
\Cref{tab:counterfactual-simulation-data} shows a counterfactual draw
of each individual's instrument conditional on the adjustment set
given in \Cref{eqn:sim-adjustment-set} in
\Cref{sec:simulation-description}, along with the adjusted outcome
$Q_i(-0.3)$. Note that $\tilde{Z}_i$ is independent of $Q_i(-0.3)$ by
construction, so the null hypothesis is necessarily satisfied for this
counterfactual. As expected individual $4$ has the same genotype in
this counterfactual, however, individual $1$ inherits both minor
alleles in this case. \Cref{fig:simulation-null-distribution} plots a
distribution of 10,000 counterfactual test statistics drawn under the
null hypothesis. The test statistic is the F-statistic from a
regression of the adjusted outcome on the instruments. The bars
highlighted in red are larger than the observed test statistic, such
that the almost exact p-value is around 0.13.

\begin{table*}[ht]
   \caption{First 6 rows of observed data from the simulation}
   \begin{center}
     \begin{tabular}{cccccccccc}
       \toprule
       $i$ & $Z_i$ & \textbf{$\tilde{Z}_i$} & $M_i^m$ & $M_i^f$ & $F_i^m$ & $F_i^f$ & $D_i$ & $Y_i$ & $Q_i(-0.3)$ \\
       \midrule
       1 & 1 & \textbf{2} & 1 & 0 & 1 & 0 & 1.11 & 0.73 & 1.06 \\
       2 & 0 & \textbf{1} & 1 & 0 & 0 & 0 & 0.83 & -0.52 & 0.77 \\
       3 & 1 & \textbf{1} & 1 & 0 & 0 & 0 & 0.94 & 0.31 & 0.59 \\
       4 & 0 & \textbf{0} & 0 & 0 & 0 & 0 & 1.43 & 3.30 & 3.73 \\
       5 & 0 & \textbf{0} & 0 & 0 & 0 & 0 & 0.15 & 1.34 & 1.38 \\
       6 & 0 & \textbf{0} & 0 & 0 & 0 & 0 & -0.14 & 1.60 & 1.56 \\
       \bottomrule
     \end{tabular}
     \label{tab:counterfactual-simulation-data}
   \end{center}
 \end{table*}

 \begin{figure}[th]
   \begin{center}
     \includegraphics[width = 0.6\textwidth]{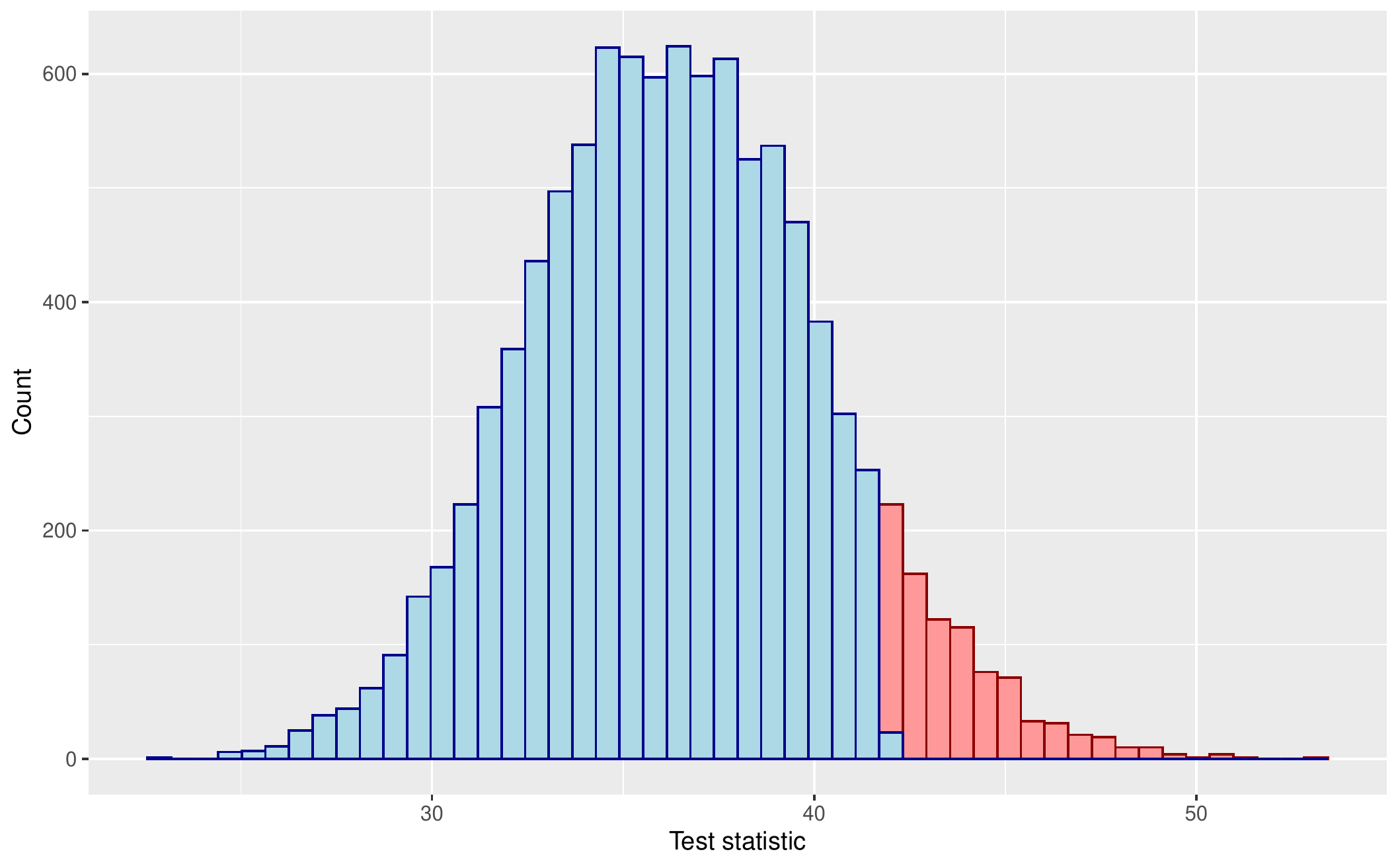}
     \caption{Histogram of 10,000 test statistics under the null
       hypothesis $H_0: \beta = -0.3$}
     \label{fig:simulation-null-distribution}
   \end{center}
 \end{figure}


\end{document}